\documentclass[12pt]{article}
\usepackage[colorlinks=true,linkcolor=blue,citecolor=Red]{hyperref}
\usepackage{breakcites}

\usepackage{mathtools}

\usepackage{subfigure}

\usepackage{algorithm}
\usepackage{algpseudocode}

\usepackage{bbm}
\usepackage[framed,numbered,autolinebreaks,useliterate]{mcode}

\usepackage{latexsym}
\usepackage{graphics}
\usepackage{amsmath}
\usepackage[algo2e]{algorithm2e}
\usepackage{xspace}
\usepackage{amssymb}
\usepackage{psfrag}
\usepackage{epsfig}
\usepackage{amsthm}
\usepackage{url}
\usepackage{pst-all}
\usepackage{bbm}

\usepackage[title]{appendix}

\usepackage{color}

\definecolor{Red}{rgb}{1,0,0}
\definecolor{Blue}{rgb}{0,0,1}
\definecolor{Olive}{rgb}{0.41,0.55,0.13}
\definecolor{Yarok}{rgb}{0,0.5,0}
\definecolor{Green}{rgb}{0,1,0}
\definecolor{MGreen}{rgb}{0,0.8,0}
\definecolor{DGreen}{rgb}{0,0.55,0}
\definecolor{Yellow}{rgb}{1,1,0}
\definecolor{Cyan}{rgb}{0,1,1}
\definecolor{Magenta}{rgb}{1,0,1}
\definecolor{Orange}{rgb}{1,.5,0}
\definecolor{Violet}{rgb}{.5,0,.5}
\definecolor{Purple}{rgb}{.75,0,.25}
\definecolor{Brown}{rgb}{.75,.5,.25}
\definecolor{Grey}{rgb}{.5,.5,.5}

\newcommand{\ind}{\mathbbm{1}}

\newcommand{\R}{\mathbb{R}}
\newcommand{\N}{\mathbb{N}}

\newcommand{\SBP}{\texttt{SBP }}

\newcommand{\ip}[2]{\langle{#1},{#2}\rangle} 
\newcommand{\KRA}{\mathcal{A}_{\rm KR}}

\newcommand{\ones}{\mathbf{e}}

\renewcommand{\R}{\mathbb{R}}

\newcommand{\distr}{\stackrel{d}{=}}
\newcommand{\A}{\mathcal{A}}
\newcommand{\bincube}{\mathcal{B}_n}

\newcommand{\M}{\mathcal{M}}
\newcommand{\cN}{{\bf \mathcal{N}}}
\newcommand{\Overlap}[2]{\mathcal{O}\left(#1,#2\right)}

\newcommand{\ignore}[1]{\relax}

\newlength\myindent
\newtheorem{theorem}{Theorem}[section]

\newtheorem{lemma}[theorem]{Lemma}
\newtheorem{conjecture}[theorem]{Conjecture}

\newtheorem{proposition}[theorem]{Proposition}
\newtheorem{coro}[theorem]{Corollary}

\newtheorem{definition}[theorem]{Definition}

\renewcommand{\ip}[2]{\left\langle#1,#2\right\rangle}


\newcounter{parentnumber}

\makeatletter
\def\BState{\State\hskip-\ALG@thistlm}
\makeatother

\definecolor{Red}{rgb}{1,0,0}
\definecolor{Blue}{rgb}{0,0,1}
\definecolor{Olive}{rgb}{0.41,0.55,0.13}
\definecolor{Green}{rgb}{0,1,0}
\definecolor{MGreen}{rgb}{0,0.8,0}
\definecolor{DGreen}{rgb}{0,0.55,0}
\definecolor{Yellow}{rgb}{1,1,0}
\definecolor{Cyan}{rgb}{0,1,1}
\definecolor{Magenta}{rgb}{1,0,1}
\definecolor{Orange}{rgb}{1,.5,0}
\definecolor{Violet}{rgb}{.5,0,.5}
\definecolor{Purple}{rgb}{.75,0,.25}
\definecolor{Brown}{rgb}{.75,.5,.25}
\definecolor{Grey}{rgb}{.5,.5,.5}
\definecolor{Pink}{rgb}{1,0,1}
\definecolor{DBrown}{rgb}{.5,.34,.16}
\definecolor{Black}{rgb}{0,0,0}

\usepackage{float}

\usepackage{float}
\title{Algorithms and Barriers in the Symmetric Binary Perceptron Model}
\author{{\sf David Gamarnik}\thanks{Sloan School of Management, Massachusetts Institute of Technology; e-mail: {\tt gamarnik@mit.edu}.} 
\and
{\sf Eren C. K{\i}z{\i}lda\u{g}}\thanks{Laboratory for Information and Decision Systems (LIDS), Massachusetts Institute of Technology; e-mail: {\tt kizildag@mit.edu}.}
\and
{\sf Will Perkins}\thanks{Department of Mathematics, Statistics, and Computer Science, University of Illinois at Chicago; e-mail: {\tt math@willperkins.org}.}
\and 
{\sf Changji Xu}\thanks{Center of Mathematical Sciences and Applications, Harvard University; e-mail: {\tt cxu@cmsa.fas.harvard.edu}.}
}
\begin{document}
\maketitle
\begin{abstract}
The binary (or Ising) perceptron is a toy model of a single-layer neural network and can be viewed as a random constraint satisfaction problem with a high degree of connectivity.  The model and its symmetric variant, the symmetric binary perceptron (\texttt{SBP}), have been studied widely in statistical physics, mathematics, and machine learning. 

The \texttt{SBP} exhibits a dramatic \emph{statistical-to-computational gap}: the densities at which known efficient algorithms find solutions are far below the  threshold for the existence of solutions.
Furthermore, the \texttt{SBP} exhibits a striking structural property: at all positive constraint densities almost all of its solutions are `totally frozen' singletons separated by large Hamming distance~\cite{perkins2021frozen,abbe2021proof}. This suggests that finding a solution to the \texttt{SBP} may be computationally intractable. At the same time, however, the \texttt{SBP} does admit polynomial-time search algorithms at low enough densities. A conjectural explanation for this conundrum was put forth in~\cite{baldassi2020clustering}:  efficient algorithms  succeed in the face of freezing by finding exponentially rare clusters of large size. However, it was discovered recently that such rare large clusters exist at all subcritical densities, even at those well above the limits of known efficient algorithms~\cite{abbe2021binary}. Thus the driver of the statistical-to-computational gap exhibited by this model remains a mystery.

In this paper, we conduct a different landscape analysis to explain the algorithmic tractability of this problem. We show that at high enough densities the \texttt{SBP} exhibits the multi Overlap Gap Property ($m-$OGP), an intricate geometrical property known to be a rigorous barrier for large classes of algorithms. Our analysis shows that the $m-$OGP threshold (a) is well below the satisfiability threshold; and (b) matches the best known algorithmic threshold up to logarithmic factors as $m\to\infty$.  We then prove that the $m-$OGP rules out the class of stable algorithms for the \texttt{SBP} above this threshold. We conjecture that the $m \to \infty$ limit of the $m$-OGP threshold marks the  algorithmic threshold for the problem. Furthermore, we investigate the stability of known efficient algorithms for perceptron models and show that the Kim-Roche algorithm~\cite{kim1998covering}, devised for the \emph{asymmetric binary perceptron}, is stable in the sense we consider.
\end{abstract}
\newpage
\tableofcontents
\newpage
\section{Introduction}
In this paper, we study the \emph{perceptron model}. Proposed initially in the 1960's~\cite{joseph1960number,winder1961single,wendel1962problem,cover1965geometrical}, this is a toy model of one-layer neural network storing random patterns as well as a very natural model in high-dimensional probability. Let $X_i\in\R^n$, $1\le i\le M$, be i.i.d.\;random patterns to be stored. \textit{Storage} of these patterns is achieved if one finds a vector of synaptic weights $\sigma \in \R^n$ consistent with all $X_i$: that is, $\ip{X_i}{\sigma}\ge 0$ for $1\le i\le M$.  There are two main variants of the perceptron: when the vector $\sigma$ lies on on sphere in $\R^n$ (the spherical perceptron) and when $\sigma\in\bincube\triangleq \{-1,1\}^n$ (the binary or Ising perceptron).  For more on the spherical perceptron see~\cite{gardner1988space,shcherbina2003rigorous,stojnic2013another,talagrand2011mean,alaoui2020algorithmic}; in this paper we will focus only on the binary perceptron. 

A key quantity associated to the perceptron is the \emph{storage capacity}: the maximum number $M^*$ of such patterns for which there exists a vector of weights $\sigma\in\bincube$ that is consistent with all $X_i$, $1\le i\le M^*$.  Investigations beginning with Gardner~\cite{gardner1987maximum,gardner1988space} and Gardner-Derrida~\cite{gardner1988optimal} in the statistical physics literature provided a detailed, yet non-rigorous, picture for the storage capacity in the case of patterns distributed as $n$-dimensional Gaussian vectors.

More general perceptron models are defined by an activation function $U:\R \to \{0,1\}$\footnote{For an even more general setting see~\cite{bolthausen2021gardner}.}. We say a pattern $X_i$ is stored by $\sigma$ with respect to $U$ if $U(\ip{X_i}{\sigma})=1$.  Much recent work on these models have focused on two classes of activity functions: $U(x) = \mathbf 1_{x \ge \kappa \sqrt{n}}$ and $U(x) = \mathbf 1_{|x| \le \kappa \sqrt{n}}$.  The first defines the \textit{asymmetric binary perceptron}, the second the \textit{symmetric binary perceptron}.  We now detail some of the previous work on these models.
\subsection{Perceptron models}

\subsubsection{Asymmetric Binary Perceptron} 

We now define the classic binary perceptron, which we call the \textit{asymmetric binary perceptron} (\texttt{ABP}) throughout. Fix $\kappa\in\mathbb{R}$, $\alpha>0$; and set $M=\lfloor n\alpha \rfloor\in\mathbb{N}$. Let $X_i\distr \mathcal{N}(0,I_n)$, $1\le i\le M$, be i.i.d.\,random vectors, where $\mathcal{N}(0,I_n)$ denotes the $n-$dimensional multivariate normal distribution with zero mean and identity covariance. Consider  the (random) set
\begin{equation}\label{eq:asym-model}
S^{{\rm A}}_\alpha(\kappa) \triangleq \bigcap_{1\le i\le M}\Bigl\{\sigma\in \bincube:\ip{\sigma}{X_i}\ge \kappa\sqrt{n}\Bigr\}.
\end{equation}
The vectors $X_i\in\R^n$, $1\le i\le M$, are collectively referred to as the \emph{disorder}. In what follows, we slightly abuse the terminology and use ``disorder" to refer to both the vectors $X_i$, $1\le i\le M$; as well as the matrix $\M\in\R^{M\times n}$ whose rows are $X_i$.   The set $S^{{\rm A}}_\alpha(\kappa)$ is the \emph{solution space}, a random subset of $\bincube$. 

The computer science take on the perceptron model is to view it as an instance of a \emph{random constraint satisfaction problem}. Indeed, observe that $S_\alpha^{{\rm A}}(\kappa)$ is an intersection of $M$ random halfspaces, each defined by the \emph{constraint vector} $X_i$ (and threshold $\kappa$). Each constraint rules out certain solutions in the space $\mathcal{B}_n$ of all possible solutions; and the parameter $\alpha$ plays a role akin to the constraint density in the literature on random $k-$SAT, see e.g.~\cite{aubin2019storage,perkins2021frozen,abbe2021proof} for more discussion. For these reasons, we refer to $\alpha$ as the \emph{constraint density} in the sequel.

Perhaps the most important \emph{structural} question is whether $S_\alpha^{{\rm A}}(\kappa)$ is empty/non-empty (w.h.p., as $n\to\infty$). Krauth and M{\'e}zard conjectured in~\cite{krauth1989storage} that the event, $\bigl\{S^{{\rm A}}_\alpha(\kappa)\ne\varnothing\bigr\}$, exhibits what is known as a \emph{sharp  threshold}: there is an explicit threshold $\alpha_{{\rm KM}}(\kappa)$  such that
\begin{equation}\label{eq:asym-phase-transition}
\lim_{n\to\infty}\mathbb{P}\bigl[S_\alpha^{{\rm A}}(\kappa)\ne\varnothing\bigr] =\begin{cases}0,&\text{if } \alpha>\alpha_{{\rm KM}}(\kappa)\\1,&\text{if } \alpha<\alpha_{{\rm KM}}(\kappa)
\end{cases}.
\end{equation}
Using non-rigorous calculations based on the so-called \emph{replica method}, Krauth and M{\'e}zard~\cite{krauth1989storage} conjecture a precise  value of $\alpha_{{\rm KM}}(0)$ around 0.833. It is worth noting that this value deviates significantly from the \emph{first moment threshold}: note that for $\kappa=0$, $\mathbb{E}\bigl[\bigl|S_\alpha(\kappa)\bigr|\bigr] = \exp_2\bigl(n-n\alpha\bigr)$, which is exponentially small (in $n$) only for $\alpha>1$. 

The structure of $S_\alpha^{{\rm A}}(\kappa)$ and the aforementioned phase transition still (largely) remain as  open problems. Even the very existence of such a sharp phase transition point remains open, though Xu\footnote{Xu establishes this in a slightly different setting, where the disorder $X_i$ consists of i.i.d.\,Rademacher entries.}~\cite{xu2019sharp} has shown sharpness of the threshold around a possibly $n$-dependent value $\alpha^{(n)}_c(\kappa)$, as in~\cite{friedgut1999sharp} in the setting of random CSP's.  With that in mind, we can define 
\[ \alpha_c^*(\kappa) = \inf  \left \{ \alpha : \lim_{n\to\infty}\mathbb{P}\Bigl[S_\alpha^{{\rm A}}(0)= \varnothing\Bigr] =1  \right \}. \] 
The work by Ding and Sun~\cite{ding2019capacity} establishes, using an elegant second-moment argument, that for every $\alpha\le \alpha_{{\rm KM}}(0)$, 
\[
\liminf_{n\to\infty}\mathbb{P}\Bigl[S_\alpha^{{\rm A}}(0)\ne \varnothing\Bigr]>0.
\]
Hence, $\alpha_c^*(0)\ge \alpha_{\rm KM}(0)$. However, a matching upper bound is still missing: the best known bound is due to Kim and Roche~\cite[Theorem~1.2]{kim1998covering}, which show $\alpha_c^*(0)\le 0.9963$. More precisely, they show for any $\epsilon<0.0037$, $\mathbb{P}\Bigl[S_{1-\epsilon}^{A}(0)\ne\varnothing\Bigr]=o(1)$. For a similar negative result with a stronger convergence guarantee; that is a guarantee of form
\[
\mathbb{P}\Bigl[S_{1-\delta}^{A}(0)\ne\varnothing\Bigr]\le \exp\bigl(-\delta n\bigr)
\]
for some small $\delta>0$ (though potentially worse than $0.0037$), see Talagrand~\cite{talagrand1999intersecting}. 

When $S_\alpha^{{\rm A}}(\kappa)\ne \varnothing$ (w.h.p.), a follow-up algorithmic question is whether such a satisfying $\sigma\in\bincube$ can  be found algorithmically (in polynomial time). Regarding such positive results, the best known guarantee is again due to Kim and Roche. They devise in~\cite{kim1998covering} an (multi-stage majority) algorithm that w.h.p. returns a solution $\sigma\in S_\alpha^{{\rm A}}(0)$ as long as $\alpha<0.005$. (In particular, their algorithm is a constructive proof that $S_\alpha^{{\rm A}}(0)\ne\varnothing$ w.h.p.\,for $\alpha<0.005$.)  Later in Section~\ref{sec:kim-roche1}, we informally describe the implementation of their algorithm and establish that it is stable in an appropriate sense.

\subsubsection{Symmetric Binary Perceptron} 
Proposed initially by Aubin, Perkins, and Zdeborov{\'a} in~\cite{aubin2019storage}; the symmetric binary perceptron (\texttt{SBP}) model is our main focus in the present paper. Similar to the asymmetric case, fix a $\kappa>0$,  $\alpha>0$; and set $M=\lfloor n\alpha\rfloor$. Let $X_i\distr \mathcal{N}(0,I_n)$, $1\le i\le M$, be i.i.d.\,random vectors, and consider
\begin{equation}\label{eq:S_alpha}
S_\alpha(\kappa) \triangleq \bigcap_{1\le i\le M}\Bigl\{\sigma\in \bincube:\bigl|\ip{\sigma}{X_i}\bigr|\le \kappa\sqrt{n}\Bigr\} = \Bigl\{\sigma\in \bincube:\bigl\|\M\sigma\bigr\|_\infty \le \kappa\sqrt{n}\Bigr\},
\end{equation}
where $\M\in\R^{M\times n}$ with rows $X_1,\dots,X_M$. This model is called \emph{symmetric} since $\sigma\in S_\alpha(\kappa)$ iff $-\sigma\in S_\alpha(\kappa)$. It turns out that the symmetry makes the \SBP more amenable to analysis compared to its asymmetric counterpart, while retaining the relevant conjectural structural properties nearly intact, see~\cite{baldassi2020clustering}. Though not our focus here, it is worth mentioning that this is analogous to the random $k-$SAT model. Its symmetric variant, NAE $k-$SAT, is mathematically more tractable, yet at the same time exhibits similar structural properties. 

As its asymmetric counterpart,  it was conjectured that the \texttt{SBP} also undergoes a \emph{sharp phase transition}. More concretely, it was conjectured that there exists a $\alpha_c(\kappa)$ such that the event, $\{S_\alpha(\kappa)\ne\varnothing\}$, undergoes a sharp phase transition as $\alpha$ crosses $\alpha_c(\kappa)$. Notably, $\alpha_c(\kappa)$ matches with the first moment prediction:
\begin{equation}\label{eq:critical-alpha}
    \alpha_c(\kappa) \triangleq -\frac{1}{\log_2\mathbb{P}\bigl[|Z|\le \kappa\bigr]},\quad\text{where}\quad Z\sim \mathcal{N}(0,1).
\end{equation}
It was established in~\cite{aubin2019storage}  that (a) $\lim_{n\to\infty}\mathbb{P}\bigl[S_\alpha(\kappa)\ne\varnothing\bigr]=0$ for $\alpha>\alpha_c(\kappa)$; and (b) $\liminf_{n\to\infty}\mathbb{P}\bigl[S_\alpha(\kappa)\ne\varnothing\bigr]>0$ for $\alpha<\alpha_c(\kappa)$. The latter guarantee uses the so-called \emph{second moment method}, though falling short of establishing the high probability guarantee. Subsequent works by Perkins and Xu~\cite{perkins2021frozen}; and Abbe, Li, and Sly~\cite{abbe2021proof} establish that $\mathbb{P}\bigl[S_\alpha(\kappa)\ne\varnothing\bigr]=1-o(1)$ for all $\alpha<\alpha_c(\kappa)$. Namely, $\alpha_c(\kappa)$ is indeed a sharp threshold for the \texttt{SBP}. Having established the existence and the location of such a sharp phase transition; the next question, once again, is whether such a $\sigma\in S_\alpha(\kappa)$ can be found \emph{efficiently}; that is, by means of polynomial-time algorithms. This is our main focus in the present paper.

The \texttt{SBP} is closely related to \emph{combinatorial discrepancy theory}~\cite{spencer1985six,matousek1999geometric}. Given a matrix $\M\in\R^{M\times n}$, a central problem in discrepancy theory is to compute, approximate, or bound its \emph{discrepancy} $\mathcal{D}(\M)$:
\[
\mathcal{D}(\M) \triangleq \min_{\sigma\in\bincube} \bigl\|\M\sigma\bigr\|_\infty.
\]
Several different settings are considered in the discrepancy literature: \emph{worst-case} $\M$ and \emph{average-case} $\M$ (where the entries of $\M$ either i.i.d.\,Rademacher or i.i.d.\,Gaussian); and both existential and algorithmic results are sought. In the \emph{proportional regime}, the discrepancy perspective is to fix the aspect ratio $\alpha=M/n$ and find a solution $\sigma$ with small $\|\M\sigma\|_\infty$. This is  the inverse of the perceptron perspective: fixing  $\kappa>0$ and finding the largest $\alpha$ for which a solution $\sigma$ exists. In particular, the sharp threshold result for the \texttt{SBP} described above settles the question of discrepancy in the random proportional regime: for  $\M\in\R^{M\times n}$ with i.i.d. $\mathcal{N}(0,1)$ entries, $\mathcal{D}(\M) = (1+o(1))f(\alpha)\sqrt{n}$ w.h.p.\,where $f(\cdot)$ is the inverse function of $\alpha_c$. The first and second moment methods can also be employed to establish the value of discrepancy in the random setting  in other regimes, e.g.~\cite{potukuchi2018discrepancy,turner2020balancing,altschuler2021discrepancy}.  Moreover, as we describe below, discrepancy algorithms (e.g.~\cite{bansal2010constructive,lovett2015constructive,chandrasekaran2014integer,bansalspenceronline,potukuchi2020}) can be employed for the \texttt{SBP}. 
\subsection{Main Results}

From an algorithmic point of view, the most striking fact about the \texttt{SBP} is the existence of a large \textit{statistical-to-computational gap}.  Explanations for both the algorithmic hardness of the model and for the success of efficient algorithms at low densities have been put forth recently. 
\paragraph{A Statistical-to-Computational Gap.}

A random constraint satisfaction problem like the \texttt{SBP} is said to exhibit a \textit{statistical-to-computational gap} if the density below which solutions are known to exist w.h.p.\,is higher than the densities at which known efficient algorithms can find a solution.  As we now demonstrate, the \texttt{SBP} exhibits a statistical-to-computational gap for all $\kappa>0$, but this gap is most pronounced in the regime of small $\kappa$. 
 In this regime, the best known  algorithmic guarantee for finding a solution in the  \texttt{SBP} is due to Bansal and Spencer~\cite{bansalspenceronline} from the literature on combinatorial discrepancy. As we detail in Section~\ref{sec:alg-conj} and show in Corollary~\ref{thm:alpha-star-positive}, their algorithm works for $\alpha=O(\kappa^2)$ as $\kappa\to 0$. This stands in stark contrast to the threshold for the existence of solutions. From~\eqref{eq:critical-alpha},  $\alpha_c(\kappa)$  behaves  like $\frac{1}{\log_2(1/\kappa)}$: 
\[
\alpha_c(\kappa) = -\frac{1}{\log_2\mathbb{P}\bigl[|Z|\le \kappa \bigr]} = -\frac{1}{\frac12\log_2\frac{2}{\pi} + \log_2\bigl(1+o_\kappa(1)\bigr)+\log_2\kappa} = \frac{1}{\log_2(1/\kappa)}\bigl(1+o_\kappa(1)\bigr).
\]
Namely, $\alpha_c(\kappa)$ is asymptotically much larger than the algorithmic $\kappa^2$ threshold. The main motivation of the present paper is to inquire into the origins of this gap in the \texttt{SBP} by leveraging insights from  statistical physics. In particular, we will establish the presence of a geometric property known as the \emph{Overlap Gap Property} (OGP), and use it to rule out classes of \emph{stable algorithms}, appropriately defined.
\paragraph{Freezing, rare clusters, and algorithms.}

The \texttt{SBP} exhibits  striking structural properties which are thought to contribute to both the success of polynomial-time algorithms at low densities and the failure of efficient algorithms at higher densities.

On one hand, the model exhibits the ``frozen one-step Replica Symmetry Breaking (1-RSB)" scenario at all positive densities $\alpha<\alpha_c$.
 This states that whp over the instance, almost every solution $\sigma$ is \textit{totally frozen} and isolated: the nearest other solution is at linear Hamming distance to $\sigma$. This extreme form of clustering was conjectured to hold for the \texttt{ABP} and \texttt{SBP}
 in~\cite{krauth1989storage,huang2013entropy,aubin2019storage,baldassi2020clustering}, and subsequently established for the \texttt{SBP} in~\cite{perkins2021frozen,abbe2021proof}. In light of the earlier works by M\'ezard, Mora, and Zecchina~\cite{mezard2005clustering} and Achlioptas and Ricci-Tersenghi~\cite{achlioptas2006solution} positing a link between  clustering, freezing, and  algorithmic hardness, it is  tempting to postulate that finding a solution $\sigma$ for the \texttt{SBP}  is hard for every $\alpha\in(0,\alpha_c(\kappa))$, but this is contradicted by the existence of efficient algorithms at low densities such as that of~\cite{bansalspenceronline,braunstein2006learning,baldassi2007efficient,baldassi2009generalization,baldassi2015max} including the algorithm by Bansal and Spencer discussed above. Combining these facts, we arrive at the conclusion that the \SBP exhibits an intriguing phenomenon: the existence of polynomial-time algorithms can coexist with the  frozen 1-RSB phenomenon. This conundrum challenges the view that  clustering and freezing necessarily lead to algorithmic hardness.

In an attempt to explain this apparent conundrum, it was conjectured in~\cite{baldassi2015subdominant} that while a  $1-o(1)$ fraction of all solutions are totally frozen, an exponentially small fraction of solutions appear in clusters of exponential (in $n$) size; and the efficient learning algorithms that manage to find solutions find solutions belonging to such rare clusters, see~\cite{perkins2021frozen} for further discussion. In this direction, Abbe, Li, and Sly~\cite{abbe2021binary} established very recently that whp a  connected cluster of solutions of linear diameter does indeed exist at \emph{all} densities $\alpha< \alpha_c$. Furthermore, they show that an efficient multi-stage majority algorithm (based on that of~\cite{kim1998covering})  can find such a large cluster at  densities $\alpha =O(\kappa^{10})$ in the $\kappa \to 0$ regime\footnote{See in particular $\alpha_0$ appearing in~\cite[Page~6]{abbe2021binary}.}. 

These results and conjectures prompt several questions regarding the statistical-to-computational gap exhibited by the \texttt{SBP}. If large connected clusters exist at all subcritical densities, what is the reason for the apparent algorithmic hardness?  Do the efficient algorithms for densities $\alpha = O(\kappa^2)$ also find solutions lying in one of these large connected clusters?  At what densities are these large clusters algorithmically accessible?   In particular, while we now know  detailed structural information about the \texttt{SBP}, its statistical-to-computational gap remains a mystery. 
\paragraph{Our results on the Overlap Gap Property and failure of stable algorithms.}

We investigate the statistical-to-computational gap in the \texttt{SBP} via the \textit{Overlap Gap Property (OGP)}, an intricate geometrical property of the solution space that has been used to rigorously rule out large classes of search algorithms for many important random computational problems including  random $k-$SAT~\cite{gamarnik2017performance,coja2017walksat,bresler2021algorithmic} and  independent sets in sparse random graphs~\cite{gamarnik2014limits,rahman2017local,wein2020optimal}, see also the survey paper by Gamarnik~\cite{gamarnik2021overlap}. We will describe the OGP in more detail below. At a high level, it asserts the non-existence of tuples of solutions at prescribed distances in the solution space.

Our first main result establishes the OGP for $m-$tuples of solutions (dubbed as $m-$OGP) at densities $\Omega(\kappa^2\log_2\frac1\kappa)$:
\begin{theorem}[Informal, see Theorem~\ref{thm:m-OGP-small-kappa}]
For densities $\alpha = \Omega(\kappa^2\log_2\frac1\kappa)$, the \texttt{SBP} exhibits the $m-$OGP for appropriately chosen parameters. 
\end{theorem}
We also establish the presence of $2-$OGP and the $3-$OGP for the \texttt{SBP} in the high $\kappa$ regime, i.e. when $\kappa=1$, respectively in Theorem~\ref{thm:2-OGP} and Theorem~\ref{thm:3-OGP}. As we show in Theorem~\ref{thm:ogp-universality} through the multi-dimensional version of Berry-Esseen Theorem, our OGP results enjoy \emph{universality}: they remain valid under milder distributional assumptions on the entries of $\M$. 

Our next main result shows that the $m-$OGP rules out the class of \emph{stable algorithms} formalized in Definition~\ref{def:admissible-alg}. At a high level, an algorithm is stable if a small perturbation of its input results in a small perturbation of the solution $\sigma$ it outputs. In the literature on other random computational problems, it has been shown that the class of stable algorithms captures powerful classes of algorithms including  Approximate Message Passing algorithms~\cite{gamarnikjagannath2021overlap}, low-degree polynomials~\cite{gamarnik2020low,bresler2021algorithmic}, and low-depth circuits~\cite{gamarnik2021circuit}.
\begin{theorem}[Informal, see Theorem~\ref{thm:stable-hardness}]
The $m-$OGP implies the failure of stable algorithms for the \texttt{SBP}.
\end{theorem}
Thus, we obtain the following corollary:
\begin{coro}[Informal, see Theorem~\ref{thm:stable-hardness}]
Stable algorithms (with appropriate parameters) fail to find a solution for the \texttt{SBP} for densities $\alpha=\Omega(\kappa^2\log_2\frac1\kappa)$. 
\end{coro}
In particular, this hardness result matches the algorithmic $\kappa^2$ threshold up to a logarithmic factor. Hence, while the view that freezing implies algorithmic hardness for the \texttt{SBP} breaks down, the rigorous link between the OGP and  algorithmic hardness remains intact.

In addition to stable algorithms; we also consider the class of \emph{online algorithms} which includes the Bansal-Spencer algorithm~\cite{bansalspenceronline}. Informally, an algorithm $\A$ is online if the $t^{\rm th}$ coordinate of the solution it outputs depends only on the first $t$ columns of $\M$. 
\begin{theorem}[Informal, see Theorem~\ref{thm:online-alg-fail}]
Online algorithms fail to find a solution for the \texttt{SBP} for sufficiently high densities.
\end{theorem}

Having established the hardness of stable algorithms for the \texttt{SBP} at the $m-$OGP threshold; a natural follow-up question is whether the known efficient algorithms for perceptron models are stable and whether the \texttt{ABP} also exhibits the $m-$OGP. To that end, we investigate the stability property of the Kim-Roche algorithm~\cite{kim1998covering} for the \texttt{ABP}.
\begin{theorem}[Informal, see Theorem~\ref{thm:kr-stable}]
The Kim-Roche algorithm~\cite{kim1998covering} for the \texttt{ABP} is stable in the sense of Definition~\ref{def:admissible-alg}.
\end{theorem}
Investigating the stability of the Bansal-Spencer algorithm~\cite{bansalspenceronline} and whether the \texttt{ABP} also exhibits the OGP are among several open  questions we discuss in Section~\ref{sec:open}.
\subsection{Background and Related Work} 
\paragraph{Statistical-to-Computational Gaps.} As we noted, the \texttt{SBP} model exhibits a \emph{statistical-to-computational gap} (\texttt{SCG}): a gap between what the existential results  
guarantee (and thus what can be found with unbounded computational power), and what algorithms with bounded computational power (such as polynomial-time algorithms) can promise. Such \texttt{SCG}s are a ubiquitous feature in many algorithmic problems (with \emph{random inputs}) appearing in high-dimensional statistical inference tasks and in the study of random combinatorial structures. A partial (and certainly incomplete) list of problems with an \texttt{SCG} includes constraint satisfaction problems~\cite{mezard2005clustering,achlioptas2006solution,achlioptas2008algorithmic}, optimization problems over random graphs~\cite{gamarnik2014limits,coja2015independent,gamarnik2017} and spin glass models~\cite{chen2019suboptimality,gamarnik2020low,gamarnikjagannath2021overlap,gamarnik2021circuit}, number partitioning problem~\cite{gamarnik2021algorithmic}, principal component analysis~\cite{berthet2013computational,lesieur2015mmse,lesieur2015phase}, and the ``infamous" planted clique problem~\cite{jerrum1992large,deshpande2015improved,barak2019nearly}; see also the introduction of~\cite{gamarnik2021algorithmic}, the recent survey~\cite{gamarnik2021overlap}; and the references therein. 

Unfortunately, due to the so-called \emph{average-case} nature of these problems, the standard NP-completeness theory often fails to establish hardness for those problems even under the assumption $P\ne NP$. (It is worth noting though that a notable exception to this is when the problem exhibits \emph{random self-reducibility}, see e.g.~\cite{GK-SK-AAP} for such a hardness result regarding a spin glass model, conditional on a weaker assumption $P\ne \#P$.) Nevertheless, a very fruitful (and still active) line of research proposed certain forms of \emph{rigorous evidences} of algorithmic hardness for such average-case problems. These approaches include the failure of Monte Carlo Markov Chain methods~\cite{jerrum1992large,dyer2002counting}, low-degree methods and failure of low-degree polynomials~\cite{hopkins2018statistical,kunisky2019notes,gamarnik2020low,wein2020optimal,bresler2021algorithmic}, Sum-of-Squares~\cite{hopkins2015tensor,hopkins2017power,raghavendra2018high,barak2019nearly} and Statistical Query~\cite{kearns1998efficient, diakonikolas2017statistical,feldman2017statistical,feldman2018complexity} lower bounds, failure of the approximate message passing algorithm (an algorithm that is information-theoretically optimal for certain important problems, see e.g.~\cite{deshpande2014information,deshpande2017asymptotic})~\cite{zdeborova2016statistical,bandeira2018notes}; and the reductions from the planted clique problem~\cite{berthet2013computational,brennan2018reducibility,brennan2019optimal}, just to name a few. Yet another very promising such approach is through the \emph{intricate geometry} of the problem, via the so-called \emph{Overlap Gap Property} (OGP).

\paragraph{Overlap Gap Property (OGP).} Implicitly discovered by M{\'e}zard, Mora, and Zecchina~\cite{mezard2005clustering} and Achlioptas and Ricci-Tersenghi~\cite{achlioptas2006solution} (though coined later in~\cite{gamarnik2018finding}), the OGP approach leverages insights from the statistical physics to form a rigorous link between the intricate geometry of the solution space and formal algorithmic hardness. Informally, the OGP is a topological disconnectivity property, and states (in the context of a random combinatorial optimization problem, say over $\bincube$) that (w.h.p.\,over the randomness) any two near-optimal $\sigma_1,\sigma_2\in\bincube$ are either ``close" or ``far" from each other: there exists $0<\nu_1<\nu_2<1$ such that $n^{-1}\ip{\sigma_1}{\sigma_2}\in[0,\nu_1]\cup [\nu_2,1]$. That is, their (normalized) overlaps do not admit \emph{intermediate values}; and no two near-optimal solutions of intermediate distance can be found. It has been shown (see below) that  the OGP is a rigorous barrier for large classes of algorithms. See~\cite{gamarnik2021overlap} for a survey on OGP.

\paragraph{Algorithmic Implications of OGP.} The line of research relating the OGP to algorithmic hardness was initiated by Gamarnik and Sudan~\cite{gamarnik2014limits,gamarnik2017}. They consider the problem of finding a large independent set in the sparse random graphs with average degree $d$. It is known, see e.g.~\cite{frieze1990independence,frieze1992independence,bayati2010combinatorial}, that in the double limit (first sending $n\to\infty$, then letting $d\to\infty$), the largest independent set of this graph is of size $2\frac{\log d}{d}n$. On the other hand, the best known polynomial-time algorithm~\cite{karp1976probabilistic} (a very simple greedy protocol) returns an independent set that is half optimal, namely of size $\frac{\log d}{d}n$. In order to reconcile this apparent \texttt{SCG}, Gamarnik and Sudan study the space of all large independent sets. They establish that any two independent sets of size greater than $(1+1/\sqrt{2})\frac{\log d}{d}n$ exhibit OGP. By leveraging this, they show, through a contradiction argument, that \emph{local algorithms} (known as the \emph{factors of i.i.d.}) fail to find an independent set of size greater than $(1+1/\sqrt{2})\frac{\log d}{d}n$. Subsequent research (again via the lens of OGP) extended this hardness result to the class of low-degree polynomials~\cite{gamarnik2020low}. The extra ``oversampling" factor, $1/\sqrt{2}$, was removed by inspecting instead the the overlap pattern of many large independent sets (rather than the pairs), therefore establishing hardness all the way down to the algorithmic threshold. This was done by Rahman and Vir{\'a}g~\cite{rahman2017local} for \emph{local algorithms}, and by Wein~\cite{wein2020optimal} for \emph{low-degree polynomials}; and is also our focus here (see below). A list of problems where the OGP is leveraged to rule out certain classes of algorithms includes optimization over random graphs and spin glass models~\cite{gamarnikjagannath2021overlap,gamarnik2020low,gamarnik2021circuit,huang2021tight}, number partitioning problem~\cite{gamarnik2021algorithmic}, random constraint satisfaction problems~\cite{gamarnik2017performance,bresler2021algorithmic}.
\paragraph{Multi OGP ($m-$OGP).} As we just mentioned,  it was previously observed that by considering more intricate overlap patterns, one can potentially lower the (algorithmic) phase transition points further. This idea was employed for the first time by Rahman and Vir{\'a}g~\cite{rahman2017local} in the context of the aforementioned independent set problem. They managed to ``shave off" the extra $1/\sqrt{2}$ factor present in the earlier result by Gamarnik and Sudan~\cite{gamarnik2014limits,gamarnik2017}, and reached all the way down to the algorithmic threshold, $\frac{\log d}{d}n$. In a similar vein, Gamarnik and Sudan~\cite{gamarnik2017performance} studied the overlap structure of $m-$tuples $\sigma^{(i)}\in\bincube$, $1\le i\le m$ of satisfying assignments in the context of the Not-All-Equal (NAE) $k-$SAT problem. By showing the presence of OGP for $m-$tuples of nearly equidistant points (in $\mathcal{B}_n$), they established nearly tight hardness for \emph{sequential local algorithms}: their results match the computational threshold modulo factors that are polylogarithmic (in $k$). A similar overlap pattern (for $m-$tuples consisting of nearly equidistant points) was also considered by Gamarnik and K{\i}z{\i}lda\u{g}~\cite{gamarnik2021algorithmic} in the context of random number partitioning problem (NPP), where they established hardness well below the existential threshold. (It is worth noting that~\cite{gamarnik2021algorithmic} considers $m-$tuples where $m$ itself also grows in $n$, $m=\omega_n(1)$.)

More recently, $m-$OGP for more intricate forbidden patterns were considered to establish formal hardness in other settings. In particular, by leveraging $m-$OGP, Wein~\cite{wein2020optimal} showed that low-degree polynomials fail to return a large independent set (in sparse random graphs) of size greater than $\frac{\log d}{d}n$, thereby strengthening the earlier result by Gamarnik, Jagannath, and Wein~\cite{gamarnik2020low}. Wein's work establishes the \emph{ensemble} variant of OGP (an idea emerged originally in~\cite{chen2019suboptimality}): he considers $m-$tuples of independent sets where each set do not necessarily come from the same random graph, but rather from correlated random graphs. The ensemble variant of OGP was also considered in~\cite{gamarnik2021algorithmic} for the NPP. While technically more involved to establish, it appears that the ensemble $m-$OGP can be leveraged to rule out virtually any \emph{stable algorithm} (appropriately defined); and will also be our focus here. More recently, by leveraging the ensemble $m-$OGP; Bresler and Huang~\cite{bresler2021algorithmic} established nearly tight low-degree hardness results for the random $k-$SAT problem: they show that low-degree polynomials fail to return a satisfying assignment when the clause density is only a constant factor off by the computational threshold. In yet another work, Huang and Sellke~\cite{huang2021tight} construct a very intricate forbidden structure consisting of an ultrametric tree of solutions, which they refer to as the \emph{branching OGP}. By leveraging this branching OGP, they rule out overlap concentrated algorithms\footnote{A class that captures $O(1)$ iterations of gradient descent, approximate message passing; and Langevin Dynamics run for $O(1)$ time.} at the algorithmic threshold for the problem of optimizing mixed, even $p-$spin model Hamiltonian. 

\subsection{Open Problems}\label{sec:open}

\paragraph{Location of the Algorithmic Threshold.} 
We establish in Theorem~\ref{thm:m-OGP-small-kappa} that the \texttt{SBP} exhibits  $m-$OGP if $\alpha = \Omega\bigl(\kappa^2\log_2\frac1\kappa\bigr)$. On the other hand, we have per Corollary~\ref{thm:alpha-star-positive} that the Bansal-Spencer algorithm~\cite{bansalspenceronline} works when $\alpha = O(\kappa^2)$. In light of these, we make the following conjecture:
\begin{conjecture}\label{conj:tilde}
As $\kappa\to 0$, the algorithmic threshold for the \texttt{SBP} is at $\widetilde{\Theta}(\kappa^2)$. 
\end{conjecture}
In particular, we conjecture that up to factors that are polylogarithmic in $\frac1\kappa$, the Bansal-Spencer algorithm is the best possible within the class of efficient algorithms. That is, up to polylogarithmic factors no polynomial-time algorithm succeeds above the $m-$OGP threshold. An interesting question is whether the  $\log_2\frac1\kappa$ factor is necessary or it can be `shaved off'. We believe it might be possible to remove this factor by considering a more intricate overlap pattern, e.g.\,similar to those considered in~\cite{wein2020optimal,bresler2021algorithmic,huang2021tight}.

We now make Conjecture~\ref{conj:tilde} more precise. Given an $m\in\mathbb{N}$ and $\kappa>0$, let $\alpha_m^*(\kappa)$ be the smallest subcritical density such that the \texttt{SBP} exhibits $m-$OGP with appropriate parameters when $\alpha\ge \alpha_m^*$. 
We conjecture that the $m\to\infty$ limit of the $m-$OGP threshold marks the true algorithmic threshold: for every $\epsilon>0$ and $\kappa$ small enough, there do not exist polynomial-time algorithms for the \texttt{SBP} when $\alpha \ge (1+\epsilon)\lim_{m\to\infty}\alpha_m^*(\kappa)$. See Conjecture~\ref{conj-alpha-m} for details. This conjecture is backed up by the evidence that for many random computational problems including random $k-$SAT~\cite{bresler2021algorithmic},  independent sets in sparse random graphs~\cite{rahman2017local,wein2020optimal}, and mixed even $p-$spin model~\cite{huang2021tight}, the $m-$OGP matches or nearly matches the best known algorithmic threshold. 

Abbe, Li and Sly ask in~\cite[Question~1]{abbe2021binary} whether the algorithmic threshold for the \texttt{SBP}  coincides with the threshold for the existence of a `wide web': a cluster of solutions with maximum possible  diameter $n$.   One one hand, the existence of a wide web  rules out the 2-OGP: pairs of solutions of every possible overlap exist.  It would be very interesting to determine whether the threshold for existence of the wide web coincides with the conjectured  algorithmic threshold of $\tilde \Theta(\kappa^2)$ above, or even more precisely the limiting $m$-OGP threshold $\lim_{m\to\infty}\alpha_m^*$ (at least asymptotically as $\kappa \to 0$).

\paragraph{The Asymmetric Model.} As we noted earlier, the \texttt{ABP} is  more challenging from a mathematical perspective, and some of its basic properties are still far from being rigorously understood. In particular, even the very existence of a sharp phase transition and the frozen 1-RSB picture---both rigorously known to hold for the \texttt{SBP}---remain open. 

The \texttt{ABP} also exhibits a statistical-to-computational gap. On one hand, Kim-Roche algorithm~\cite{kim1998covering} finds solutions at low enough densities, specifically when $\alpha<0.005$. On the other hand, the result of Ding and Sun~\cite{ding2019capacity}  shows that solutions do exist (with probability bounded away from $0$) when $\alpha< \alpha_{\rm KM}(0) \approx 0.83$. It would be interesting to show that the \texttt{ABP}  exhibits $m-$OGP for some densities $\alpha < \alpha_{\rm KM}(0)$. To understand the statistical-to-computational gap of \texttt{ABP} further, it would be interesting to explore the model in the regime $\kappa\to\infty$ and investigate the $m-$OGP threshold and threshold for the existence of efficient algorithms. Further, there are other perceptron models one could explore in this regard, e.g.\,the \emph{$U-$function binary perceptron} introduced in~\cite{aubin2019storage}.

\paragraph{Stability of Other Algorithms.} We established in Theorem~\ref{thm:kr-stable} that the Kim-Roche algorithm for \texttt{ABP} is stable. In light of this, we make the following conjecture regarding the \texttt{SBP}:
\begin{conjecture}\label{conj:stable}
There exists a stable algorithm that finds a solution for the \texttt{SBP} w.h.p.\,when $\alpha=O(\kappa^2)$. 
\end{conjecture}
In particular, proving  stability of the Bansal-Spencer algorithm would resolve Conjecture~\ref{conj:stable}, but this seems challenging: the presence of a certain non-linear potential function (see~\cite[Equation~2.5]{bansalspenceronline}) renders the stability analysis difficult.

The algorithm of~\cite{abbe2021binary} is a  variant of the Kim-Roche algorithm that works for the \texttt{SBP} for  $\alpha=O\bigl(\kappa^{10}\bigr)$. Proving the stability of this algorithm would be an interesting first step towards resolving Conjecture~\ref{conj:stable}.

\paragraph{Broader Research Agendas on the OGP.} As mentioned above, the OGP is a provable barrier for a broad class of algorithms for many random computational problems. A list of such algorithms includes local/sequential local algorithms, Monte Carlo Markov Chain (MCMC) methods, low-degree polynomials, Langevin dynamics, approximate message passing type algorithms, low-depth circuits, and stable algorithms in general.  In many random computational problems (like $k$-SAT and independent sets) the OGP coincides with the threshold for the existence of known efficient search algorithms.  One might then conjecture (as we do here) that the OGP marks the true algorithmic threshold.  It would thus be very surprising and very interesting to find a case  where efficient algorithms succeed in the face of the OGP\footnote{Beyond those cases where algebraic techniques like Gaussian elimination can find solutions to `noiseless' problems like solving random linear equations.}. While random $k$-SAT, independent sets in random graphs, and other random CSP's have been studied for decades without finding such algorithms, algorithms for perceptron models have not been studied as extensively, especially not in the limiting regime $\kappa \to 0$ we focus on here, and thus this might be fruitful direction to pursue. 

\subsection{Paper Organization and Notation}
\paragraph{Paper Organization.} The rest of the paper is organized as follows. Our OGP results are stated in Section~\ref{sec:main-landscape}. In particular, 
we establish $2-$OGP and $3-$OGP for the high $\kappa$ case ($\kappa=1$) in Section~\ref{sec:high-kappa}; and the $m-$OGP for the regime $\kappa\to 0$ in Section~\ref{sec:low-kappa}. We then take an algorithmic route, and establish our main hardness result in Section~\ref{sec:alg-hard}; and formulate a conjecture pertaining the true algorithmic threshold in Section~\ref{sec:alg-conj}. In Section~\ref{sec:kim-roche1} we describe the Kim-Roche algorithm and show that it is stable. We record certain limitations of our approach in Section~\ref{sec:m-ogp-best-possible}. We show in Section~\ref{sec:future} that our OGP results enjoy \emph{universality} and extend beyond the Gaussian disorder. We provide complete proofs in Section~\ref{sec:proofs}. Finally in Appendix~\ref{appendix:matlab}, we provide a MATLAB code for verifying Lemma~\ref{asm:negativity} using which we establish $2-$OGP and $3-$OGP for $\kappa=1$.

\paragraph{Notation.} For any $n\in\N$, $[n]\triangleq \{1,2,\dots,n\}$. The binary cube $\{-1,1\}^n$ is denoted by $\bincube$. For any set $A$, $|A|$ denotes its cardinality. For any $r>0$ and $x\in\R$; $\exp_r(x)$ and $\log_r(x)$ denote respectively the exponential and logarithm functions base $r$. For any $v=(v_i:1\le i\le n)\in\R^n$ and $p>0$, $\|v\|_p\triangleq \bigl(\sum_{1\le i\le n}|v_i|^p\bigr)^{1/p}$, and $\|v\|_\infty=\max_{1\le i\le n}|v_i|$. For any $v,v'\in\R^n$, $\ip{v}{v'}\triangleq \sum_{1\le i\le n}v_iv_i'$ and $\Overlap{v}{v'}\triangleq n^{-1}\ip{v}{v'}$. For any $\sigma,\sigma'\in\bincube$, $d_H(\sigma,\sigma')$ denotes their Hamming distance. For $k\in\mathbb{N}$, $\ones\in \R^k$ denotes the vector of all ones (where the dimension will be clear from the context); and $I_k$ denotes the $k\times k$ identity matrix.  For any $x\in\R$, $\lfloor x\rfloor$ and $\lceil x\rceil$ respectively denote its floor and ceil. For $p\in[0,1]$, $h(p)\triangleq -p\log_2 p-(1-p)\log_2(1-p)$ is the binary entropy function (logarithm base two). For any $n\in\N$, $\boldsymbol{\mu}\in\R^n$ and $\boldsymbol{\Sigma}\in\R^{n\times n}$; $\cN(\boldsymbol{\mu},\boldsymbol{\Sigma})$  denotes the $n-$dimensional random vector having multivariate normal distribution with mean $\boldsymbol{\mu}$ and covariance $\boldsymbol{\Sigma}$. For any event $\mathcal{E}$, $\ind\{\mathcal{E}\}$ denotes its indicator. Given a matrix $\M$; $\|\M\|_F$, $\|\M\|_2$, $\sigma(\M)$, $\sigma_{\min}(\M)$, $\sigma_{\max}(\M)$, $|\M|$ and ${\rm trace}(\M)$ denote, respectively, its Frobenius norm, spectral norm, spectrum (that is, the set of its eigenvalues), smallest and largest singular values, determinant, and trace. A graph $\mathbb{G}=(V,E)$ is a collection of vertices $V$ together with some edges $(v,v')\in E$ between $v,v'\in V$. We consider herein only the simple graphs, namely those that are undirected with no loops. A graph $\mathbb{G}=(V,E)$ is called a \emph{clique} if for every distinct $v,v'\in V$, $(v,v')\in E$. We denote the clique on $m-$vertices ($m\in\mathbb{N}$) by $K_m$. A subset $S\subset V$ of vertices (of a $\mathbb{G}=(V,E)$) is called an \emph{independent set} if for every distinct $v,v'\in V$, $(v,v')\notin E$. The largest cardinality of such an independent set is called the \emph{independence number} of $\mathbb{G}$, denoted $\alpha(\mathbb{G})$. A $q-$coloring of a graph $\mathbb{G}=(V,E)$ is a function $\varphi:E\to\{1,2,\dots,q\}$ assigning to each $e\in E$ one of $q$ available colors. 

Throughout the paper, we employ the standard Bachmann-Landau asymptotic notation, e.g. $\Theta(\cdot),O(\cdot),o(\cdot)$, and $\Omega(\cdot)$. If there is no subscript, the asymptotic is with respect to $n\to\infty$. In the case where we consider asymptotics other than $n\to\infty$, we reflect this by a subscript: for instance, if $f$ is a function such that $f(\kappa)\to \infty$ as $\kappa\to 0$, we denote this by $f=\omega_\kappa(1)$. To keep our exposition clean, we omit floor/ceiling signs whenever appropriate. 
\section{OGP in the Symmetric Binary Perceptron}\label{sec:main-landscape}
In this section, we establish landscape results, dubbed as  \emph{ensemble $m-$OGP}, concerning the overlap structures of $m-$tuples $\left(\sigma^{(i)}:1\le i\le m\right)$, $\sigma^{(i)}\in\bincube$, that satisfy ``box constraints" with respect to potentially correlated instances of Gaussian disorder. 

\subsection{Technical Preliminaries}
We next formalize the notion of correlated instances through an appropriate interpolation scheme. 
\begin{definition}\label{def:ogp-set}
Fix a $\kappa>0$, and recall
\[
\alpha_c(\kappa) = -\frac{1}{\log_2 \mathbb{P}(|\mathcal{N}(0,1)|\le \kappa)}.
\]
Let $0<\alpha<\alpha_c(\kappa)$, $m\in\mathbb{N}$, $0<\eta<\beta<1$, and $\mathcal{I}\subset [0,\pi/2]$. Set $M=\lfloor n\alpha\rfloor$ and suppose that $\M_i\in\R^{M\times n}$, $0\le i\le m$, is a sequence of i.i.d.\,random matrices, each having i.i.d.\,$\mathcal{N}(0,1)$ coordinates. Denote by $\mathcal{S}_\kappa(\beta,\eta,m,\alpha,\mathcal{I})$ the set of all $m-$tuples $\bigl(\sigma^{(i)}:1\le i\le  m\bigr)$, $\sigma^{(i)}\in\bincube$, satisfying the following conditions.
\begin{itemize}
    \item[(a)] {\bf (Pairwise Overlap Condition)} For any $1\le i<j\le m$,
    \[
    \beta-\eta\le\Overlap{\sigma^{(i)}}{\sigma^{(j)}}\le \beta,
    \]
    where $\Overlap{\sigma^{(i)}}{\sigma^{(j)}}\triangleq  n^{-1}\ip{\sigma^{(i)}}{\sigma^{(j)}}$. 
    \item[(b)] {\bf (Rectangular Constraints)} There exists $\tau_i\in\mathcal{I}$, $1\le i\le m$, such that
    \[
   \bigl\|\M_i(\tau_i)\sigma^{(i)}\bigr\|_\infty\le \kappa\sqrt{n},\quad 1\le i\le m
    \]
    where  
    \begin{equation}\label{eq:interpolate-matrix}
        \M_i(\tau_i) = \cos(\tau_i)\M_0 + \sin(\tau_i)\M_i\in\R^{M\times n},\quad 1\le i\le m.
       \end{equation}
\end{itemize}
\end{definition}
The interpretations of the parameters appearing in Definition~\ref{def:ogp-set} are as follows. The parameter $m$ is the size of the tuples we inspect;  $\kappa$ is the constraint threshold; and $\alpha$ is the constraint density. That is, we consider $M=\lfloor n\alpha\rfloor$ random constraints. Parameters $\beta$ and $\eta$ control the (forbidden) region of pairwise overlaps. Finally, the index set, $\mathcal{I}$, is used for generating correlated instances of random constraints via interpolation $\M_i(\tau_i)$ in~\eqref{eq:interpolate-matrix}, $\tau_i\in\mathcal{I}$. This is necessary to study the ensemble OGP, see below.

As a concrete example to Definition~\ref{def:ogp-set}, consider the toy setting $m=2$ and $\mathcal{I}=\{0\}$. In this case, $S_\kappa\bigl(\beta,\eta,2,\alpha,\{0\}\bigr)$ is simply the set of all pairs $(\sigma_1,\sigma_2)\in\bincube\times \bincube$ such that (a) $\beta-\eta \le n^{-1}\ip{\sigma_1}{\sigma_2}\le \beta$ and (b) $\bigl\|\mathcal{M}\sigma_i\bigr\|_\infty\le \kappa\sqrt{n}$ for $i=1,2$; where $\mathcal{M}\in\R^{\lfloor \alpha n\rfloor\times n}$ is a random matrix with i.i.d.\,standard normal entries. 

\subsection{Landscape Results: High $\kappa$ Regime}\label{sec:high-kappa}
Our first focus is on the regime where $\kappa$ is \emph{large}. While we set $\kappa=1$ (thus $\alpha_c(\kappa)$ is approximately $1.8159$) for simplicity; our results extend easily to any fixed $\kappa>0$. In this case, we also drop the subscript $\kappa$ appearing in Definition~\ref{def:ogp-set}, and simply use the notation $\mathcal{S}(\beta,\eta,m,\alpha,\mathcal{I})$ to denote $\mathcal{S}_1(\beta,\eta,m,\alpha,\mathcal{I})$.

Our first result establishes $2-$OGP above $\alpha\ge 1.71$.
\begin{theorem}\label{thm:2-OGP}
Let $1.71\le \alpha\le \alpha_c(1)\approx 1.8159$. Then, there exists $0<\eta_2^*<\beta_2^*<1$ and a constant $c^*>0$ such that the following holds. Fix any $\mathcal{I}\subset [0,\pi/2]$ with $|\mathcal{I}|\le \exp_2\bigl(c^*n\bigr)$. Then,
\[
\mathbb{P}\Bigl[\mathcal{S}\bigl(\beta_2^*,\eta_2^*,2,\alpha,\mathcal{I}\bigr)\ne\varnothing\Bigr] \le \exp_2\bigl(-\Theta(n)\bigr).
\]
\end{theorem}
By considering the overlap structure of triples, one can further reduce the threshold (on $\alpha$) to approximately $1.667$ above which the overlap gap property takes place.
\begin{theorem}\label{thm:3-OGP}
Let $1.667\le \alpha\le \alpha_c(1)\approx 1.8159$. Then, there exists $0<\eta_3^*<\beta_3^*<1$ and a constant $c^*>0$ such that the following holds. Fix any $\mathcal{I}\subset [0,\pi/2]$ with $|\mathcal{I}|\le \exp_2\bigl(c^*n\bigr)$. Then,
\[
\mathbb{P}\Bigl[\mathcal{S}\bigl(\beta_3^*,\eta_3^*,3,\alpha,\mathcal{I}\bigr)\ne\varnothing\Bigr] \le \exp_2\bigl(-\Theta(n)\bigr).
\]
\end{theorem}
The proof of Theorem~\ref{thm:3-OGP} is provided  in Section~\ref{sec:proof-3-OGP}. The proof of Theorem~\ref{thm:2-OGP} is quite similar to that of Theorem~\ref{thm:3-OGP} (and in fact much simpler in terms of technical details); and is omitted. 

Theorem~\ref{thm:3-OGP} implies that $3-$OGP (with appropriate parameters) takes place for $\alpha\ge 1.667$, which is indeed  strictly smaller than the corresponding threshold of $\alpha\ge 1.71$ for $2-$OGP established in Theorem~\ref{thm:2-OGP}. An inspection of the proof reveals that our choice of $\eta^*$ satisfies $\eta^*\ll \beta^*$. That is, the structure that Theorem~\ref{thm:3-OGP} rules out corresponds essentially to (nearly) equilateral triangles in Hamming space.

 Theorem~\ref{thm:3-OGP} is established using the \emph{first-moment method}. More specifically, we let a certain random variable count the number of such triples. We then leverage Lemma~\ref{asm:negativity} to ensure that the exponent of the first moment of that random variable is negative under appropriate choices of parameters. That is, the  expectation is exponentially small (in $n$). Markov's inequality then yields Theorem~\ref{thm:3-OGP}. At a technical level, this amounts, in particular, to (a) counting the number of nearly equilateral triangles in the Hamming space;  and (b) applying a Gaussian comparison inequality by Sid{\'a}k~\cite{sidak1968multivariate} (reproduced herein as Theorem~\ref{thm:sidak} for completeness). It is worth noting though that unlike~\cite{gamarnik2021algorithmic}, our counting bound is  exact (up to lower-order terms). This appears necessary. Indeed, it appears not possible to improve upon Theorem~\ref{thm:2-OGP} if one  considers instead the relaxation to the ``star-shaped" forbidden structures (where the overlap constraint is relaxed to $\Overlap{\sigma^{(1)}}{\sigma^{(j)}}\in[\beta-\eta,\beta]$, $j\ge 2$) as in the counting step of~\cite[Theorems~2.3~and~2.6]{gamarnik2021algorithmic}. The aforementioned counting term appears more involved for $m\ge 4$.

As we noted earlier in the introduction, we do not pursue the $m-$OGP improvement for $m\ge 4$ in the high $\kappa$ regime. This is due to the following reason: the first moment method employed for establishing $m-$OGP actually fails as $m$ gets larger. That is, one can in fact (a) lower bound the first moment of the number $N$ of $m-$tuples corresponding to the forbidden structure that $m-$OGP deals with, and (b) show that for $m$ large, the value of $\alpha$ above which $\mathbb{E}[N]$ is $o(1)$ is actually strictly larger than $1.71$. This, of course, is only a failure of the first moment method, and does not necessarily imply that the $m-$OGP itself yields a worse threshold. In fact, given the previously mentioned prior work employing $m-$OGP as well as the fact that $m-$OGP deals with a more nested structure, it indeed makes sense that $m-$OGP (for $m\ge 4$) should hold for a much broader range of $\alpha$.\footnote{Here, it is worth noting that such a strict monotonicity in $m$ has also been conjectured by Ben Arous and Jagannath in the context of spherical spin glass models~\cite{arous2021shattering}.} For this reason, it is plausible to conjecture that considering $m-$OGP beyond $m\in\{2,3\}$ lowers the threshold on $\alpha$. We leave the formal verification of this for future investigation.

 Before we close this section, we remark that Baldassi, Della Vecchia, Lucibello, and Zecchina established in~\cite{baldassi2020clustering} similar OGP results for the high $\kappa$ case. To that end, fix any $x\in[0,1]$ and $K>0$. Using a first moment argument, they show the existence of a critical threshold $\alpha_{{\rm UB}}^{(m)}(x,K)$ such that the following holds: fix any $\alpha>\alpha_{{\rm UB}}^{(m)}(x,K)$; then w.h.p.\,there exists no $m-$tuple $\sigma_i\in S_\alpha(K)$ with fixed pairwise Hamming distances of $\lfloor nx\rfloor$. Namely, their results correspond to the case $\eta_2^*=\eta_3^*=0$. Furthermore, their results are rigorous for $m\in\{2,3,4\}$. However, they also suffer from technical difficulties similar to ours arising from the combinatorial terms for $m>4$. For this reason, they resort to non-rigorous calculations and a replica symmetric ansatz to study $m-$tuples beyond $m=4$.
 \subsection{Landscape Results: The Regime $\kappa\to 0$.}\label{sec:low-kappa}
We now turn to our results in the regime $\kappa\to 0$. Observe that for any fixed $\kappa>0$, the volume of the ``rectangular box" $[-\kappa,\kappa]^m$ (which eventually controls the probabilistic term) appearing in Definition~\ref{def:ogp-set} is $(2\kappa)^m$. When $\kappa\to 0$, this term actually shrinks further by increasing $m$. Thus, one can hope to pursue the $m-$OGP improvement. This is the subject of the present subsection.  Our main result to that end is as follows.
\begin{theorem}\label{thm:m-OGP-small-kappa}
Let
\begin{equation}\label{eq:OGP-threshold}
    \alpha_{{\rm OGP}}(\kappa)\triangleq 10\kappa^2\log \frac1\kappa.
\end{equation}
Then, for every sufficiently small $\kappa>0$ and $\alpha\ge \alpha_{{\rm OGP}}(\kappa)$, there exist $0<\eta<\beta<1$, $c>0$, and an $m\in\mathbb{N}$ such that the following holds. Fix any $\mathcal{I}\subset[0,\pi/2]$ with $|\mathcal{I}|\le \exp_2\bigl(cn\bigr)$. Then,
\[
\mathbb{P}\Bigl[\mathcal{S}_\kappa\bigl(\beta,\eta,m,\alpha,\mathcal{I}\bigr)\ne \varnothing\Bigr]\le \exp_2\bigl(-\Theta(n)\bigr).
\]
\end{theorem}
The proof of Theorem~\ref{thm:m-OGP-small-kappa} is in Section~\ref{sec:pf-m-OGP}.

Recall from our earlier discussion (also see Section~\ref{sec:alg-conj} and Corollary~\ref{thm:alpha-star-positive} therein) that the algorithm by Bansal and Spencer~\cite{bansalspenceronline} works for $\alpha= O(\kappa^2)$. On the other hand, no (efficient) algorithm is known for $\alpha \ge C\kappa^2$, where $C>0$ is a large absolute constant. Namely, the current known algorithmic threshold for the symmetric binary perceptron model is $\Theta(\kappa^2)$. In light of these facts, Theorem~\ref{thm:m-OGP-small-kappa} shows that the OGP threshold $\alpha_{{\rm OGP}}(\kappa)$ is \emph{nearly matching}: the onset of OGP  coincides up to polylogarithmic (in $\kappa$) factors with the  threshold (on $\alpha$) above which no polynomial-time algorithms are known to  work. The choice of the constant 10 appearing in~\eqref{eq:OGP-threshold} is for convenience and  can potentially be improved. 

We now comment on the extra $\log_2\frac1\kappa$ factor appearing in~\eqref{eq:OGP-threshold}. As we detail in Section~\ref{sec:m-ogp-best-possible}, the exponent of the first moment of the cardinality term, $\bigl|S_\kappa\bigl(\beta,\eta,m,\alpha,\mathcal{I}\bigr)\bigr|$, appears to be strictly positive (for every $\beta,\eta,m$) if $\alpha=O\bigl(\kappa^2\log_2(1/\kappa)\bigr)$. That is, Theorem~\ref{thm:m-OGP-small-kappa} is in a sense the best possible using our techniques. However, it is plausible that by considering a more delicate forbidden structure (akin to the ones studied in~\cite{wein2020optimal,bresler2021algorithmic,huang2021tight}), one may in fact be able to remove this logarithmic factor. This suggests two conjectures:  (a) in the regime $\kappa\to 0$, the algorithm by Bansal and Spencer~\cite{bansalspenceronline} is  best possible (up to constant factors); and that (b) the OGP  marks the onset of  algorithmic hardness.
\section{Algorithmic Barriers for the Perceptron Model}
\subsection{$m-$Overlap Gap Property Implies Failure of Stable Algorithms}\label{sec:alg-hard}
We commence this section by recalling our setup. We fix a $\kappa>0$, and an $\alpha<\alpha_c(\kappa)$ so that w.h.p.\,as $n\to\infty$, there exists a $\sigma \in S_\alpha(\kappa)$, where $S_\alpha(\kappa)$ is the (random) set introduced in~\eqref{eq:S_alpha}. Having ensured that $S_\alpha(\kappa)$ is (w.h.p.) non-empty; our focus in this section is  the problem of finding such a $\sigma$ by using \emph{stable algorithms}, formalized below.
\paragraph{Algorithmic Setting.} We interpret an algorithm $\A$ as a mapping from $\R^{M\times n}$ to $\mathcal{B}_n$. We allow $\A$ to be potentially randomized: we assume there exists an underlying probability space $(\Omega,\mathbb{P}_\omega)$ such that $\A:\R^{M\times n}\times \Omega\to \bincube$. That is, for any $\omega\in\Omega$ and disorder matrix $\M\in\R^{M\times n}$; $\A(\cdot,\omega)$ returns a $\sigma_{\rm ALG}\triangleq\A(\M,\omega)\in\bincube$; and we want $\sigma_{{\rm ALG}}$ to satisfy $\|\M\sigma_{{\rm ALG}}\|_\infty \le\kappa\sqrt{n}$.

We now formalize the class of stable algorithms that we investigate in the present paper. 
\begin{definition}\label{def:admissible-alg}
Fix a $\kappa>0$, an $\alpha<\alpha_c(\kappa)$; and set $M=\lfloor n\alpha\rfloor$. An algorithm $\A:\R^{M\times n}\times \Omega\to\bincube$ is called $(\rho,p_f,p_{\rm st},f,L)-$stable for the \texttt{SBP} model, if it satisfies the following for all sufficiently large $n$.
\begin{itemize}
    \item {\bf (Success)} Let $\M\in\R^{M\times n}$ be a random matrix with i.i.d $\cN(0,1)$ coordinates. Then,
    \[
    \mathbb{P}_{(\M,\omega)} \Bigl[\bigl\|\M\A(\M,\omega)\bigr\|_\infty\le \kappa\sqrt{n}\Bigr]\ge 1-p_f.
    \]
    \item {\bf (Stability)} Let $\M,\overline{\M}\in\R^{M\times n}$ be random matrices, each with i.i.d. $\cN(0,1)$ coordinates such that $\mathbb{E}\bigl[\M_{ij}\overline{\M}_{ij}\bigr]=\rho$ for $1\le i\le M$ and $1\le j\le n$. Then,
    \[
    \mathbb{P}_{(\M,\overline{\M},\omega)}\Bigl[d_H\bigl(\A(\M,\omega),\A(\overline{\M},\omega)\bigr)\le f+L\|\M-\overline{\M}\|_F\Bigr]\ge 1-p_{\rm st}.
    \]
\end{itemize}
\end{definition}
Definition~\ref{def:admissible-alg} is similar to the notion of stability considered in~\cite[Definition~3.1]{gamarnik2021algorithmic}. It is also worth noting that Definition~\ref{def:admissible-alg} applies also to deterministic algorithms $\A$. In this case, we simply modify the probability statements to reflect the fact that the only source of randomness is the input $\M$ (and $\overline{\M}$) to the algorithm. In the remainder of the paper, we often abuse the notation by dropping $\omega$ and simply referring to $\A:\R^{M\times n}\to \bincube$ as a randomized algorithm. 

We next highlight the operational parameters appearing in Definition~\ref{def:admissible-alg}. $\kappa$ is the ``width" of the ``rectangles" defined by the constraints. $\alpha$ is the  constraint density (also known as the aspect ratio). That is, $M=\lfloor n\alpha\rfloor$ is the number of constraints. The parameter $p_f$ controls the success guarantee. The parameters $\rho,p_{\rm st},f$ and $L$ collectively control the stability guarantee. The parameter $\rho$ essentially controls the amount of correlation. Stability parameters $p_{\rm st}$, $f$ and $L$ describe the amount of sensitivity of the algorithm's output to the correlation values. Our stability guarantee is probabilistic, where the probability is taken with respect to the joint randomness in $\M,\overline{\M}$ as well as to the coin flips $\omega$ of $\A$. The ``extra room" of $f$ bits makes our negative result only stronger: even when $\M$ and $\overline{\M}$ are very close, the algorithm is still allowed to make roughly $f$ flips. 

We now state our next main result.
\begin{theorem}\label{thm:stable-hardness}
Fix any sufficiently small $\kappa>0$, $\alpha\ge \alpha_{{\rm OGP}}(\kappa)=10\kappa^2\log_2\frac1\kappa$, and $L>0$. Let $m\in\mathbb{N}$ and  $0<\eta<\beta<1$ be the $m-$OGP parameters prescribed by Theorem~\ref{thm:m-OGP-small-kappa}. Set
\begin{equation}\label{eq:param-C-Q-T}
C=\frac{\eta^2}{1600},\qquad Q\triangleq \frac{4800L\pi }{\eta^2}\sqrt{\alpha},\qquad\text{and}\qquad T=\exp_2\Bigl(2^{4mQ\log_2 Q}\Bigr).
\end{equation}
Then, there exists an $n_0\in\mathbb{N}$ such that the following holds. For every $n\ge n_0$, there exists no randomized algorithm $\A:\R^{M\times n}\to \bincube$ that is
\[
\left(\cos\left(\frac{\pi}{2Q}\right),\frac{1}{9(Q+1)T},\frac{1}{9Q(T+1)},Cn,L\right)-\text{stable}
\]
for the \texttt{SBP}, in the sense of Definition~\ref{def:admissible-alg}.
\end{theorem}
The proof of Theorem~\ref{thm:stable-hardness} is provided in Section~\ref{sec:pf-hardness}. Several remarks are now in order. First, observe that there is no restriction on the running time of $\A$: as long as it is stable in the sense of Definition~\ref{def:admissible-alg} with appropriate parameters, Theorem~\ref{thm:stable-hardness} applies. 

Our second remark pertains to the scaling of parameters in the regime $n\to\infty$. Observe that the parameters $\alpha$, $L$, $m$ and $\eta$ are all $O(1)$ (in $n$) as $n\to\infty$; hence the parameters $C,Q$, and $T$ appearing in~\eqref{eq:param-C-Q-T} are all constants. In particular, $p_f$ and $p_{\rm st}$ are of constant order. This is an important feature of our result: the algorithms that we rule out have a \emph{constant probability} of success/stability. Namely, $\A$ need not have a high-probability guarantee. This is a notable departure from the main hardness result in~\cite[Theorem~3.2]{gamarnik2021algorithmic}, as well as from those appeared in prior works: unlike our case, the algorithms ruled out via OGP in those papers are required to succeed with high probability.

Our next remark pertains to the stability guarantee. Note that the algorithms that we rule out satisfy
\[
d_H\Bigl(\A\bigl(\M\bigr),\A\bigl(\overline{\M}\bigr)\Bigr)\le Cn+L\bigl\|\M-\overline{\M}\bigr\|_F.
\]
Namely, under our notation of stability the algorithm is still allowed to make $\Theta(n)$ flips when $\M$ and $\overline{\M}$ are ``nearly identical".

Our final remark pertains to the parameter $L$. We establish Theorem~\ref{thm:stable-hardness} for the case when $L$ is constant in order to keep our exposition clean. However, an inspection of our argument reveals $L$ can be pushed to $O\left(\frac{\log n}{\log \log n}\right)$.
\subsection{Failure of Online Algorithms for \texttt{SBP}}\label{sec:online-alg-fails}
Our next focus is on the class of  \emph{online algorithms}, formalized below.
\begin{definition}\label{def:online-alg}
Fix a $\kappa>0$, an $\alpha<\alpha_c(\kappa)$; and set $M=\lfloor n\alpha\rfloor\in\mathbb{N}$. Let $\mathcal{M}\in\R^{M\times n}$ be a disorder matrix with columns $\mathcal{C}_1,\mathcal{C}_2,\dots,\mathcal{C}_n\in\R^M$, and $\mathcal{A}:\R^{M\times n}\to\mathcal{B}_n$ be an algorithm where
\[
\mathcal{A}\bigl(\mathcal{M}\bigr)=\sigma=(\sigma_1,\sigma_2,\dots,\sigma_n)\in\bincube.
\]
We call $\mathcal{A}$ $p_f-$online if the following hold.
\begin{itemize}
    \item {\bf (Success)} For $\M$ consisting of i.i.d. $\mathcal{N}(0,1)$ entries,
    \[
    \mathbb{P} \Bigl[\bigl\|\M\A(\M)\bigr\|_\infty\le \kappa\sqrt{n}\Bigr]\ge 1-p_f.
    \]
    \item {\bf (Online)} There exists deterministic functions $f_t$, $1\le t\le n$ such that
    \[
    \sigma_t = f_t\bigl(\mathcal{C}_i:1\le i\le t\bigr)\in\{-1,1\}\qquad\text{for}\qquad 1\le t\le n.
    \]
\end{itemize}
\end{definition}
Several remarks are now in order. The parameter $p_f$ is the failure probability of $\A$: $\mathcal{A}(\mathcal{M})\in S_\alpha(\kappa)$ w.p.\,at least $1-p_f$. The second condition states that for all $1\le t\le n$, $\sigma_t$ is a function of $\mathcal{C}_1,\dots,\mathcal{C}_t$ only. More precisely, the signs $\sigma_i\in\{-1,1\}$, $1\le i\le t-1$, have been assigned at the end of round $t-1$. A new column $\mathcal{C}_t\in\R^M$ arrives in the beginning of round $t$, and $\mathcal{A}$ assigns a $\sigma_t\in\{-1,1\}$ depending only on the previous decisions. This highlights the online nature of $\A$.

Definition~\ref{def:online-alg} is an abstraction that captures, in particular, the algorithm by Bansal and Spencer~\cite{bansalspenceronline}. Our next result establishes that online algorithms fail to return a $\sigma\in S_\alpha(\kappa)$ for densities $\alpha$ close to the critical threshold $\alpha_c(\kappa)$. Similar to our treatment in Section~\ref{sec:high-kappa}, we stick to the case $\kappa=1$ for simplicity, even though our argument easily extends to arbitrary $\kappa>0$.
\begin{theorem}\label{thm:online-alg-fail}
Let $1.77\le \alpha\le \alpha_c(1)\approx 1.8159$. Then, there exists a constant $c_f>0$ such that the following holds. For any $p_f<\frac12-\exp\bigl(-c_f n\bigr)$, there exists no  $\mathcal{A}$ for \texttt{SBP} which is $p_f-$online in the sense of Definition~\ref{def:online-alg}.
\end{theorem}
The proof of Theorem~\ref{thm:online-alg-fail} is provided in Section~\ref{sec:pf-online-alg-fail}. The proof is based on a contradiction argument, which we informally describe. Given $\Delta\in(0,1)$, $\M\in\R^{M\times n}$, let $\M_\Delta\in\R^{M\times n}$ be obtained from $\M$ by independently resampling the last $\Delta\cdot n$ columns on $\M$. Fix an online algorithm $\A$, and let $\sigma\triangleq \A\bigl(\M\bigr)$, $\sigma_\Delta\triangleq \bigl(\M_\Delta\bigr)$. Then w.p.\,at least $1-2p_f$, $\|\M\sigma\|_\infty\le \sqrt{n}$ and $\|\M_\Delta\sigma_\Delta\|_\infty\le\sqrt{n}$. Furthermore, $\sigma$ and $\sigma_\Delta$ agree on first $n-\Delta n$ coordinates due to the online nature of $\A$. Namely, assuming such an $\A$ exists, we have $\mathbb{P}\bigl[\Xi(\Delta)\ne\varnothing\bigr]\ge 1-2p_f$, where $\Xi(\Delta)$ is the set of all pairs $(\sigma,\sigma_\Delta)\in\bincube\times\bincube$ such that $\|\M\sigma\|_\infty \le \sqrt{n}$, $\|\M_\Delta\sigma_\Delta\|_\infty\le \sqrt{n}$ and $n^{-1}\ip{\sigma}{\sigma_\Delta}\ge 1-2\Delta$. On the other hand,  a first moment argument (see in particular Proposition~\ref{prop:online-alg-fail}) reveals that for the same choice of $\Delta$, we actually have $\mathbb{P}\bigl[\Xi(\Delta)\ne\varnothing\bigr]\le \exp(-\Theta(n))$. This yields a contradiction and proves Theorem~\ref{thm:online-alg-fail}. 

The contradiction argument described above is slightly different than $2-$OGP, yielding a lower bound $\alpha\ge 1.77$. Notice that this is strictly larger than the corresponding $2-$OGP threshold, i.e.  $\alpha\ge 1.71$, for the same setting ($\kappa=1$) per Theorem~\ref{thm:2-OGP}. Lastly, the online algorithms that we rule out need not have a high probability guarantee: a success probability slightly above $\frac12$ suffices. 

\subsection{Algorithmic Threshold in \texttt{SBP}: A Lower Bound and a Conjecture}\label{sec:alg-conj}
\paragraph{Algorithmic Lower Bound in \texttt{SBP}.} Heretofore, we used $\Theta\bigl(\kappa^2\bigr)$ as our baseline for the current computational threshold for the \texttt{SBP}. Namely, against this threshold;  we (a) formulated the aforementioned \emph{statistical-to-computational gap} and (b) compared our hardness result, Theorem~\ref{thm:stable-hardness}, for the stable algorithms established via the $m-$OGP approach. In this section, we justify this choice for the algorithmic threshold, from the lower bound perspective.

As we mentioned in the introduction, the \texttt{SBP} is closely related to the well-known problem of minimizing the discrepancy of a matrix (or set system). The discrepancy minimization problem received much attention in the field of combinatorics and theoretical computer science; several efficient algorithms have been devised for it, see e.g.~\cite{rothvoss2017constructive,levy2017deterministic,eldan2018efficient,bansalspenceronline}. In what follows, we use the recent work by Bansal and Spencer~\cite{bansalspenceronline} as our baseline for postulating a computational threshold on $\alpha$ as one varies $\kappa$; though several of the algorithms cited above essentially yield the same $\Theta\bigl(\kappa^2\bigr)$ guarantee modulo different absolute constants. Before we proceed with the result of Bansal and Spencer~\cite{bansalspenceronline}; it is worth noting that there is yet another complementary line of research focusing on the so-called online guarantees, see e.g.~\cite{bansal2020online,bansal2021online,alweiss2021discrepancy,liu2021gaussian}. However, all of these algorithms suffer from extra polylogarithmic factors; and therefore their implied guarantees on $\alpha$ are poorer. That is they provably work only for $\alpha$ asymptotically much smaller than  $\kappa^2$. 

The work by Bansal and Spencer (see in particular~\cite[Section~3.3]{bansalspenceronline}) establishes the following.
\begin{theorem}{\cite[Theorem~3.4]{bansalspenceronline}}
\label{thm:bansal-spencer}
Let $T\in\mathbb{N}$ be an arbitrary time horizon, and $v_i\sim {\rm Unif}(\mathcal{B}_M)$, $1\le i\le T$, be i.i.d.\,random vectors. Then there exists a value $K>0$ and an algorithm that returns signs $s_1,\dots,s_T\in\{-1,1\}$ in ${\rm Poly}(M,T)$ time such that 
\[
\mathbb{P}\left[\left\|\sum_{i\le T}s_iv_i\right\|_\infty\le K\sqrt{M}\right]\ge 1-\exp\bigl(-cM\bigr).
\]
Here, $c,K>0$ are absolute constants independent of $M$ and $T$.
\end{theorem}
\begin{coro}\label{thm:alpha-star-positive}
There exists an absolute constant $K>0$ such that the following holds. Fix any $\kappa>0$, $\alpha<(\kappa/K)^2$; and consider the matrix $\mathcal{M}\in\R^{\alpha n\times n}$ with i.i.d.\,entries subject to the condition
\[
\mathbb{P}\bigl[\M_{ij}=+1\bigr] = \frac12=\mathbb{P}\bigl[\M_{ij}=-1\bigr],\qquad\text{for all}\qquad 1\le i\le \alpha n,1\le j\le n.
\]
Then, there exists an algorithm $\A$, running in ${\rm poly}(n)$ time, such that w.h.p.\,
\[
\Bigl\|\mathcal{M}\cdot \A\bigl(\mathcal{M}\bigr)\Bigr\|_\infty \le \kappa \sqrt{n}.
\]
\end{coro}
Corollary~\ref{thm:alpha-star-positive} is a direct consequence of Theorem~\ref{thm:bansal-spencer}. Indeed, consider $\mathcal{M}\in\{\pm 1\}^{\alpha T\times T}$ with $\alpha = n/T$, whose columns are $v_i$, $1\le i\le T$. Then one can find, in polynomial (in $n,T$) time, a $\sigma\in \mathcal{B}_T$ such that \[
\bigl\|\mathcal{M}\sigma \bigr\|_\infty\le K\sqrt{n} = K\sqrt{\alpha T}.
\]
Since $\alpha<(\kappa/K)^2$, the claim follows.

Admittedly, their result is established for the case of i.i.d.\,Rademacher disorder. Nevertheless, due to the aforementioned universality guarantees encountered in perceptron-like models, it is expected that the exact same guarantee (perhaps with a modified constant $K$) remains true for the case of i.i.d.\,standard normal disorder. 

\paragraph{A Conjecture on the Algorithmic Threshold.}
Recall from our prior discussion that for many random computational problems, the $m-$OGP threshold coincides (or nearly coincides) with conjectured algorithmic threshold. Examples include the problem of finding the largest independent set in random sparse graphs~\cite{rahman2017local,wein2020optimal}, NAE-$k$-SAT~\cite{gamarnik2017performance}, random $k-$SAT~\cite{bresler2021algorithmic}, mixed even $p-$spin model~\cite{huang2021tight}, and so on. In light of the preceding discussion, this is also the case for the \texttt{SBP} model: the limit of known algorithms  is at $\Theta(\kappa^2)$, whereas, as we establish in Theorem~\ref{thm:m-OGP-small-kappa}, the ensemble $m-$OGP holds for densities $\Omega\bigl(\kappa^2\log_2\frac1\kappa\bigr)$ in the regime $\kappa\to 0$.

On the other hand, unlike models such as the independent set problem, $k-$SAT, or the planted clique; prior to this work no conjectures were proposed regarding the threshold for algorithmic hardness in \texttt{SBP} model in the $\kappa\to 0$ regime. Here, we do put forward such a conjecture. To that end, let
\begin{equation}\label{eq:alpha-m}
    \alpha_m^*(\kappa)\triangleq \inf\left\{\alpha\in\bigl[0,\alpha_c(\kappa)\bigr]:\exists 1>\beta>\eta>0,\liminf_{n\to\infty}\mathbb{P}\Bigl[\mathcal{S}_\kappa\bigl(\beta,\eta,m,\alpha,\{0\}\bigr)=\varnothing\Bigr]=1\right\}.
\end{equation}
That is, $\alpha_m^*(\kappa)$ is the threshold for the $m-$OGP (with appropriate $\beta,\eta$). Let
\begin{equation}\label{eq:alpha-star}
\alpha_{\infty}^*(\kappa)\triangleq \lim_{m\to\infty}\alpha_m^*(\kappa),
\end{equation}
where the limit is well-defined  since $\bigl(\alpha_m^*\bigr)_{m\ge 1}$ is a non-increasing sequence of non-negative real numbers. Then we conjecture $\alpha_\infty^*(\kappa)$ marks the true algorithmic threshold for this problem. 
\begin{conjecture}\label{conj-alpha-m}
For any $\epsilon>0$, there exists a $\kappa^*(\epsilon)>0$ such that the following hold for every $\kappa\le \kappa^*(\epsilon)$:
\begin{itemize}
    \item There exists no polynomial-time search algorithms for the \texttt{SBP} if $\alpha>(1+\epsilon)\alpha_\infty^*(\kappa)$.
    \item There exists a polynomial-time search algorithm for the \texttt{SBP} if $\alpha<(1-\epsilon)\alpha_\infty^*(\kappa)$.
\end{itemize}
\end{conjecture} 
Recall that per Theorem~\ref{thm:m-OGP-small-kappa}, $\alpha_\infty^*(\kappa) = O\bigl(\kappa^2\log_2\frac1\kappa\bigr)$. Notice that the $\alpha_m^*(\kappa)$ (hence  the $\alpha_{\infty}^*(\kappa)$) are defined for the non-ensemble variant of $m-$OGP, $\mathcal{I}=\{0\}$. That is, $\sigma^{(i)}$, $1\le i\le m$, satisfy constraints dictated  by the rows of the {\bf same} disorder matrix $\M\in\R^{M\times n}$ with i.i.d. $\cN(0,1)$ (or Rademacher) entries, where $M=\lfloor\alpha n\rfloor$. This is merely for simplicity: the ensemble $m-$OGP and the non-ensemble $m-$OGP often take place at the {\bf exact same} threshold. The  former, on the other hand, is just technically more involved; and is necessary to rule out  certain classes of algorithms via an interpolation/contradiction argument as we do in this paper. The structural property implied by the non-ensemble OGP already suffices to predict the desired algorithmic threshold.
\subsection{Stability of the Kim-Roche Algorithm}\label{sec:kim-roche1}
Having established that the $m-$OGP is a provable barrier for the class of stable algorithms, it is then natural to inquire whether the class of stable algorithms captures the implementations of known algorithms for  perceptron models. In this section, we investigate this question for a certain algorithm devised for the \emph{asymmetric} model, which we recall from~\eqref{eq:asym-model}.

Kim and Roche devised in~\cite{kim1998covering} an algorithm which admits, as its input, a disorder matrix $\M\in\R^{k\times n}$ with i.i.d.\,entries; and returns a $\sigma\in\bincube$ such that $\M\sigma\in\R^k$ is entry-wise non-negative as long as $k<0.005n$. That is, their algorithm provably returns a $\sigma\in S_\alpha^{A}(0)$ as long as $\alpha<0.005$. (We use $k$ in place of $M$ for the number of constraints so as to be consistent with their notation.) We denote their algorithm by $\KRA:\R^{k\times n}\to \bincube$ as a shorthand notation. It is worth noting that while their results are established for the case where $\M$ consists of i.i.d.\,Rademacher entries, they easily extend to the case of Gaussian $\cN(0,1)$ entries, which will be our focus here. $\KRA$ takes $O\bigl(\log_{10}\log n_{10}\bigr)$ steps, each requiring ${\rm poly}(n)$ time.\footnote{Throughout this section, we consider all logarithms in base 10 in order to be consistent with the notation of~\cite{kim1998covering}.} Namely, $\KRA$ is an \emph{efficient} algorithm that provably works in the so-called \emph{linear regime}, $k=\Theta(n)$. Admittedly, $\KRA$ is tailored for the asymmetric model. Nevertheless, there are only a few known algorithms with rigorous guarantees for perceptron models; thus it is indeed natural to explore the stability of $\KRA$. 

\paragraph{Operational Parameters.}
We next provide details of the Kim-Roche algorithm from~\cite{kim1998covering}. Let
\begin{equation}\label{eq:f-j}
f_0=1,\quad f_1=\frac{1}{200},\qquad\text{and}\quad f_j = 10^{-2^j},\qquad\text{for}\qquad 2\le j\le N,
\end{equation}
as in~\cite[Equation~5.42]{kim1998covering}, where $N=\lceil C\log_{10}\log_{10} n\rceil$ is the total number of rounds. Next, let
\begin{equation}\label{eq:k-j}
k_0 =k,\quad k_1= 2\lfloor (1/2)(n/10^8)\rfloor +1,\quad\text{and}\quad k_s=2\lfloor (1/2)(n\cdot f_s^3)\rfloor+1\quad\text{for}\quad 2\le s\le N
\end{equation}
as in~\cite[Equation~5.46]{kim1998covering}. 
Set 
\begin{equation}\label{eq:kr-A}
A\triangleq \sum_{0\le j\le N}f_j;
\end{equation}
and let $n_j$  be defined by
\begin{equation}\label{eq:n-j}
n_0=\lfloor n/A\rfloor,\quad\text{and}\quad n_j = \left \lfloor \frac{n}{A} \sum_{0\le i\le j}f_i\right\rfloor  - \left \lfloor\frac{n}{A}\sum_{0\le i\le j-1} f_i \right\rfloor,\quad1\le j\le N
\end{equation}
per~\cite[Equation~5.44]{kim1998covering}. 
\paragraph{Informal Description of the Algorithm.} We now describe the $\KRA$ algorithm. To that end, denote by $R_1,\dots,R_k\in\R^n$ the rows of $\M$. Given a set $P$ of rows  and a set $Q$ of columns, let $\M(P,Q)$ denote the $|P|\times |Q|$ submatrix $M_{ij}$ obtained by retaining rows $i\in P$ and columns $j\in Q$. 
\begin{itemize}
    \item In the first round, $\KRA$ assigns $n_0$ coordinates of $\sigma$ by taking the majority vote in submatrix $\mathcal{M}([1,k]:[1,n_0])$. That is,
    \[
    \sigma_j = {\rm sgn}\left(\sum_{1\le i\le k}\M_{ij}\right)\qquad\text{for}\qquad 1\le j\le n_0. 
    \]
    \item For each row $R_i$ of $\M$, it then computes the partial inner products $\ip{R_i}{\sigma}$ (restricted to $\R^{n_0}$), finds an index set $\mathcal{I}_1$ corresponding to $k_1$ smallest indices; and takes a majority vote in the submatrix $\mathcal{M}(\mathcal{I}_1:[n_0+1,n_0+n_1])$ and repeats this procedure. 
    \item In particular, at the beginning of round $j\ge 1$, one has a vector $\sigma\in\{\pm 1\}^{\sum_{0\le s\le j-1}n_s}$. One  then computes the partial inner products $\ip{R_i}{\sigma}$, $1\le i\le k$, and computes an index set  $\mathcal{I}_j$ with $|\mathcal{I}_j|=k_j$ such that $i\in \mathcal{I}_j$ iff $\ip{R_i}{\sigma}$ is among the $k_j$ smallest (partial) inner products. Taking then the majority vote in the submatrix
\[
\mathcal{M}\left(\mathcal{I}_j:\left[1+\sum_{0\le s\le j-1}n_s,\sum_{0\le s\le  j}n_s\right]\right)
\]
settles next $n_j$ entries of $\sigma$, that is the entries $\sigma_j$, $1+\sum_{0\le s\le j-1}n_s\le j\le \sum_{0\le s\le  j}n_s$. Namely, a $k_j\times n_j$ submatrix is used for determining next $n_j$ components of $\sigma$ and $\overline{\sigma}$. 
\end{itemize}
Numerically, $n_0\approx 0.995n$. Thus even at the beginning, $\KRA$ already settles \emph{most} of the entries of $\sigma\in\bincube$.

 Having described the $\KRA$ informally, we are now in a position to state our main result. We show that $\KRA$ is stable in the sense of Definition~\ref{def:admissible-alg}.
\begin{theorem}\label{thm:kr-stable}
Let $\M\in\R^{k\times n}$ and $\M'\in\R^{k\times n}$ be two i.i.d.\,random matrices each with i.i.d.$\mathcal{N}(0,1)$ entries; and let
\begin{equation}\label{eq:matrix-interpolate}
\overline{\M}(\tau)\triangleq \cos(\tau)\M+\sin(\tau)\M'\in\R^{k\times n},\qquad \tau\in\left[0,\frac{\pi}{2}\right].
\end{equation}
Set $\tau=n^{-0.02}$. Then,
\[
\mathbb{P}\Bigl[d_H\Bigl(\KRA(\M),\KRA\bigl(\overline{\M}(\tau)\bigr)\Bigr)=o(n)\Bigr]\ge 1-O\left(n^{-\frac{1}{41}}\right).
\]
\end{theorem}
As a result, the Kim-Roche algorithm is
\[
\Bigl(\cos\bigl(n^{-0.02}\bigr),o\bigl(1/n\bigr),O\left(n^{-1/41}\right),Cn,L\Bigr)-\text{stable}
\]
in the sense of Definition~\ref{def:admissible-alg} for any $C>0$ and $L>0$ (see below for further details).

In order to establish Theorem~\ref{thm:kr-stable}, we first establish in Section~\ref{sec:proof-of-kr} an auxiliary result, Proposition~\ref{prop:kr-stable}, which pertains to the partial implementation of $\KRA$. That is, we analyze $\KRA$ run for $c\log_{10}\log_{10}n$ rounds (where $c>0$ is a small enough constant) as opposed to its full $N=\lceil C\log_{10}\log_{10}n \rceil$ round implementation; and show that it is stable. We then show in Section~\ref{sec:sub-sub-kr-stable} that the number of unassigned coordinates, $\sum_{c\log_{10}\log_{10}n+1\le j\le N}n_j$, is $o(n)$. This, together with Proposition~\ref{prop:kr-stable}, establishes Theorem~\ref{thm:kr-stable}.

Several pertinent remarks are now in order. We first highlight that $\KRA$ is indeed stable in the sense of Definition~\ref{def:admissible-alg} with the parameters noted above. (Here, we suppress the randomness, and the success guarantee is now for the event $\bigl\{\M\A(\M)\ge 0\bigr\}$ entry-wise.) To that end, an inspection of~\cite[Theorem~1.4]{kim1998covering} reveals that for $\M\in\R^{k\times n}$ with $k=\alpha n$ having i.i.d.\,$\cN(0,1)$ entries,
\[
\mathbb{P}_{\M}\Bigl[\M\A(\M)\ge 0\Bigr]\ge 1-o\bigl(n^{-1}\bigr),
\]
as long as $\alpha<0.005$. With this, we obtain that $\KRA$ is 
\[
\Bigl(\cos\bigl(n^{-0.02}\bigr),o\bigl(1/n\bigr),O\left(n^{-1/41}\right),Cn,L\Bigr)-\text{stable}
\]
for any $\alpha<0.005$, $C>0$ and $L>0$ for the asymmetric binary perceptron in the sense of Definition~\ref{def:admissible-alg}. (In fact, one can take $L=0$ as $C>0$.) Note though that this parameter scaling is not comparable with Theorem~\ref{thm:m-OGP-small-kappa} since Theorem~\ref{thm:m-OGP-small-kappa} pertains to the symmetric model, see below for more details.

Recalling
\[
\mathcal{O}\bigl(\sigma,\overline{\sigma}\bigr)=n^{-1}\ip{\sigma}{\overline{\sigma}}=1-2d_H\bigl(\sigma,\overline{\sigma}\bigr)/n,
\]
it follows that in the setting of Theorem~\ref{thm:kr-stable}, $\mathcal{O}\bigl(\sigma,\overline{\sigma}\bigr)=1-o(1)$. That is, $\sigma$ and $\overline{\sigma}$ agree on all but a vanishing fraction of coordinates. Informally, this suggests that $\KRA$ cannot overcome the overlap barrier of $\eta$ appearing in Theorems~\ref{thm:2-OGP},~\ref{thm:3-OGP}, and~\ref{thm:m-OGP-small-kappa} as $\eta=O(1)$. However, we established the OGP results for the symmetric case as opposed to the asymmetric model for which $\KRA$ is devised. Thus Theorem~\ref{thm:kr-stable} is not exactly compatible with the hardness result, Theorem~\ref{thm:stable-hardness}. A more compelling picture would be to show that the OGP takes place also for the asymmetric model, with an $\eta$ that is of order $O(1)$; and then couple such result with Theorem~\ref{thm:kr-stable}. We leave this as a very interesting direction for future work.

Lastly, it would also be very interesting to prove that the algorithm by Abbe, Li and Sly~\cite{abbe2021binary} devised for the \texttt{SBP} is also stable in the relevant sense. An inspection of~\cite{abbe2021binary} reveals that several of the key  steps are similar to~\cite{kim1998covering}, but there are a few differences which prevent immediate verification of stability. We now elaborate on this by highlighting fundamental differences between their algorithm and $\A_{\rm KR}$. 
Inspecting~\cite[Page~7]{abbe2021binary}, it appears that a major change is the incorporation of extra sign parameters inside the summation: using the exact same notation as in~\cite{abbe2021binary},  this is
\[
{\rm sgn}\left(\sum_{r\in\mathcal{R}_i}-{\rm sgn}\bigl(S^{(r)}(0:i-1)\bigr)G_{r,j}\right).
\]
This extra term is independent of the summands, and therefore, is benign. As a result, it appears that our Lemmas~\ref{lemma:main-card} and~\ref{lemma:variance-up-bd} apply almost verbatimly. Their algorithm has two additional steps, one in the beginning and one at the end; these steps appear to be stable, as well. A main technical challenge, however, is that their algorithm requires $O(\sqrt{\log n})$ rounds, as opposed to the Kim-Roche algorithm that requires only $O(\log \log n)$ rounds. The choice of $\log\log n$ is crucial for our argument, see in particular Proposition~\ref{prop:kr-stable} below. It is not clear though if they need $O(\sqrt{\log n})$ steps to find a large cluster and whether one can achieve the much more modest goal of finding a solution $\sigma$ in $O(\log\log n)$ rounds. We leave the formal investigation of this as a very interesting direction for future research.

\section{Natural Limitations of Our Techniques}\label{sec:m-ogp-best-possible}
Recall from our earlier discussion that the algorithmic threshold for the SBP model appears to be $\Theta(\kappa^2)$, whereas we established $m-$OGP for densities above $\Omega\bigl(\kappa^2\log_2\frac1\kappa)$. That is, the OGP threshold is off by a polylogaritmic (in $1/\kappa$) factor.

In this section, we investigate whether one can shave off this extra $\log_2\frac1\kappa$ factor. In a nutshell, we provide an informal argument suggesting that to establish $m-$OGP for the structure that we consider, $\alpha=\Omega\bigl(\kappa^2\log_2 \frac1\kappa\bigr)$ appears necessary. 

To that end, fix first a $\kappa>0$, where one should think of $\kappa$ to be sufficiently small. An inspection of the proof of Theorem~\ref{thm:m-OGP-small-kappa} reveals that the first moment is controlled by a certain $\Upsilon(\beta,\alpha)$ appearing in~\eqref{eq:upsilon}, which we repeat below for convenience:
\begin{equation}\label{eq:upsilon-text}
    \Upsilon(\beta,\alpha) = h\left(\frac{1-\beta}{2}\right)-\frac{\alpha}{2}\log_2(2\pi)+\alpha \log_2(2\kappa)-\frac{\alpha}{2}\log_2(1-\beta).
\end{equation}
In particular, for the first moment argument to work, it should be the case that $\Upsilon(\beta,\alpha)$ is negative for an appropriate choice of parameters $\beta,\alpha$.

Now set $\delta\triangleq \frac{1-\beta}{2}$. Yet another inspection of the proof of Theorem~\ref{thm:m-OGP-small-kappa} shows that  for the first moment method to be applicable, $\beta$ should be close to one. For this reason, the regime of interest below is therefore $\delta \to 0$ and $\kappa\to 0$.
\subsection*{Step 1: $\delta>\kappa^2$ is necessary.} Note that as $\delta\to 0$,
\[
h(\delta) = -\delta\log_2\delta -(1-\delta)\log_2(1-\delta) = -\delta\log_2\delta + \Theta_\delta\bigl(\delta\bigr),
\]
using the Taylor expansion, \[
\log_2(1-\delta)=-\frac{1}{\ln 2}\delta + o(\delta).
\]
With this, we manipulate~\eqref{eq:upsilon-text} to arrive at
\[
\Upsilon(\beta,\alpha) = \delta\log_2\frac1\delta +\frac{\alpha}{2}\log_2\frac1\delta - \alpha\log_2\frac1\kappa +\Theta_\delta\bigl(\delta)+o_{\delta,\kappa}\bigl(\alpha\log_2\kappa\bigr).
\]
Now, for $\Upsilon(\beta,\alpha)$ to be negative, we must have $\alpha\log_2\frac 1\kappa>\frac{\alpha}{2}\log_2\frac1\delta$. This immediately yields $\delta>\kappa^2$ to be a \emph{necessary} condition. 
\subsection*{Step 2: $\alpha=\Omega\bigl(\kappa^2\log_2\frac1\kappa\bigr)$ is necessary.} Set $\delta=C\kappa^2$ for $C\triangleq C(\kappa)>1$.  The expression for $\Upsilon(\beta,\alpha)$ then becomes
\begin{align}
    \Upsilon(\beta,\alpha)&=\underbrace{h\left(\frac{1-\beta}{2}\right)}_{-\delta\log_2\delta+\Theta_\delta(\delta)}-\frac{\alpha}{2}\log_2(2\pi)+\alpha \log_2(2\kappa)-\frac{\alpha}{2}\log_2(1-\beta)\nonumber\\
    &=-\delta\log_2\delta+\Theta_\delta(\delta)+\alpha\log_2\kappa-\frac{\alpha}{2}\log_2\delta-\frac{\alpha}{2}\log_2 \pi\nonumber\\
    &=2C\kappa^2\log_2\frac1\kappa+C\kappa^2\log_2\frac1C - \frac{\alpha}{2}\log_2 C-\frac{\alpha}{2}\log_2\pi\label{eq:upsilon-helpful}.
\end{align}
Note that $\delta<1$, and thus $C<\frac{1}{\kappa^2}$. Thus $\log_2 C<2\log_2\frac1\kappa$. Equipped with this observation, we investigate two separate cases for the growth of $C$.

\paragraph{Case 1: $\log_2 C=o_\kappa(\log_2\frac1\kappa)$.} Then $C\kappa^2\log_2\frac1\kappa$ dominates the term, $C\kappa^2\log_2\frac1C$ apearing in~\eqref{eq:upsilon-helpful}. In this case for $\Upsilon(\beta,\alpha)$ to be negative, one must indeed ensure
\[
\frac{\alpha}{2}\log_2 C>2C\kappa^2\log_2\frac1\kappa.
\]
Rearranging this, we find
\[
\alpha>\frac{4C}{\log_2 C}\kappa^2\log_2\frac1\kappa \implies \alpha = \Omega\left(\kappa^2\log_2\frac1\kappa\right),
\]
which is precisely our claim. In fact, the parameter $\beta$ for which Theorem~\ref{thm:m-OGP-small-kappa} is established is of form $\beta = 1-\Theta(\kappa^2)$,  see~\eqref{eq:beta-star}. Hence, one has $\delta=\Theta(\kappa^2)$ and $C=\Theta_\kappa(1)$, thus $\log_2 C $ is indeed $o_\kappa\bigl(\log_2\frac1\kappa\bigr)$.
\paragraph{Case 2: $\log_2C = \Theta_\kappa\bigl(\log_2\frac1\kappa\bigr)$.} In this case, we now show that the threshold on $\alpha$ is worse than the one appearing in the previous case. 

To that end, set $C\sim \kappa^{-\gamma}$, $\gamma<2$: that is, we assume
\[
\lim_{\kappa \to 0}\frac{\log_2 C}{\log_2\frac1\kappa}=\gamma.
\]
We focus on certain terms appearing in~\eqref{eq:upsilon-helpful}. Note that, \[
C\kappa^2\log_2\frac1\kappa \sim 2\kappa^{2-\gamma}\log_2\frac1\kappa,\quad C\kappa^2\log_2 \frac1C=-\gamma \kappa^{2-\gamma}\log_2\frac1\kappa,\quad\text{and}\quad  \frac{\alpha}{2}\log_2 C \sim \frac{\alpha}{2}\gamma \log_2\frac1\kappa.
\]
Combining these findings, we immediately observe that for $\Upsilon(\beta,\alpha)$ to be negative, one must have $\alpha\sim \kappa^{2-\gamma}$, where any $\gamma<2$ works. Notice that this threshold is strictly worse than $\kappa^2\log_2\frac1\kappa$.

Hence, $\alpha=\Omega\bigl(\kappa^2\log_2\frac1\kappa\bigr)$ is indeed necessary for $m-$OGP (for the configuration we consider with a sufficiently large $m\in\mathbb{N}$ and $0<\eta<\beta<1$) to take place. It is though conceivable that by establishing the OGP for a potentially more intricate structure, like the ones considered in~\cite{wein2020optimal,bresler2021algorithmic,huang2021tight}, one may in fact reach all the way down to $\Theta(\kappa^2)$. We leave this extension as an interesting future research direction.
\section{Universality in OGP: Beyond Gaussian Disorder}\label{sec:future}
Our OGP results, Theorems~\ref{thm:2-OGP},~\ref{thm:3-OGP} and~\ref{thm:m-OGP-small-kappa}, are established for the case where the disorder matrix $\mathcal{M}\in\R^{M\times n}$  consists of i.i.d.\,\,$\cN(0,1)$ entries. However, much like many other properties regarding the perceptron model, the OGP also enjoys the \emph{universality}. In other words, the exact details of the distribution (of disorder) are immaterial; and provided that certain (rather mild) conditions on the distribution are satisfied, the OGP results still remain valid. 

We now (somewhat informally) elaborate on the mechanics of this extension. The main technical tool that we employ is the multi-dimensional version of the Berry-Esseen Theorem, which is reproduced herein for convenience.
\begin{theorem}\label{thm:berry-esseen}
Let $Y_1,Y_2,\dots,Y_n\in\R^m$ be independent centered random vectors. Suppose $S=\sum_{1\le i\le n}Y_i$ and $\Sigma = {\rm Cov}(S)\in\R^{m\times m}$ is invertible. Let $Z\sim \mathcal{N}(0,\Sigma)$ be an $m-$dimensional multivariate normal random vector, whose covariance is $\Sigma$. Then, there exists a universal constant $C$ such that for all convex $U\subseteq \R^d$,
\[
\Bigl|\mathbb{P}[S\in U]-\mathbb{P}[Z\in U]\Bigr|\le Cm^{\frac14}\sum_{1\le j\le n}\mathbb{E}\Bigl[\bigl\|\Sigma^{-\frac12}Y_j\bigr\|_2^3\Bigr].
\]
\end{theorem}
We will apply Theorem~\ref{thm:berry-esseen} for $U=[-\kappa,\kappa]^m$. While our results still transfer to the ensemble OGP, we restrict our attention to the non-ensemble variant for simplicity. That is, we focus on the case where the set $\mathcal{I}$ appearing in Definition~\ref{def:ogp-set} is $\{0\}$.
\begin{theorem}\label{thm:ogp-universality}
Let $\mathcal{D}$ be a distribution on $\R$ with the property that for $T\sim \mathcal{D}$,
\[
\mathbb{E}\bigl[T\bigr] = 0,\quad \mathbb{E}\bigl[T^2\bigr] = 1;\qquad\text{and}\qquad \mathbb{E}\bigl[T^3\bigr]<\infty.
\]
Fix $\kappa>0$, $\alpha<\alpha_c(\kappa)$, $m\in\mathbb{N}$, $0<\eta<\beta<1$. Then, 
\[
\mathbb{E}_{\mathcal{M}\in\R^{M\times n}:\mathcal{M}_{ij}\sim \mathcal{D},\text{i.i.d.}}\Bigl[\mathcal{S}_\kappa\bigl(\beta,\eta,m,\alpha,\{0\}\bigr)\Bigr] \le \mathbb{E}_{\mathcal{M}\in\R^{M\times n}:\mathcal{M}_{ij}\sim \mathcal{N}(0,1),\text{i.i.d.}}\Bigl[\mathcal{S}_\kappa\bigl(\beta,\eta,m,\alpha,\{0\}\bigr)\Bigr]e^{O(\sqrt{n})}.
\]
\end{theorem}
The proof of Theorem~\ref{thm:ogp-universality} is provided in Section~\ref{pf:ogp-universality}. Hence, if $0<\eta<\beta<1$ and $m\in\mathbb{N}$ are such that \[
\mathbb{E}_{\mathcal{M}\in\R^{M\times n}:\mathcal{M}_{ij}\sim \mathcal{N}(0,1),\text{i.i.d.}}\Bigl[\mathcal{S}_\kappa\bigl(\beta,\eta,m,\alpha,\{0\}\bigr)\Bigr]=\exp\bigl(-\Theta(n)\bigr),
\]then \[
\mathbb{E}_{\mathcal{M}\in\R^{M\times n}:\mathcal{M}_{ij}\sim \mathcal{D},\text{i.i.d.}}\Bigl[\mathcal{S}_\kappa\bigl(\beta,\eta,m,\alpha,\{0\}\bigr)\Bigr]=\exp\bigl(-\Theta(n)\bigr).
\]
A particular case of interest is when the (i.i.d.) entries of the disorder matrix is Rademacher. That is, $\mathbb{P}[\mathcal{M}_{ij}=1]=1/2=\mathbb{P}[\mathcal{M}_{ij}=-1]$, i.i.d.\,across $1\le i\le M$ and $1\le j\le n$. In this case, Theorem~\ref{thm:ogp-universality} asserts that the $m-$OGP holds with the exact same parameters appearing in Theorem~\ref{thm:m-OGP-small-kappa}.
\section{Proofs}\label{sec:proofs}
\subsection{Some Auxiliary Results}
Our $2-$OGP and $3-$OGP results for the high $\kappa$ case (namely Theorem~\ref{thm:2-OGP} and Theorem~\ref{thm:3-OGP})  require the following auxiliary result. We remind the reader that $h(\cdot)$ is the binary entropy function logarithm base two. 
\begin{lemma}\label{asm:negativity}
Let
\begin{equation}\label{eq:f}
    f_1(\Delta,\alpha) = 1+h(\Delta) +\alpha\log_2\mathbb{P}\bigl[|Z_1|\le 1,|Z_2|\le 1\bigr],
\end{equation}
where $(Z_1,Z_2)\sim \mathcal{N}\bigl(0,\Delta I + (1-\Delta)\ones\ones^T\bigr)$. Let
\begin{equation}\label{eq:f-2}
f_2(\beta,\alpha) \triangleq 1+h\left(\frac{1-\beta}{2}\right)+\alpha\log_2\mathbb{P}\bigl[|Z_1|\le 1,|Z_2|\le 1\bigr],
\end{equation}
where $(Z_1,Z_2)\sim \mathcal{N}\bigl(0,(1-\beta)I_2+\beta\ones\ones^T\bigr)$; and let
\begin{equation}\label{eq:f-3}
f_3(\beta,\alpha)\triangleq 1+\frac{1-\beta}{2}+h\left(\frac{1-\beta}{2}\right)+\frac{1+\beta}{2}h\left(\frac{1-\beta}{2(1+\beta)}\right)+\alpha\log_2\mathbb{P}\bigl[|Z_i|\le 1,1\le i\le 3\bigr],
\end{equation}
where $(Z_1,Z_2,Z_3)\sim \mathcal{N}\bigl(0,(1-\beta)I_3+\beta \ones\ones^T\bigr)$. Then, the following holds.
\begin{itemize}
    \item[(a)] Let $S_1(\alpha) = \bigl\{\Delta\in[0.00001,0.1]:f_1(\Delta,\alpha)<0\bigr\}$. Then, $S_1(1.77)\ne\varnothing$. Hence,  $S_1(\alpha)\ne\varnothing$, $\forall \alpha\ge 1.77$.
    \item[(b)] Let $S_2(\alpha) = \bigl\{\beta\in(0,1):f_2(\beta,\alpha)<0\bigr\}$. Then, $S_2(1.71)\ne\varnothing$. Hence,  $S_2(\alpha)\ne\varnothing$, $\forall \alpha\ge 1.710$.
    \item[(c)] Let $S_3(\alpha) = \bigl\{\beta\in(0,1):f_3(\beta,\alpha)<0\bigr\}$. Then, $S_3(1.667)\ne\varnothing$. Hence, $S_3(\alpha)\ne\varnothing$, $\forall\alpha\ge 1.667$.
\end{itemize}
\end{lemma}
Lemma~\ref{asm:negativity} is established numerically using MATLAB's \texttt{mvncdf} function to evaluate the probability term. The accompanying code is provided in Appendix~\ref{appendix:matlab}.

We next record two useful auxiliary results regarding bivariate Gaussian random variables. These will later be useful in Section~\ref{sec:proof-of-kr} to prove the stability of Kim-Roche algorithm. Our first lemma to that end pertains to the quadrant probabilities for the bivariate normal distribution. 
\begin{lemma}\label{lemma:gaussian-orthant}
Let $(X,Y)$ be a bivariate normal random variable with
\[
(X,Y) \distr \mathcal{N}\left(\begin{bmatrix} 0\\ 0\end{bmatrix},\begin{bmatrix}1&\rho \\ \rho & 1\end{bmatrix}\right).
\]
Then,
\[
\mathbb{P}\bigl(X\ge 0,Y\ge 0\bigr)  = \frac{1}{4} + \frac{1}{2\pi}\sin^{-1}(\rho). 
\]
\end{lemma}
Lemma~\ref{lemma:gaussian-orthant} is quite well-known; a proof is provided below for completeness.
\begin{proof}[Proof of Lemma~\ref{lemma:gaussian-orthant}]
Note that the pair $(X,Y)$ has bivariate normal distribution with parameter $\rho$. Next, define 
\[
Z \triangleq \frac{Y-\rho X}{\sqrt{1-\rho^2}}.
\]
Clearly $X$ and $Z$ are i.i.d. standard normals. Let $\zeta\triangleq -\frac{\rho}{\sqrt{1-\rho^2}}$. Observe that
\begin{align*}
\mathbb{P}\bigl(X\ge 0,Y\ge 0\bigr) &=\mathbb{P}\left(X\ge 0,Z\ge -\frac{\rho}{\sqrt{1-\rho^2}}X\right)\\
&=\int_{x=0}^\infty \int_{z=\zeta x}^\infty \frac{1}{2\pi}\exp\left(-\frac{x^2+z^2}{2}\right)\; dx\; dz \\
&=\int_{\theta=\tan^{-1}(\zeta)}^{\frac{\pi}{2}} \int_{r=0}^\infty \frac{1}{2\pi}\exp\left(-\frac{r^2}{2}\right)r\;dr\;d\theta \\
&=\frac{1}{2\pi}\int_{\theta=\tan^{-1}(\zeta)}^{\frac{\pi}{2}}\; d\theta \\
&=\frac{1}{4}-\frac{1}{2\pi}\tan^{-1}(\zeta)\\
&=\frac14+\frac{1}{2\pi}\sin^{-1}\rho.
\end{align*}
Here, the second line uses the independence of $X$ and $Z$; the third line is obtained upon passing to polar coordinates; and the last line follows from the fact $\tan^{-1}$ is an odd function, and that if $\tan\theta =\frac{\rho}{\sqrt{1-\rho^2}}$ then $\sin\theta =\rho$.
\end{proof}
Our next lemma is as follows.
\begin{lemma}\label{lemma:bivariate-cond-quadrant}
Let $Z_1,Z_2\distr\mathcal{N}(0,1)$ where $(Z_1,Z_2)$ is a bivariate normal with parameter $\rho$: $\mathbb{E}[Z_1Z_2]=\rho$. Then
\[
\mathbb{E}\bigl[Z_1|Z_2\ge 0\bigr] = \rho \sqrt{\frac{2}{\pi}}.
\]
\end{lemma}
\begin{proof}[Proof of Lemma~\ref{lemma:bivariate-cond-quadrant}]
Note that $Z_3:=\frac{Z_1-\rho Z_2}{\sqrt{1-\rho^2}}$ is a standard normal, independent of $Z_2$ (as $\mathbb{E}[Z_2Z_3]=0$ and $(Z_2,Z_3)$ is also bivariate normal). Hence
\begin{align*}
    \mathbb{E}\bigl[Z_1|Z_2\ge 0\bigr]&= \mathbb{E}\bigl[\rho Z_2+\sqrt{1-\rho^2}{Z_3}|Z_2\ge 0\bigr]\\
    &=\rho\mathbb{E}\bigl[Z_2|Z_2\ge 0\bigr]\\
    &=\rho \frac{\mathbb{E}\bigl[Z_2\ind\{Z_2\ge 0\}\bigr]}{\mathbb{P}(Z_2\ge 0)} \\
    &= 2\rho\int_{0}^\infty\frac{1}{\sqrt{2\pi}}x\exp\left(-\frac{x^2}{2}\right)\; dx\\
    &=\rho\sqrt{\frac{2}{\pi}},
\end{align*}
yielding Lemma~\ref{lemma:bivariate-cond-quadrant}.
\end{proof}
\subsection{Proof of Theorem~\ref{thm:3-OGP}}\label{sec:proof-3-OGP}
Our proof is based on the \emph{first moment method}. We begin by observing the following monotonicity: if 
\[
\alpha\le \alpha'\le \alpha_c(1) = -\frac{1}{\log_2\mathbb{P}\bigl(|\mathcal{N}(0,1)|\le 1\bigr)}\approx 1.8159,
\]
then
\[
\mathbb{P}\Bigl[\mathcal{S}(\beta,\eta,3,\alpha',\mathcal{I}) \ne \varnothing\Bigr]\le \mathbb{P}\Bigl[\mathcal{S}(\beta,\eta,3,\alpha,\mathcal{I}) \ne \varnothing\Bigr].
\]
For this reason, it suffices to consider $\alpha=1.667$. \paragraph{Counting term.} Let $0<\eta<\beta<1$. We first count the number of triples $(\sigma^{(i)}:1\le i\le 3)$ in $\bincube$ subject to the overlap condition. (In what follows, we omit floor/ceiling operations to keep our exposition clean.)
\begin{lemma}\label{lemma:n-equilateral}
Let $0<\eta<\beta<1$ be fixed. Denote by $M(\beta,\eta)$ the number of triples $(\sigma^{(i)}:1\le i\le 3)$ with $\sigma^{(i)}\in \bincube$ subject to the condition
\[
\beta-\eta\le \Overlap{\sigma^{(i)}}{\sigma^{(j)}}=\frac1n\ip{\sigma^{(i)}}{\sigma^{(j)}} \le \beta,\qquad 1\le i<j\le 3.
\]
Then, 
\begin{equation}\label{eq:M-up-bd}
    M(\beta,\eta)\le \exp_2\Bigl(n\varphi_{{\rm Count}}(\beta,\eta) + O(\log_2 n)\Bigr),
\end{equation}
where
\begin{equation}\label{eq:phi-count}
    \varphi_{\rm Count}(\beta,\eta) = 1+h\left(\frac{1-\beta+\eta}{2}\right) +\frac{1-\beta+\eta}{2}+ \frac{1+\beta}{2}h\left(\frac{1-\beta+2\eta}{2(1+\beta)}\right). 
\end{equation}
In particular, for any fixed $\beta$, the map $\eta\mapsto \varphi_{\rm Count}(\beta,\eta)$ is continuous at $\eta=0$.
\end{lemma}
\begin{proof}[Proof of Lemma~\ref{lemma:n-equilateral}] For $\sigma^{(k)}\in \bincube$, denote its $i^{\rm th}$ coordinate ($1\le i\le n$) by $\sigma^{(k)}_i$. 

Clearly, there are $2^n$ ways of choosing  $\sigma^{(1)}$. Having fixed $\sigma^{(1)}$, there are $\binom{n}{n\frac{1-\rho}{2}}$ ways of choosing $\sigma^{(2)}$ with $\Overlap{\sigma^{(1)}}{\sigma^{(2)}}= \rho \in[\beta-\eta,\beta]$. Assume now that both $\sigma^{(1)}$ and $\sigma^{(2)}$ are fixed, and define $I\subset[n]$ with $|I|=n\frac{1-\rho}{2}$ as
\[
I\triangleq \Bigl\{1\le i\le n:\sigma^{(1)}_i \ne \sigma^{(2)}_i\Bigr\},
\]
and let $I^c\triangleq [n]\setminus I$ with $|I^c|=n\frac{1+\rho}{2}$. In particular, $\sigma^{(1)}$ and $\sigma^{(2)}$ agree on coordinates in $I^c$ and disagree on coordinates in $I$. Having fixed $\sigma^{(1)}$ and $\sigma^{(2)}$, now let $N_3(\rho)$ denote  the number of all admissible $\sigma^{(3)}$ satisfying the inner product condition (with $\sigma^{(1)}$ and $\sigma^{(2)}$). Then, it is evident  that
\begin{equation}\label{eq:up-bd-M-beta-eta}
M(\beta,\eta) = 2^n\sum_{\rho:\beta-\eta\le \rho\le\beta,\rho n\in\mathbb{N}}\binom{n}{n\frac{1-\rho}{2}}N_3(\rho).
\end{equation}
Now, suppose that
\[
t_1\triangleq \Bigl|\Bigl\{i\in I:\sigma^{(2)}_i = \sigma^{(3)}_i\Bigr\}\Bigr|\qquad\text{and}\qquad t_2\triangleq \Bigl|\Bigl\{i\in I^c:\sigma^{(2)}_i\ne \sigma^{(3)}_i\Bigr\}\Bigr|.
\]
Then
\[
d_H\bigl(\sigma^{(1)},\sigma^{(3)}\bigr)=t_1+t_2\quad\text{and}\quad 
    d_H\bigl(\sigma^{(2)},\sigma^{(3)}\bigr)=n\frac{1-\rho}{2}-t_1+t_2.
\]
Next, observe that using  the condition on $\Overlap{\sigma^{(1)}}{\sigma^{(3)}}$  and $\Overlap{\sigma^{(2)}}{\sigma^{(3)}}$, we  arrive at
\[
n\frac{1-\beta}{2} \le d_H\bigl(\sigma^{(1)},\sigma^{(3)}\bigr) = t_1+t_2 \le n\frac{1-\beta+\eta}{2},
\]
and
\[
n\frac{1-\beta}{2} \le d_H\bigl(\sigma^{(2)},\sigma^{(3)}\bigr)  = t_2-t_1 +n\frac{1-\rho}{2} \le n\frac{1-\beta+\eta}{2}.
\]
We thus arrive  at
\[
n\frac{1-\beta}{2} \le t_1+t_2\le n\frac{1-\beta+\eta}{2}\quad\text{and}\quad n\frac{\rho-\beta}{2} \le t_2-t_1 \le n\frac{\rho-\beta+\eta}{2}.
\]
This yields the following lower and upper bounds on $t_1,t_2$:
\begin{align}
n\frac{1-\rho-\eta}{4}&\le t_1\le n\frac{1-\rho+\eta}{4} \label{eq:range-t1}\\
    n\frac{1+\rho-2\beta}{4}&\le t_2\le n\frac{1+\rho -2\beta+2\eta}{4}\label{eq:range-t2}.
\end{align}
Define now the rectangle
\[
\mathcal{T}\triangleq \left[\frac{1-\rho-\eta}{4},\frac{1-\rho+\eta}{4}\right]\times\left[ \frac{1+\rho-2\beta}{4},\frac{1+\rho -2\beta+2\eta}{4}\right].
\]
Note that (a) for $\eta$ small enough, $\mathcal{T}\subset(0,\infty)^2$; and (b) the set of all admissible $(t_1,t_2)\in\mathbb{N}^2$ pairs are precisely the set of all lattice points in the box $n\mathcal{T}$. Having fixed $\sigma^{(1)}$ and $\sigma^{(2)}$, the number $N_3(\rho)$ of admissible $\sigma^{(3)}$ then computes as
\begin{equation}\label{num-n3}
   N_3(\rho) = \sum_{(t_1,t_2)\in \mathbb{N}^2\cap n\mathcal{T}}\binom{n\frac{1-\rho}{2}}{t_1}\binom{n\frac{1+\rho}{2}}{t_2}.
\end{equation}
Note that using the fact $\binom{n}{\alpha}$ is maximized for $\alpha=\lfloor n/2\rfloor$, we obtain
\begin{equation}\label{eq:up-bd-prod-binom}
\binom{n\frac{1-\rho}{2}}{t_1}\binom{n\frac{1+\rho}{2}}{t_2} \le \binom{n\frac{1-\rho}{2}}{n\frac{1-\rho}{4}}\binom{n\frac{1+\rho}{2}}{n\frac{1+\rho-2\beta+2\eta}{4}},\quad \forall (t_1,t_2)\in \mathbb{N}^2\cap n\mathcal{T}.
\end{equation}
Now, we use the well-known asymptotic on binomial coefficients: for any $r\in(0,1)$, $\binom{n}{nr} = \exp_2(nh(r)+O(\log_2 n))$. Combining this fact together  with~\eqref{num-n3} and~\eqref{eq:up-bd-prod-binom}, we obtain the following upper bound on number of  such $\sigma^{(3)}$:
\begin{equation}\label{eq:sigma-3-up-bd}
   N_3(\rho)\le  \exp_2\left(n\frac{1-\rho}{2} +n\frac{1+\rho}{2}h\left(\frac{1+\rho-2\beta+2\eta}{2(1+\rho)}\right)+O(\log_2 n)\right).
\end{equation}
We next study the argument  $(1+\rho-2\beta+2\eta)/2(1+\rho)$ of the entropy term appearing in~\eqref{eq:sigma-3-up-bd}. Clearly, as $\eta<\beta$, the argument is less than $1/2$. Now, observe that
\begin{align*}
&\frac{1+\rho_1-2\beta+2\eta}{2(1+\rho_1)}<\frac{1+\rho_2-2\beta+2\eta}{2(1+\rho_2)} \\
&\Leftrightarrow \frac12 - \frac{\beta-\eta}{1+\rho_1}<\frac12 - \frac{\beta-\eta}{1+\rho_2}\\&\Leftrightarrow \frac{1}{1+\rho_2}<\frac{1}{1+\rho_1}\Leftrightarrow\rho_1<\rho_2.
\end{align*}
Consequently, using the monotonicity of $h$ in $[0,\frac12]$,
\[
h\left(\frac{1+\rho-2\beta+2\eta}{2(1+\rho)}\right)<h\left(\frac{1-\beta+2\eta}{2(1+\beta)}\right).
\]
Using this and~\eqref{eq:sigma-3-up-bd}, $N_3(\rho)$ is further upper bounded by
\begin{equation}\label{eq:sigma-3-further-bd}
N_3(\rho)\le \exp_2\left(n\frac{1-\beta+\eta}{2} + n\frac{1+\beta}{2} h\left(\frac{1-\beta+2\eta}{2(1+\beta)}\right)+O(\log_2 n)\right).
\end{equation}
Finally, we  have
\begin{equation}\label{eq:sigma-2-up-bd}
\binom{n}{n\frac{1-\rho}{2}}\le \binom{n}{n\frac{1-\beta+\eta}{2}} = \exp_2\left(nh\left(\frac{1-\beta+\eta}{2}\right)+O(\log_2 n)\right).
\end{equation}
Combining~\eqref{eq:up-bd-M-beta-eta},~\eqref{eq:sigma-3-further-bd} and~\eqref{eq:sigma-2-up-bd}, we obtain
\begin{align*}
M(\beta,\eta)&\le \exp_2\left(n+nh\left(\frac{1-\beta+\eta}{2}\right)+n\frac{1-\beta+\eta}{2} + n\frac{1+\beta}{2} h\left(\frac{1-\beta+2\eta}{2(1+\beta)}\right)  +O(\log_2 n)\right) \\
&=\exp_2\Bigl(n\varphi_{\rm  count}(\beta,\eta)+O(\log_2 n)\Bigr),
\end{align*}
yielding~\eqref{eq:M-up-bd}. Since the continuity follows immediately from the continuity of the entropy, the proof of Lemma~\ref{lemma:n-equilateral} is complete.
\end{proof}
\paragraph{Probability term.} Now, fix any $\left(\sigma^{(i)}:1\le i\le 3\right)$ with $\beta-\eta \le \Overlap{\sigma^{(i)}}{\sigma^{(j)}}\le \beta$. More concretely,  let
\begin{equation}\label{eq:eta-i-j}
    \Overlap{\sigma^{(i)}}{\sigma^{(j)}}\triangleq \beta-\eta_{ij},\qquad \text{where}\qquad 0\le \eta_{ij}\le \eta, \quad 1\le i<j\le 3.
\end{equation}
We control the probability term
\begin{equation}\label{eq:prob-prelim-bd}
    \mathbb{P}\left[\exists\tau_1,\tau_2,\tau_3 \in\mathcal{I}:n^{-\frac12}\left|\mathcal{M}_i(\tau_i)\sigma^{(i)}\right|\le\ones,  1\le i\le  3 \right],
\end{equation}
where $\ones\in\R^{M\times 1}$ is the vector of all ones,  and the inequality is coordinate-wise. As a first  step, we  take a union bound over $\mathcal{I}$ to obtain
\begin{align}
  &   \mathbb{P}\left[\exists\tau_1,\tau_2,\tau_3 \in\mathcal{I}:n^{-\frac12}\left|\mathcal{M}_i(\tau_i)\sigma^{(i)}\right|\le\ones,  1\le i\le  3 \right]\nonumber  \\
  &\le  |\mathcal{I}|^3  \max_{\tau_i\in\mathcal{I},1\le i\le 3}\mathbb{P}\left[n^{-\frac12}\left|\mathcal{M}_i(\tau_i)\sigma^{(i)}\right|\le\ones,  1\le i\le  3 \right].\label{eq:prob-prelim-bd-2}
\end{align}
Next, let the first row of $\mathcal{M}_i(\tau_i)$ be $\texttt{R}_i\distr \mathcal{N}(0,I_n)\in\R^n$, $1\le i\le 3$. Using the independence across rows, we have
\begin{equation}\label{eq:prob-prelim-bd-3}
\mathbb{P}\left[n^{-\frac12}\left|\mathcal{M}_i(\tau_i)\sigma^{(i)}\right|\le\ones,  1\le i\le  3 \right] = \mathbb{P}\left[n^{-\frac12}\left|\ip{\texttt{R}_i}{\sigma^{(i)}}\right|\le 1,1\le i\le 3\right]^{\alpha n}.
\end{equation}
To upper bound  the probability appearing in~\eqref{eq:prob-prelim-bd-3}, observe that $\left(n^{-1/2}\ip{\texttt{R}_i}{\sigma^{(i)}}:1\le i\le  3\right)$ is a multivariate normal with each component having zero mean and unit variance. We now compute its covariance matrix $\overline{\Sigma}\in\R^{3\times 3}$. Observe that for $i\ne j$,
\begin{align*}
\overline{\Sigma}_{ij}&=\mathbb{E}\left[n^{-\frac12}\ip{\texttt{R}_i}{\sigma^{(i)}}\cdot n^{-\frac12}\ip{\texttt{R}_j}{\sigma^{(j)}}\right] \\
&=\frac1n\left(\sigma^{(i)}\right)^T\mathbb{E}\bigl[\texttt{R}_i\texttt{R}_j^T\bigr]\sigma^{(j)} \\
&=\cos(\tau_i)\cos(\tau_j)\Overlap{\sigma^{(i)}}{\sigma^{(j)}}\\
&=\cos(\tau_i)\cos(\tau_j)\left(\beta-\eta_{ij}\right),
\end{align*}
where the last line uses~\eqref{eq:eta-i-j}. In order to remove the  dependence  of $\overline{\Sigma}$ on $\tau_i$, we now  employ the  following Gaussian comparison inequality established by Sid{\'a}k~\cite[Corollary~1]{sidak1968multivariate}.
\begin{theorem}\label{thm:sidak}
Let $(X_1,\dots,X_k)\in\R^k$ be a multivariate  normal random vector each of whose coordinates have zero mean and unit variance. Suppose that its covariance matrix $\Sigma\in\R^{k\times k}$ has the following form: there exists $ \lambda_1,\dots,\lambda_k$ ($0\le \lambda_i\le 1$,  $1\le i\le k$) such that for every $1\le i\ne j\le  k$, $\Sigma_{ij}=\lambda_i\lambda_j\rho_{ij}$ where $(\rho_{ij}:1\le i\ne j\le k)$ is some fixed covariance matrix. Fix any $c_1,\dots,c_k>0$, and denote
\[
P(\lambda_1,\dots,\lambda_k) = \mathbb{P}\bigl[|X_1|<c_1,|X_2|<c_2,\dots,|X_k|<c_k\bigr].
\]
Then, $P(\lambda_1,\dots,\lambda_k)$ is a non-decreasing function of each $\lambda_i$, $i=1,2,\dots,k$, $0\le \lambda_i\le 1$. That is,
\[
P(\lambda_1,\lambda_2,\dots,\lambda_k)\le P(1,1,\dots,1).
\]
\end{theorem}
Applying Theorem~\ref{thm:sidak}, we find that
\begin{equation}\label{eq:prob-prelim-bd-4}
\max_{\tau_i\in\mathcal{I},1\le i\le 3}\mathbb{P}\left[n^{-\frac12}\left|\ip{\texttt{R}_i}{\sigma^{(i)}}\right|\le  1,1\le i\le 3\right] \le \mathbb{P}\bigl[|Z_1|\le 1,|Z_2|\le 1,|Z_3|\le 1\bigr]
\end{equation}
where
\begin{equation}\label{eq:dist-z1-z2-z3}
    (Z_1,Z_2,Z_3)\distr \mathcal{N}\left(\begin{bmatrix}0\\0\\0\end{bmatrix},\begin{bmatrix}1&\beta-\eta_{12}&\beta-\eta_{13}\\\beta-\eta_{12}&1&\beta-\eta_{23}\\\beta-\eta_{13}&\beta-\eta_{23}&1\end{bmatrix}\right).
\end{equation}
Now, note that
\begin{equation}\label{eq:matrix-sigma}
\Sigma\triangleq \begin{bmatrix}1&\beta-\eta_{12}&\beta-\eta_{13}\\\beta-\eta_{12}&1&\beta-\eta_{23}\\\beta-\eta_{13}&\beta-\eta_{23}&1\end{bmatrix} = (1-\beta)I_3+\beta \ones\ones^T+E,
\end{equation}
where $E\in\R^{3\times 3}$  has zero diagonal entries, and $E_{ij}=-\eta_{ij}$ for $1\le i\ne j\le 3$. In particular, $\|E\|_F\le \eta\sqrt{6}$. Since the eigenvalues of $(1-\beta)I_3+\beta\ones\ones^T$ are $1+2\beta$ (with multiplicity one) and $1-\beta$ (with multiplicity two), it follows that provided $\eta<(1-\beta)/\sqrt{6}$, the covariance matrix $\Sigma$ appearing in~\eqref{eq:matrix-sigma} is invertible. We assume that this is indeed the case from this point on. 

Define
\begin{equation}\label{eq:phi-prob}
    \varphi_{\rm Prob}\bigl(\beta,\eta_{12},\eta_{13},\eta_{23}\bigr)\triangleq \mathbb{P}\bigl(|Z_1|\le 1,|Z_2|\le 1,|Z_3|\le 1\bigr)
\end{equation}
where $(Z_1,Z_2,Z_3)$ has distribution in~\eqref{eq:dist-z1-z2-z3}. We now combine~\eqref{eq:prob-prelim-bd},~\eqref{eq:prob-prelim-bd-2},~\eqref{eq:prob-prelim-bd-3},~\eqref{eq:prob-prelim-bd-4}, and~\eqref{eq:phi-prob} to arrive at
\begin{equation}\label{eq:prob-prelim-bd-5}
      \mathbb{P}\left[\exists\tau_1,\tau_2,\tau_3 \in\mathcal{I}:n^{-\frac12}\left|\mathcal{M}_i(\tau_i)\sigma^{(i)}\right|\le\ones,  1\le i\le  3 \right]\le |\mathcal{I}|^3 \varphi_{\rm Prob}\bigl(\beta,\eta_{12},\eta_{13},\eta_{23}\bigr)^{\alpha n}.
\end{equation}
Our next technical result pertains to $\varphi_{\rm Prob}$.
\begin{lemma}\label{lemma:prob-fnc-cts}
Fix any $\beta\in(0,1)$. Then the map
\[
\bigl(\eta_{12},\eta_{13},\eta_{23}\bigr)\mapsto \log_2\varphi_{\rm Prob}\bigl(\beta,\eta_{12},\eta_{13},\eta_{23}\bigr)
\]
(from $\R^3$ to $\R$) is continuous at $(0,0,0)$. 
\end{lemma}
\begin{proof}[Proof of Lemma~\ref{lemma:prob-fnc-cts}] Define a sequence $(\boldsymbol{\zeta_k})_{k\ge 1}$ such that
\[
\boldsymbol{\zeta_k}=\bigl(\eta_{12}(k),\eta_{13}(k),\eta_{23}(k)\bigr)\in\R^3\quad\text{and}\quad \lim_{k\to \infty}\boldsymbol{\zeta_k}=(0,0,0).
\]
Let $\Sigma_k\in\R^{3\times 3}$ be the matrix $\Sigma$ appearing in~\eqref{eq:matrix-sigma} with parameters $\boldsymbol{\zeta_k}$. Let $\mathbf{v}=(x,y,z)$, and define functions
\[
f_k(\mathbf{v})\triangleq \exp\left(-\frac12\mathbf{v}^T\Sigma_k^{-1}\mathbf{v}\right).
\]
Moreover, let \[
\Sigma_\infty\triangleq(1-\beta)I+\beta\ones\ones^T\quad\text{and}\quad f_\infty(\mathbf{v})\triangleq\exp\left(-\frac12\mathbf{v}^T\Sigma_\infty^{-1} \mathbf{v}\right).
\]
Note that, $|f_k(\mathbf{v})|\le 1$ for every $\mathbf{v}\in\R^3$ (as long  as $\Sigma_k$ is positive definite). Moreover, we have the pointwise convergence: \[
\lim_{k\to\infty}f_k(\mathbf{v})=f_\infty(\mathbf{v})\qquad\text{for all}\qquad \mathbf{v}\in\R^3.
\]
Therefore, by the dominated convergence theorem
\begin{equation}\label{eq:integral-conv}
    \lim_{k\to\infty}\int_{[-1,1]^3}f_k(\mathbf{v})  \;d\mathbf{v} = \int_{[-1,1]^3}f_\infty(\mathbf{v})  \;d\mathbf{v}.
\end{equation}
Moreover,
\begin{equation}\label{eq:det-conv}
   \lim_{k\to\infty} \bigl|\Sigma_k\bigr|  =\bigl|\Sigma_\infty\bigr|.
\end{equation}
Finally, 
\[
\log_2\varphi_{\rm Prob}\bigl(\beta,\boldsymbol{\zeta_k}\bigr)=-\frac32\log_2(2\pi)-\frac12\log_2\bigl|\Sigma_k\bigr|+\log_2\left(\int_{[-1,1]^3}f_k(\mathbf{v})\;d\mathbf{v}\right).
\]
Since $x\mapsto \log_2 x$ is continuous, we obtain by combining~\eqref{eq:integral-conv}  and~\eqref{eq:det-conv} that
\[
\lim_{k\to\infty}\log_2\varphi_{\rm Prob}\bigl(\beta,\boldsymbol{\zeta_k}\bigr) = \log_2 \varphi_{\rm Prob}(\beta,0,0,0).
\]
Since the sequence $(\boldsymbol{\zeta_k})_{k\ge 1}$ is arbitrary, the proof of Lemma~\ref{lemma:prob-fnc-cts} is complete. 
\end{proof}
\paragraph{Choice of $\beta,\eta$.} In the remainder, let $\alpha^*=1.667$. Notice next that
\[
f_3(\beta,\alpha) = \varphi_{\rm Count}\bigl(\beta,0\bigr)+\alpha\cdot\log_2\varphi_{\rm Prob}\bigl(\beta,0,0,0\bigr),
\]
where $f_3$ is defined in~\eqref{eq:f-3}, $\varphi_{\rm Count}$ is defined in~\eqref{eq:phi-count}; and $\varphi_{\rm Prob}$ is defined in~\eqref{eq:phi-prob}. Let $\beta^*$ be such that
\begin{equation}\label{eq:f-star}
f^*\triangleq f_3\bigl(\beta^*,\alpha^*\bigr)=\inf_{\beta\in[0,1]} f_3(\beta,\alpha^*).
\end{equation}
Since $S_3(\alpha^*)\ne\varnothing$ by Lemma~\ref{asm:negativity}, it follows $f^*<0$. Having fixed $\beta^*$, let 
\begin{equation}\label{eq:epsilon-star-and-c-star}
    \epsilon^*\triangleq -f^*/8\alpha^*>0\qquad\text{and}\qquad c^*\triangleq -f^*/24.
\end{equation}
Using Lemma~\ref{lemma:prob-fnc-cts}, it follows that there exists a $\delta_1^*\triangleq \delta_1^*\bigl(\beta^*,\epsilon^*\bigr)>0$, such that
\begin{equation}\label{eq:sup-delta}
\sup_{\substack{(\eta_{12},\eta_{13},\eta_{23})\\|\eta_{ij}|<\delta_1^*, 1\le i<j\le 3}}\Bigl|\log_2\varphi_{\rm Prob}(\beta^*,\eta_{12},\eta_{13},\eta_{23}) - \log_2\varphi_{\rm Prob}(\beta^*,0,0,0)\Bigr|<\epsilon^*.
\end{equation}
Take $\eta<\delta_1^*$, ensuring $0\le \eta_{ij}\le \eta<\delta_1^*$, $1\le i<j\le 3$. Using Markov's inequality, we have
\[
\mathbb{P}\bigl(\mathcal{S}(\beta^*,\eta,3,\alpha^*,\mathcal{I})\ne \varnothing\bigr) = \mathbb{P}\bigl(\bigl|\mathcal{S}(\beta^*,\eta,3,\alpha^*,\mathcal{I})\bigr|\ge 1\bigr)\le \mathbb{E}\bigl[\bigl|\mathcal{S}(\beta^*,\eta,3,\alpha^*,\mathcal{I})\bigr|\bigr].
\]
We now combine the counting bound~\eqref{eq:M-up-bd}, the probability bound~\eqref{eq:prob-prelim-bd-5} and the bound~\eqref{eq:sup-delta} to upper bound $\mathbb{E}\bigl[\bigl|\mathcal{S}(\beta^*,\eta,3,\alpha^*,\mathcal{I})\bigr|\bigr]$:
\begin{align}
\mathbb{E}\bigl[\bigl|\mathcal{S}(\beta^*,\eta,3,\alpha^*,\mathcal{I})\bigr|\bigr]&\le  \exp_2\Bigl(n\varphi_{\rm Count}(\beta^*,\eta)+n\alpha^* \log_2\varphi_{\rm Prob}(\beta^*,0,0,0)+n\alpha^* \epsilon^*+3\log_2|\mathcal{I}|+O(\log_2 n)\Bigr)\nonumber\\
&\le \exp_2\left(n\left(\varphi_{\rm Count}(\beta^*,\eta) + \alpha^*\log_2\varphi_{\rm Prob}(\beta^*,0,0,0)-\frac{f^*}{4}+O\left(\frac{\log_2 n}{n}\right)\right)\right)\label{eq:towards-E-prelim},
\end{align}
where~\eqref{eq:towards-E-prelim} uses $\max\{\alpha^*\epsilon^*,3\log_2|\mathcal{I}|\}\le -f^*/8$ which  follows from~\eqref{eq:epsilon-star-and-c-star}. Now,  using the continuity of $\eta\mapsto \varphi_{\rm Count}(\beta^*,\eta)$ at $\eta=0$, it follows that there is a $\delta_2^*\triangleq \delta_2^*(\beta^*)>0$ such that
\begin{equation}\label{eq:towards-E-prelim-2}
|\eta|<\delta_2^* \implies \varphi_{\rm Count}(\beta^*,\eta)<\varphi_{\rm Count}(\beta^*,0)-\frac{f^*}{4}.
\end{equation}
Finally, we let 
\begin{equation}\label{eq:eta-star}
\eta^* = \frac12\min\bigl\{\delta_1^*,\delta_2^*\bigr\}.
\end{equation}
With this choice of $\eta^*$, we have by using~\eqref{eq:towards-E-prelim} and~\eqref{eq:towards-E-prelim-2} that
\begin{align*}
    \mathbb{E}\bigl[\bigl|\mathcal{S}(\beta^*,\eta^*,3,\alpha^*,\mathcal{I})\bigr|\bigr]&\le \exp_2\left(n\left(\varphi_{\rm Count}(\beta^*,\eta^*) + \alpha^*\log_2\varphi_{\rm Prob}(\beta^*,0,0,0)-\frac{f^*}{4}+O\left(\frac{\log_2 n}{n}\right)\right)\right)\\
    &\le \exp_2\left(n\left(\varphi_{\rm Count}(\beta^*,0)+\alpha^*\log_2\varphi_{\rm Prob}(\beta^*,0,0,0)-\frac{f^*}{2}+O\left(\frac{\log_2 n}{n}\right)\right)\right)\\
    &\le \exp_2\left(n\left(\frac{f^*}{2}+O\left(\frac{\log_2 n}{n}\right)\right)\right)\\
    &=\exp_2\Bigl(-\Theta(n)\Bigr),
\end{align*}
using the fact per~\eqref{eq:f-star} that $f^*<0$. This completes the proof of Theorem~\ref{thm:3-OGP}.
\subsection{Proof of Theorem~\ref{thm:m-OGP-small-kappa}}\label{sec:pf-m-OGP}
Fix $\kappa>0$. We start by observing the following obvious monotonicity property: for any fixed $0<\eta<\beta<1$, $m\in\mathbb{N}$, $\mathcal{I}\subset[0,\pi/2]$, and $\alpha\le \alpha'$; we have
\[
\mathbb{P}\Bigl[\mathcal{S}_\kappa(\beta,\eta,m,\alpha',\mathcal{I})\ne \varnothing\Bigr]\le 
\mathbb{P}\Bigl[\mathcal{S}_\kappa(\beta,\eta,m,\alpha,\mathcal{I})\ne \varnothing\Bigr].
\]
For this reason, it suffices to establish the result for $\alpha =
\alpha_{{\rm OGP}}(\kappa)=10\kappa^2\log \frac1\kappa$. We will do so by using the \emph{first moment method}: note that by Markov's inequality,
\begin{equation}\label{eq:m-ogp-markov}
\mathbb{P}\Bigl[\mathcal{S}_\kappa(\beta,\eta,m,\alpha,\mathcal{I})\ne \varnothing\Bigr]=\mathbb{P}\Bigl[\Bigl|\mathcal{S}_\kappa(\beta,\eta,m,\alpha,\mathcal{I})\Bigr|\ge 1\Bigr]\le \mathbb{E}\Bigl[\Bigl|\mathcal{S}_\kappa(\beta,\eta,m,\alpha,\mathcal{I})\Bigr|\Bigr].
\end{equation}
We now study $\mathbb{E}\Bigl[\Bigl|\mathcal{S}_\kappa(\beta,\eta,m,\alpha,\mathcal{I})\Bigr|\Bigr]$.
\paragraph{Counting term.} Fix any $m\in\mathbb{N}$, and $0<\eta<\beta<1$. We upper bound the number $M(m,\beta,\eta)$ of the $m-$tuples $\bigl(\sigma^{(i)}:1\le i\le m)$, $\sigma^{(i)}\in\bincube$, subject to the constraint $\beta-\eta \le n^{-1}\ip{\sigma^{(i)}}{\sigma^{(j)}}\le \beta$ for $1\le i<j\le m$. 
\begin{lemma}\label{lemma:m-ogp-counting}
\begin{equation}\label{eq:m-ogp-counting}
    M(m,\beta,\eta)\le \exp_2\left(n + n(m-1)h\left(\frac{1-\beta+\eta}{2}\right)+O(\log_2 n)\right).
\end{equation}
\end{lemma}
\begin{proof}[Proof of Lemma~\ref{lemma:m-ogp-counting}]
Note that for any $\sigma,\sigma'\in\bincube$, $\ip{\sigma}{\sigma'}=n-2d_H(\sigma,\sigma')$. There are $2^n$ choices for $\sigma^{(1)}$. Having chosen a $\sigma^{(1)}$; any $\sigma^{(i)}$, $2\le i\le m$, can be chosen in
\[
\sum_{\substack{\rho:\frac{1-\beta}{2}\le \rho\le \frac{1-\beta+\eta}{2} \\ \rho n\in\mathbb{N}}}\binom{n}{n\rho}\le \binom{n}{n\frac{1-\beta+\eta}{2}}n^{O(1)},
\]
different ways,
subject to the constraint that $\beta-\eta\le n^{-1}\ip{\sigma^{(1)}}{\sigma^{(i)}}\le \beta$. For any $\rho\in(0,1)$, $\binom{n}{n\rho}=\exp_2 \bigl(nh(\rho)+O\bigl(\log_2 n\bigr)\bigr)$ by Stirling's approximation. Combining these, and recalling $m=O(1)$ (as $n\to\infty$), we obtain~\eqref{eq:m-ogp-counting}.
\end{proof}
\paragraph{Probability term.} Fix any $\bigl(\sigma^{(i)}:1\le i\le m\bigr)$ with the pairwise overlaps
\[
n^{-1}\ip{\sigma^{(i)}}{\sigma^{(j)}} = \beta - \eta_{ij},\qquad 1\le i<j\le m.
\]
Evidently, $\eta_{ij}\ge 0$ for all $i<j$. Moreover, if we set \[\boldsymbol{\eta}=\bigl(\eta_{ij}:1\le i<j\le m\bigr)\in\R^{m(m-1)/2},
\]
then,
\[
\|\boldsymbol{\eta}\|_\infty \le \eta.
\]
\begin{lemma}\label{lemma:m-ogp-prob}
Let $\Sigma(\boldsymbol{\eta})\in\R^{m\times m}$ be a matrix with the property that 
\begin{itemize}
    \item[(a)] $\bigl(\Sigma(\boldsymbol{\eta})\bigr)_{ii}=1$ for $1\le i\le m$.
    \item[(b)]  $\bigl(\Sigma(\boldsymbol{\eta})\bigr)_{ij}=\bigl(\Sigma(\boldsymbol{\eta})\bigr)_{ji}=\beta-\eta_{ij}$ for every $1\le i<j\le m$.
\end{itemize}
Then,
\begin{itemize}
    \item[(i)] $\Sigma(\boldsymbol{\eta})$ is positive definite (PD) if $\eta<\frac{1-\beta}{m}$.
    \item[(ii)] Assume $\Sigma(\boldsymbol{\eta})$ is PD. Let $(Z_1,Z_2,\dots,Z_m)\sim \mathcal{N}\bigl(0,\Sigma(\boldsymbol{\eta})\bigr)$ be a multivariate normal random vector, and define

\begin{equation}\label{eq:m-ogp-phi-prob}
   \varphi_{\rm Prob}\bigl(\beta,\boldsymbol{\eta},\kappa\bigr) = \mathbb{P}\Bigl[|Z_i|\le \kappa,1\le i\le m\Bigr].
\end{equation}
Then, 
\begin{equation}\label{eq:m-ogp-prob}
   \mathbb{P}\Bigl[\exists\tau_i \in\mathcal{I}, 1\le i\le m:\left|\mathcal{M}_i(\tau_i)\sigma^{(i)}\right|\le(\kappa\sqrt{n})\ones,  1\le i\le  m \Bigr]\le |\mathcal{I}|^m \varphi_{\rm Prob}\bigl(\beta,\boldsymbol{\eta},\kappa\bigr)^{\alpha n}.
\end{equation}
\item[(iii)] We have
\begin{equation}\label{eq:m-ogp-phi-prob-upper}
    \varphi_{{\rm Prob}}\bigl(\beta,\boldsymbol{\eta},\kappa\bigr)\le \bigl(2\pi\bigr)^{-\frac{m}{2}}\bigl|\Sigma(\boldsymbol{\eta})\bigr|^{-\frac12}\bigl(2\kappa\bigr)^{m}.
\end{equation}
\end{itemize}
\end{lemma}
\begin{proof}[Proof of Lemma~\ref{lemma:m-ogp-prob}]
\begin{itemize}
    \item[(i)] Note that $\Sigma(\boldsymbol{\eta}) = (1-\beta)I+\beta\ones\ones^T+E$, where $E\in\R^{m\times m}$ with $0\le E_{ij}\le \eta$ for $1\le i<j\le m$. In particular, $\|E\|_2\le \|E\|_F\le \eta m$. Noting that the smallest eigenvalue of $(1-\beta)I+\beta \ones \ones^T$ is $1-\beta$, the result follows.
    \item[(ii)] The result follows by taking a union bound over $\mathcal{I}$; and then applying the Gaussian comparison inequality, Theorem~\ref{thm:sidak}, in the exact same way as in the proof of Theorem~\ref{thm:3-OGP}.
    \item[(iii)] Recall that the multivariate normal density for $(Z_i:1\le i\le m)$ is given by
    \[
    f(z_1,\dots,z_m)=\bigl(2\pi\bigr)^{-\frac{m}{2}}\bigl|\Sigma(\boldsymbol{\eta})\bigr|^{-\frac12}\exp\left(-\frac12\boldsymbol{z}^T \Sigma(\boldsymbol{\eta})\boldsymbol{z}\right),
    \]
    where $\boldsymbol{z}=(z_i:1\le i\le m)$, and thus
    \[
    \mathbb{P}\Bigl[(Z_i:1\le i\le m)\in[-\kappa,\kappa]^m\Bigr]\le \bigl(2\pi\bigr)^{-\frac{m}{2}}\bigl|\Sigma(\boldsymbol{\eta})\bigr|^{-\frac12}\bigl(2\kappa\bigr)^m;
    \]
    using the fact
    \[
    \exp\left(-\frac12\boldsymbol{z}^T \Sigma(\boldsymbol{\eta})\boldsymbol{z}\right)\le 1,
    \]
    for every $\boldsymbol{z}\in\R^m$ as $\Sigma(\boldsymbol{\eta})$ is PD.
\end{itemize}
\end{proof}

\paragraph{Upper bounding the expectation.} Assume $0<\eta<\beta<1$ and $m\in\mathbb{N}$ are fixed as $n\to\infty$; and that $\eta$ is small enough, so that $\Sigma(\boldsymbol{\eta})$ is positive definite. (We will tune $\eta$ eventually.)  Combining the counting bound~\eqref{eq:m-ogp-counting} arising from Lemma~\ref{lemma:m-ogp-counting}; the probability bounds~\eqref{eq:m-ogp-prob} and~\eqref{eq:m-ogp-phi-prob-upper} arising from Lemma~\ref{lemma:m-ogp-prob}; and $|\mathcal{I}|\le 2^{cn}$, we upper bound the expectation by
\begin{align}
   \mathbb{E}\Bigl[\Bigl|\mathcal{S}_\kappa(\beta,\eta,m,\alpha,\mathcal{I})\Bigr|\Bigr]&\le \exp_2\left(n+n(m-1)h\left(\frac{1-\beta+\eta}{2}\right)+cmn-\frac{m\alpha n}{2}\log_2(2\pi)\right.\nonumber\\
   &\Bigl.+m\alpha n\log_2(2\kappa) - \frac{\alpha n}{2}\inf_{\boldsymbol{\eta}:\|\boldsymbol{\eta}\|_\infty\le \eta}\log_2 \bigl|\Sigma(\boldsymbol{\eta})\bigr| +O(\log_2 n)\Bigr) \nonumber\\
   &\le \exp_2\Bigl(n \Psi(c,m,\beta,\eta,\kappa) + O\bigl(\log_2 n\bigr)\Bigr)\label{eq:exp-upbd},
\end{align}
where
\begin{align}
    \Psi(c,\beta,\eta,m,\alpha) &\triangleq 1+cm+mh\left(\frac{1-\beta+\eta}{2}\right) - \frac{\alpha m}{2} \log_2(2\pi)\nonumber\\
    & +\alpha m\log_2(2\kappa)-\frac{\alpha}{2} \inf_{\boldsymbol{\eta}:\|\boldsymbol{\eta}\|_\infty\le \eta}\log_2 \bigl|\Sigma(\boldsymbol{\eta})\bigr|,\label{eq:rate-fnc-1}
\end{align}
will be called the \emph{free energy} term. 

\paragraph{Making the free energy negative.} Recall $\alpha=10\kappa^2\log_2\frac1\kappa$. We claim that for $\kappa$ small, there exist $0<\eta<\beta<1$, $c>0$ and $m\in\mathbb{N}$ such that $\Psi(c,\beta,\eta,m,\alpha)<0$. 

To that end, we first establish $\Psi(c,\beta,0,m,\alpha)<0$ for appropriately chosen $0<\beta<1$, $m\in\mathbb{N}$ and $c>0$. Once this is ensured, observe that the result follows immediately: both the binary entropy and $\log_2|\Sigma(\boldsymbol{\eta})|$ are continuous, and the domain, $\boldsymbol{\eta}\in[0,\eta]^{m(m-1)/2}$ is compact; and thus $\Psi(c,\beta,\eta,m,\alpha)<0$ for any sufficiently small $\eta>0$. 

Let $\Sigma\triangleq \Sigma(0)=(1-\beta)I+\beta\ones\ones^T$. Then, the spectrum of $\Sigma$ consists of the eigenvalue $1-\beta+\beta m$ with multiplicity one; and the eigenvalue $1-\beta$ with multiplicity $m-1$.  Consequently,
\begin{align}
    \Psi(c,\beta,0,m,\alpha)&=1+mh\left(\frac{1-\beta}{2}\right)-\frac{\alpha m}{2}\log_2(2\pi)+\alpha m\log_2(2\kappa)+cm\nonumber\\
    &-\frac{\alpha}{2}(m-1)\log_2(1-\beta)-\frac{\alpha}{2}\log_2(1-\beta+\beta m)\nonumber\\
    &\le m\left(\frac1m-\frac{\alpha}{2m}\log_2(1-\beta+\beta m)+c+\Upsilon(\beta,\alpha)\right)\label{eq:phi-eta-equals-zero},
\end{align}
 where
\begin{equation}\label{eq:upsilon}
\Upsilon(\beta,\alpha) \triangleq h\left(\frac{1-\beta}{2}\right)-\frac{\alpha}{2}\log_2(2\pi) +\alpha\log_2(2\kappa)-\frac{\alpha}{2}\log_2(1-\beta).
\end{equation}
Set 
\begin{equation}\label{eq:beta-star}
    \beta = 1-4\kappa^2,
\end{equation}
and recall $\alpha=10\kappa^2\log_2\frac1\kappa$. With these $\beta$ and $\alpha$ and $\kappa>0$ sufficiently small,
\[
\frac1m - \frac{\alpha}{2m}\log_2\bigl(1-\beta+\beta m\bigr) = o_m(1),\quad\text{as}\quad m\to\infty.
\]
For this reason, it suffices to verify that for every $\kappa>0$ sufficiently small, $\Upsilon(\beta,\alpha)<0$.
\paragraph{Analyzing $\Upsilon(\beta,\alpha)$.} 
Note that $1-\beta = 4\kappa^2=(2\kappa)^2$, and thus
\begin{equation}\label{eq:last-two-cancel}
\alpha\log_2(2\kappa)  - \frac{\alpha}{2}\log_2(1-\beta)=0.
\end{equation}
Using the Taylor expansion, $\log_2(1-x) = -\frac{x}{\ln 2}+o(x)$ as $x\to 0$, we have
\[
\log_2\bigl(1-2\kappa^2\bigr) = -\frac{2}{\ln 2}\kappa^2+o_\kappa\bigl(\kappa^2\bigr),
\]
as $\kappa\to 0$. Consequently,
\begin{align}
    h\left(\frac{1-\beta}{2}\right)= h\bigl(2\kappa^2\bigr)&=-2\kappa^2\log_2\bigl(2\kappa^2\bigr)-\bigl(1-2\kappa^2\bigr)\log_2\bigl(1-2\kappa^2\bigr)\nonumber\\
    &=4\kappa^2\log_2\frac1\kappa +\Theta_\kappa\bigl(\kappa^2\bigr)\label{eq:entropik}.
\end{align}
Combining~\eqref{eq:last-two-cancel} and~\eqref{eq:entropik}, we thus obtain
\begin{equation}\label{eq:upsilon-neg}
\Upsilon\bigl(\beta,\alpha\bigr)=\Bigl(-5\log_2(2\pi)+4\Bigr)\kappa^2\log_2\frac1\kappa + \Theta\bigl(\kappa^2\bigr),
\end{equation}
which is indeed negative for every $\kappa$ small.
\paragraph{Combining everything.} We now complete the argument. For our choice of $\beta$ and $\alpha$, \eqref{eq:upsilon-neg} implies that $\Upsilon(\beta,\alpha)<0$ for every $\kappa$ small. Having ensured this (for a fixed $\kappa$), we then simultaneously set $m\in\mathbb{N}$ to be sufficiently large and $c>0$ to be sufficiently small, so that
\[
\frac1m-\frac{\alpha}{2m}\log_2\bigl(1-\beta+\beta m\bigr)+c+\Upsilon(\beta,\alpha)<0.
\]
This, via~\eqref{eq:phi-eta-equals-zero}, ensures $\Psi(c,\beta,0,m,\alpha)<0$. Finally, (uniform) continuity in $\eta$ ensures that for every small enough $\eta>0$, $\Psi(c,\beta,\eta,m,\alpha)<0$; hence
\[
  \mathbb{E}\Bigl[\Bigl|\mathcal{S}_\kappa(\beta,\eta,m,\alpha,\mathcal{I})\Bigr|\Bigr] = \exp\bigl(-\Theta(n)\bigr)
\]
by~\eqref{eq:exp-upbd}. Finally, inserting this into~\eqref{eq:m-ogp-markov}, we complete the proof. 
\subsection{Proof of Theorem~\ref{thm:stable-hardness}}\label{sec:pf-hardness}
Our proof is quite similar to that of~\cite[Theorem~3.2]{gamarnik2021algorithmic}, including the aforementioned Ramsey argument. In order to guide the reader, we commence this section with an outline of the proof.
\subsubsection*{Proof Outline for Theorem~\ref{thm:stable-hardness}}
Fix a $\kappa>0$, an $\alpha\ge 10\kappa^2\log_2\frac1\kappa$; and recall the parameters, $m\in\mathbb{N}$ and $0<\eta<\beta<1$, prescribed by our $m-$OGP result, Theorem~\ref{thm:m-OGP-small-kappa}. In a nutshell, our proof is based on a contradiction argument. To that end, assume such a stable $\A$ exists. Using $\A$; we create, with positive probability, an instance of the forbidden structure ruled out by the $m-$OGP. A brief roadmap is as follows.
\begin{itemize}
    \item As customary, we first reduce to the case of deterministic algorithms. That is, we find an $\omega^*\in\Omega$, set $\A^*(\cdot)=\A(\cdot,\omega^*):\R^{M\times n}\to\bincube$, and operate with this deterministic $\A^*$. This is the subject of Lemma~\ref{lemma:reduce-to-deterministic}.
    \item We then study a certain high-probability event, dubbed as \emph{chaos} event. This event pertains to $m-$tuples $\sigma^{(i)}\in\bincube$, $1\le i\le m$, where $\|\M_i\sigma^{(i)}\|_\infty\le\kappa\sqrt{n}$ for i.i.d.\,random matrices $\M_i\in\R^{M\times n}$, each with i.i.d. $\cN(0,1)$ coordinates. Namely, $\sigma^{(i)}$ satisfy constraints dictated by independent instances of disorder. We show the existence of a $\beta'$, such that w.h.p.\,it is the case that for any such $m-$tuple, there exists $1\le i<j\le m$ such that $\mathcal{O}(\sigma^{(i)},\sigma^{(j)})\le \beta'$. Notably, $\beta'<\beta-\eta$. This is the subject of Lemma~\ref{lemma:chaos}.
    \item We then (a) generate $T+1$ i.i.d.\,random matrices $\M_i\in\R^{M\times n}$, $0\le i\le T$ (dubbed as \emph{replicas}); (b) divide the interval $[0,\pi/2]$ into $Q$ equal pieces, $0=\tau_0<\tau_1<\cdots<\tau_Q=1$; (c) construct \emph{interpolation trajectories}
    \[
    \M_i(\tau_k) =\cos(\tau_k)\M_0+\sin(\tau_k)\M_i\in\R^{M\times n},\quad 1\le i\le T,\quad 0\le k\le Q;
    \]
    and (d) evaluate $\A^*$ along each trajectory and time step, by setting
    \[
    \sigma_i(\tau_k)\triangleq \A^*\bigl(\M_i(\tau_k)\bigr)\in\bincube.
    \]
    We tune $T,Q\in\mathbb{N}$ appropriately.
    \item We next show in Proposition~\ref{prop:overlaps-are-stable} that since $\A^*$ is stable; the overlaps evolve smoothly along each trajectory. That is, we show that for every $1\le i<j\le T$ and $0\le k\le Q-1$,
    \[
    \Bigl|\mathcal{O}^{(ij)}(\tau_k)-\mathcal{O}^{(ij)}(\tau_{k+1})\Bigr|,\qquad\text{where}\qquad \mathcal{O}^{(ij)}(\tau)\triangleq \frac1n\ip{\sigma_i(\tau)}{\sigma_j(\tau)},
    \] 
    is small.
    \item We then show, by taking a union bound over $1\le i\le T$ and $0\le k\le Q$; that with positive probability, the algorithm is successful along each trajectory and time step:
    \[
    \Bigl\|\M_i(\tau_k)\sigma_i(\tau_k)\Bigr\|_\infty\le \kappa\sqrt{n},\quad 1\le i\le T,\quad 0\le k\le Q.
    \]
    This is the subject of Lemma~\ref{lemma:correctness}. 
    \item We next take a union bound, over all subsets $A\subset[T]$ with $|A|=m$, to extend the aforementioned \emph{chaos} event to all such subsets. 
    \item We then let $\tau$ evolve from $\tau_0=0$ to $\tau_Q=\pi/2$. Notice that in the beginning, $\sigma_i(\tau_0)$ are all equal; whereas at the end, $\sigma_i(\tau_Q)$ are obtained by applying $\A^*$ to i.i.d.\,matrices, $\M_i(\tau_Q)$. Using this observation, the chaos property, as well as the stability of the overlaps (in the sense of above), we establish in Proposition~\ref{prop:eventual-trap} the following. For every $A\subset[T]$ with $|A|=m$, there exists $1\le i_A<j_A\le T$ and a time $\tau_A\in\{\tau_1,\dots,\tau_Q\}$ such that 
    \[
    \mathcal{O}^{(i_A,j_A)}\bigl(\tau_A\bigr)\in(\beta-\eta,\beta).
    \]
    \item We then construct a graph $\mathbb{G}=(V,E)$ on $|V|=T$ vertices. More specifically, (a) vertex $i\in V$ of $\mathbb{G}$ corresponds to $i{\rm th}$ interpolation trajectory; and (b) for any $1\le i<j\le T$, $(i,j)\in E$ iff there is a time $t\in\{1,2,\dots,Q\}$ such that $\mathcal{O}^{(ij)}(\tau_t)\in (\beta-\eta,\beta)$. Note, from the previous bullet point, that the largest independent set of $\mathbb{G}$ is of size at most $m-1$. That is $\alpha(\mathbb{G})\le m-1$. We next color each edge $(i,j)\in E$ of $\mathbb{G}$ with the first time $t\in\{1,2,\dots,Q\}$ such that $\mathcal{O}^{(ij)}(\tau_t)\in(\beta-\eta,\beta)$.
    \item We next apply the Ramsey argument twice. We first use the so-called \emph{two-color} version of Ramsey Theory, Theorem~\ref{thm:2-color-ramsey}. Using the fact $\alpha(\mathbb{G})\le m-1$; it follows that $\mathbb{G}$ contains a large clique, provided that the number $T$ of vertices is sufficiently large. Call this large clique $K_M$, and observe that each edge of $K_M$ is colored with one of $Q$ potential colors. We then apply the so-called \emph{multicolor} version of Ramsey Theory, Theorem~\ref{thm:q-color-ramsey}. Provided $M$ is large (which is ensured by our eventual choice of parameters), we deduce the original graph $\mathbb{G}$ contains a monochromatic $m-$clique $K_m$. These are done in Proposition~\ref{prop:mono-chrom}.
    \item We now interpret the monochromatic $K_m$ extracted above. There exists $1\le i_1<i_2<\cdots<i_m\le T$ and a time $t\in\{1,2,\dots,Q\}$ such that $\mathcal{O}^{(i_k,i_\ell)}(\tau_t)\in(\beta-\eta,\beta)$ for $1\le k<\ell \le m$. Setting $\sigma^{(i)}\triangleq \sigma_{i_k}(\tau_t)\in\bincube$, this $m-$tuple $\sigma^{(i)}\in\bincube$ is precisely the forbidden configuration ruled out by our $m-$OGP result. 
    \item Finally, under the assumption that such an $\A^*$ exists; the whole process outlined above happens with positive probability. That is, we generate such an $m-$tuple with positive probability; contradicting with the $m-$OGP result, where the guarantee is exponentially small in $n$. This settles Theorem~\ref{thm:stable-hardness}.
    \end{itemize}
Before we provide the complete proof, we record the following auxiliary results. 
\subsubsection*{Auxiliary Results from Ramsey Theory in Extremal Combinatorics}
As was already noted, our proof uses Ramsey Theory in extremal combinatorics in a crucial way. To that end, we provide two auxiliary results. The first result pertains to the so-called two-color Ramsey numbers.
\begin{theorem}\label{thm:2-color-ramsey}
Let $k,\ell\ge 2$ be integers; and $R(k,\ell)$ denotes the smallest $n\in\mathbb{N}$ such that any red/blue (edge) coloring of $K_n$ necessarily contains either a red $K_k$ or a blue $K_\ell$. Then,
\[
R(k,\ell) \le \binom{k+\ell-2}{k-1}=\binom{k+\ell-2}{\ell-1}.
\]
In the special case where $k=\ell=M\in\mathbb{N}$, we thus have
\[
R(M,M)\le \binom{2M-2}{M-1}.
\]
\end{theorem}
Theorem~\ref{thm:2-color-ramsey} is folklore. For a proof, see e.g.~\cite[Theorem~6.6]{gamarnik2021algorithmic}. 

The second auxiliary result pertains to the so-called multicolor Ramsey numbers.
\begin{theorem}\label{thm:q-color-ramsey}
Let $q,m\in\mathbb{N}$. Denote by $R_q(m)$ the smallest $n\in\mathbb{N}$ such that any $q-$coloring of the edges of $K_n$ necessarily contains a monochromatic $K_m$. Then, 
\[
R_q(m)\le q^{qm}.
\]
\end{theorem}
Theorem~\ref{thm:q-color-ramsey} can be established by using a minor modification of the so-called \emph{neighborhood chasing} argument due to Erd{\"o}s and Szekeres~\cite{erdos1935combinatorial}, see~\cite[Page~6]{conlon2015recent} for a more elaborate discussion.
\subsubsection*{Proof of Theorem~\ref{thm:stable-hardness}}
Let $\kappa>0$ be a sufficiently small (fixed) constant, $\alpha\ge \alpha_{{\rm OGP}}(\kappa)=10\kappa^2 \log_2\frac1\kappa$ (with $\alpha<\alpha_c(\kappa)$); and $M=\lfloor n\alpha\rfloor\in\mathbb{N}$.

We establish the hardness result for stable algorithms. That is, we show that there exists no randomized algorithm $\A:\R^{M\times n}\times \Omega\to \mathcal{B}_n$ which is
\[
\bigl(\rho,p_f,p_{\rm st},f,L\bigr)-\text{stable}
\]
(for every sufficiently large $n$) 
for the \texttt{SBP}, in the sense of Definition~\ref{def:admissible-alg}. We argue by contradiction: suppose such an $\A$ exists.
\paragraph{Parameter Choice.} 
For the above choice of $\alpha$ and $\kappa$, let $m\in\mathbb{N}$, and $0<\eta<\beta=1-4\kappa^2<1$ be $m-$OGP parameters prescribed by Theorem~\ref{thm:m-OGP-small-kappa}. Observe that the $m-$OGP statement still holds with parameters $0<\eta'<\beta<1$ if $\eta'<\eta$. For this reason, we assume
\begin{equation}\label{eq:beta-eta-5kappa^2}
    \eta<\kappa^2, \qquad\beta -\eta> 1-5\kappa^2.
\end{equation}
We first set
\begin{equation}\label{eq:stable-f-C}
  f = Cn\quad\text{where}\quad C = \frac{\eta^2}{1600};
\end{equation}
then define auxiliary parameters $Q$ and $T$, where
\begin{equation}\label{eq:Q-and-T}
  Q =\frac{4800L \pi}{\eta^2}\sqrt{\alpha} \qquad\text{and}\qquad T=\exp_2\left(2^{4mQ\log_2 Q}\right);
\end{equation}
and finally prescribe $p_f,p_{\rm st}$, and $\rho$ where
\begin{equation}\label{eq:pf-pst}
  p_f =\frac{1}{9(Q+1)T}, \qquad\qquad p_{{\rm st}} = \frac{1}{9Q(T+1)},\qquad\text{and}\qquad \rho = \cos\left(\frac{\pi}{2Q}\right).
\end{equation}
\paragraph{Reduction to Deterministic Algorithms.} We next establish that randomness do not improve the performance of a stable algorithm by much. 
\begin{lemma}\label{lemma:reduce-to-deterministic}
Let $\kappa>0$, $\alpha<\alpha_c(\kappa)$, and $M=\lfloor n\alpha\rfloor$. Suppose that $\A:\R^{M\times n}\times \Omega\to \bincube$ is a randomized algorithm that is $\bigl(\rho,p_f,p_{\rm st},f,L\bigr)-$stable (for the \texttt{SBP}). Then, there exists a deterministic algorithm $\A^*:\R^{M\times n}\to \bincube$ that is $\bigl(\rho,3p_f,3p_{\rm st},f,L\bigr)-$stable\footnote{Lemma~\ref{lemma:reduce-to-deterministic} applies also to the deterministic algorithms, see remarks following Definition~\ref{def:admissible-alg}.}.
\end{lemma}
\begin{proof}[Proof of Lemma~\ref{lemma:reduce-to-deterministic}]
For any $\omega\in\Omega$, define the event
\[
\mathcal{E}_1(\omega)\triangleq \Bigl\{\bigl|\mathcal{M} \A\bigl(\mathcal{M},\omega\bigr)\bigr|\le (\kappa\sqrt{n})\ones\Bigr\},
\]
where $\M\in\R^{M\times n}$, $\A(\M,\omega)\in\bincube$, and the inequality is coordinate-wise. Observe that
\[
\mathbb{P}_{\mathcal{M},\omega}\bigl[\bigl|\mathcal{M} \A\bigl(\mathcal{M},\omega\bigr)\bigr|>(\kappa\sqrt{n})\ones\bigr] = \mathbb{E}_\omega\bigl[\mathbb{P}_{\M}\bigl(\mathcal{E}_1^c(\omega)\bigr)\bigr].
\]
Perceiving $\mathbb{P}_{\M}(\mathcal{E}_1^c)$ as a random variable with source of randomness $\omega$ (note that the randomness over $\M$ is ``integrated" over $\mathbb{P}_{\M}$), we have by Markov's inequality
\[
\mathbb{P}_\omega\bigl[\mathbb{P}_{\M}\bigl[\mathcal{E}_1^c(\omega)\bigr]\ge 3p_f\bigr]\le \frac{\mathbb{E}_\omega\bigl[\mathbb{P}_{\M}\bigl[\mathcal{E}_1^c(\omega)\bigr]\bigr]}{3p_f}\le \frac13.
\]
Hence, $\mathbb{P}[\Omega_1]\ge 2/3$, where
\[
\Omega_1\triangleq \Bigl\{\omega\in\Omega:\mathbb{P}_{\M}\bigl[\mathcal{E}_1^c(\omega)\bigr]<3p_f\Bigr\}.
\]
Defining next
\[
\mathcal{E}_2(\omega)\triangleq \Bigl\{d_H\bigl(\A(\M,\omega),\A(\overline{\M},\omega)\bigr)\le f+L\|\M-\overline{\M}\|_F\Bigr\},
\]
where $\M,\overline{\M}\in\R^{M\times n}$ with i.i.d.\,standard normal coordinates subject to the constraint $\mathbb{E}\bigl[\M_{11}\overline{\M}_{11}\bigr] = \rho$. Applying the exact same logic, we find $\mathbb{P}[\Omega_2]\ge 2/3$, where
\[
\Omega_2\triangleq \Bigl\{\omega\in\Omega:\mathbb{P}_{\M,\overline{\M}}\bigl[\mathcal{E}_2^c(\omega)\bigr]<3p_{\rm st}\Bigr\}.
\]
Noting $\mathbb{P}[\Omega_1]+\mathbb{P}[\Omega_2]=4/3>1$, it follows $\Omega_1\cap \Omega_2\ne\varnothing$. Now take any $\omega^*\in \Omega_1\cap \Omega_2$, and set $\A^*(\cdot)\triangleq \A(\cdot,\omega^*)$. Clearly, $\A^*$ is $\bigl(\rho,3p_f,3p_{\rm st},f,L\bigr)-$stable, establishing Lemma~\ref{lemma:reduce-to-deterministic}.
\end{proof}
In the remainder, we restrict our attention to deterministic $\A^*:\R^{M\times n}\to\bincube$ appearing in Lemma~\ref{lemma:reduce-to-deterministic} which is $\bigl(\rho,3p_f,3p_{\rm st},f,L\bigr)-$stable.
\paragraph{Chaos event.} We now focus on the so-called \emph{chaos} event, which pertains to $m-$tuples $\sigma^{(i)}\in\bincube$, $1\le i\le m$, where $\bigl\|\M_i\sigma^{(i)}\bigr\|_\infty\le\kappa\sqrt{n}$ for i.i.d.\,random matrices $\M_i\in\R^{M\times n}$. Namely, we investigate $m-$tuples satisfying constraints dictated by independent instances of disorder.
\begin{lemma}\label{lemma:chaos}
For every sufficiently small $\kappa>0$ and $\alpha\ge \alpha_{{\rm OGP}}(\kappa)=10\kappa^2\log_2\frac1\kappa$, and sufficiently large $m\in\mathbb{N}$,
\[
\mathbb{P}\Bigl[S_\kappa\bigl(1,5\kappa^2,m,\alpha,\{\pi/2\}\bigr)\ne\varnothing\Bigr]\le\exp_2\bigl(-\Theta(n)\bigr).
\]
\end{lemma}
\begin{proof}[Proof of Lemma~\ref{lemma:chaos}]
The proof is quite similar to (and in fact, much simpler than) that of Theorem~\ref{thm:m-OGP-small-kappa}. Hence, we only provide a brief sketch. Check that for any $\sigma \in\bincube$ and $X\distr \cN(0,I_n)$, 
\[
\mathbb{P}\left[-\kappa \le n^{-\frac12}\ip{\sigma}{X}\le \kappa \right]\le \frac{2}{\sqrt{2\pi}}\kappa.
\]
Endowed with this, a straightforward first moment argument yields
\begin{align*}
    \mathbb{E}\Bigl[\bigl|S_\kappa\bigl(5\kappa^2,1,m,\alpha,\{\pi/2\}\bigr)\bigr|\Bigr]&\le\exp_2\left(n\left(1+mh\left(\frac{5\kappa^2}{2}\right)+\alpha m\log_2\left(\frac{2\kappa}{\sqrt{2\pi}}\right)\right)+O\bigl(\log_2 n\bigr)\right)\\
    &= \exp_2\left(nm\left(\frac1m+h\left(\frac{5\kappa^2}{2}\right)+\alpha\log_2\left(\frac{2\kappa}{\sqrt{2\pi}}\right)\right)+O\bigl(\log_2 n\bigr)\right).
\end{align*}
We now apply the Taylor expansion, $\log_2(1-x) = -\frac{x}{\ln 2}+o(x)$ as $x\to 0$ to obtain
\begin{align*}
h\left(\frac{5\kappa^2}{2}\right) &=\frac{5\kappa^2}{2}\log_2\left(\frac{5\kappa^2}{2}\right) + \left(1-\frac{5\kappa^2}{2}\right)\underbrace{\log_2\left(1-\frac{5\kappa^2}{2}\right)}_{=-\Theta_\kappa\bigl(\kappa^2\bigr)}\\
&=-5\kappa^2\log_2\kappa +\Theta_\kappa\bigl(\kappa^2\bigr).
\end{align*}
Consequently,
\[
h\left(\frac{5\kappa^2}{2}\right)+\alpha\log_2\left(\frac{2\kappa}{\sqrt{2\pi}}\right)=-10\kappa^2\left(\log_2 \frac1\kappa\right)^2 +\Theta_\kappa\bigl(\kappa^2\log_2 \kappa\bigr),
\]
which is indeed negative for all sufficiently small $\kappa>0$. Having ensured $\kappa>0$ is sufficiently small, note that if $h(5\kappa^2/2)+\alpha\log_2(2\kappa/\sqrt{2\pi})<0$ then for every sufficiently large $m\in\mathbb{N}$, 
\[
\frac1m+h\left(\frac{5\kappa^2}{2}\right)+\alpha\log_2\left(\frac{2\kappa}{\sqrt{2\pi}}\right)<0.
\]
This yields the conclusion by Markov's inequality. 
\end{proof}
In particular, for every $\kappa>0$, $\alpha\ge 10\kappa^2\log_2\frac1\kappa$ and $m\in\mathbb{N}$ large enough, w.h.p.\,it is the case that for every $m-$tuple $\sigma^{(i)}\in\bincube$, $1\le i\le m$ with $\|\M_i\sigma^{(i)}\|_\infty \le \kappa\sqrt{n}$ (where $\M_i$ are i.i.d.\,random matrices with i.i.d.\,$\cN(0,1)$ coordinates), there exists $1\le i<j\le m$ such that
\begin{equation}\label{eq:chaos-main}
\frac1n\ip{\sigma^{(i)}}{\sigma^{(j)}}\le 1-5\kappa^2.
\end{equation}
\paragraph{Construction of Interpolation Paths.} Our proof uses \emph{interpolation} ideas. To that end, let $\M_i\in\R^{M\times n}$, $0\le i\le T$ (recall $T$ from~\eqref{eq:Q-and-T}), be a sequence of i.i.d.\,random matrices, each with i.i.d. $\cN(0,1)$ coordinates. Recall the interpolation appearing in~\eqref{eq:interpolate-matrix}, repeated below for convenience: 
\begin{equation}\label{eq:interpolation-path}
\M_i(\tau)\triangleq \cos(\tau)\M_0 + \sin(\tau)\M_i,\quad 1\le i\le T,\quad \tau\in[0,\pi/2].
\end{equation}
Observe that for any fixed $\tau\in[0,\pi/2]$, $\M_i(\tau)$ consists of i.i.d.\,standard normal entries. 

Next for $Q$ appearing in~\eqref{eq:Q-and-T}; we discretize $[0,\pi/2]$ into $Q$ sub intervals---each of size $\Theta\bigl(Q^{-1}\bigr)$---where the endpoints are given by
\begin{equation}\label{eq:discretize-0-pi/2}
0=\tau_0<\tau_1<\cdots<\tau_Q = \frac{\pi}{2}.
\end{equation}
We apply $\A^*$ to each $\M_i(\tau_k)$:
\begin{equation}\label{eq:sigma-i-tau}
    \sigma_i(\tau_k)\triangleq \A^*\bigl(\M_i(\tau_k)\bigr)\in\bincube,\quad 1\le i\le T,\quad 0\le k\le Q.
\end{equation}
For every $1\le i<j\le T$ and $0\le k\le Q$, define pairwise overlaps
\begin{equation}\label{eq:overlaps}
  \mathcal{O}^{(ij)}\bigl(\tau_k\bigr)\triangleq \frac1n\ip{\sigma_i\bigl(\tau_k\bigr)}{\sigma_j\bigl(\tau_k\bigr)}.
\end{equation}
A useful observation is for $k=0$, $\sigma_1(\tau_0)=\cdots=\sigma_T(\tau_0)$; and therefore the overlaps are all unity.
\paragraph{Successive Steps are Stable.} We now show the stability of overlaps, that is, 
\[
\left|\mathcal{O}^{(ij)}(\tau_k)-\mathcal{O}^{(ij)}(\tau_{k+1})\right|
\]
is small for all $1\le i<j\le T$ and $0\le k\le Q-1$. More concretely, we establish the following proposition.
\begin{proposition}\label{prop:overlaps-are-stable}
\begin{align*}
\mathbb{P}\bigl[\mathcal{E}_{\rm St}\bigr]
\ge 1-3(T+1)Qp_{\rm st}-(T+1)\exp\bigl(-\Theta(n^2)\bigr),
\end{align*}
where
\begin{equation}\label{eq:prop-stable}
    \mathcal{E}_{\rm St}\triangleq \bigcap_{1\le i<j\le T}\bigcap_{0\le k\le Q-1}\left\{\Bigl|\mathcal{O}^{(ij)}(\tau_k)-\mathcal{O}^{(ij)}(\tau_{k+1})\Bigr|\le 4\sqrt{C}+4\sqrt{3L \pi Q^{-1}}\alpha^{\frac14}\right\}
\end{equation}
for $C$ defined in~\eqref{eq:stable-f-C}.
\end{proposition}
\begin{proof}[Proof of Proposition~\ref{prop:overlaps-are-stable}]
We first establish an auxiliary concentration result. Let $\M\in\R^{M\times n}$ be a random matrix with i.i.d. $\cN(0,1)$ coordinates. Then by applying Bernstein's inequality like in the proof of~\cite[Theorem~3.1.1]{vershynin2010introduction}, we obtain that for some absolute constant $c>0$ and any $t\ge 0$,
\begin{align*} 
\mathbb{P}\left[\left|\frac{1}{Mn}\sum_{1\le i\le M}\sum_{1\le j\le n}\M_{ij}^2-1\right|\ge t\right]\le \exp\Bigl(-cMn\min\{t,t^2\}\Bigr).
\end{align*}
Taking a union bound and recalling $M=\lfloor n\alpha\rfloor = \Theta(n)$, we thus have
\begin{equation}\label{eq:frobenius-concentration}
    \mathbb{P}\Bigl[\|\M_i\|_F\le 6\sqrt{Mn},0\le i\le T\Bigr]\ge 1-(T+1)\exp\Bigl(-\Theta(n^2)\Bigr).
\end{equation}
The constant $6$ is chosen arbitrarily, and any constant greater than $1$ works. 

Next, we show a simple Lipschitzness property for $\cos(\cdot)$ and $\sin(\cdot)$: we claim that for every $x,y\in\mathbb{R}$,
\[
\bigl|\cos(x)-\cos(y)\bigr| \le |x-y|\qquad\text{and}\qquad \bigl|\sin(x)-\sin(y)\bigr|\le|x-y|.
\]
Indeed, by the mean value theorem, for $x<y$, it holds that for some $c\in(x,y)$; $|\cos(x)-\cos(y)| = |x-y|\cdot |\sin(c)|\le |x-y|$. The result for the $\sin(\cdot)$ is analogous. Consequently,
\begin{equation}\label{eq:cosine/sine-bd}
\max\Bigl\{\bigl|
\cos(\tau_k)-\cos(\tau_{k+1})\bigr|,\bigl|\sin(\tau_k)-\sin(\tau_{k+1})\bigr|\Bigr\}\le |\tau_k-\tau_{k+1}| = \frac{\pi}{2Q},
\end{equation}
where we used~\eqref{eq:discretize-0-pi/2} at the last step.

We now employ~\eqref{eq:cosine/sine-bd} to upper bound $\bigl\|\M_i(\tau_k)-\M_i(\tau_{k+1})\bigr\|_F$ on the high probability event appearing in~\eqref{eq:frobenius-concentration}. Assuming~\eqref{eq:frobenius-concentration} takes place, we have that for any fixed $1\le i\le T$ and $0\le k\le Q-1$,
\begin{align}
    \Bigl\|\M_i(\tau_k)-\M_i(\tau_{k+1})\Bigr\|_F &= \Bigl\|\cos(\tau_k)\M_0+\sin(\tau_k)\M_i -\cos(\tau_{k+1})\M_0 -\sin(\tau_{k+1})\M_i\Bigr\|_F \nonumber\\
    &\le \bigl|\cos(\tau_k)-\cos(\tau_{k+1})\bigr|\|\M_0\|_F + \bigl|\sin(\tau_k)-\sin(\tau_{k+1})\bigr|\|\M_i\|_F \label{eq:frob-triangle}\\
    &\le \frac{\pi}{2Q}\bigl(\|M_0\|_F + \|M_i\|_F\bigr) \label{eq:Q-inv}\\
    &\le \frac{3\pi}{Q}\sqrt{Mn}\label{eq:use-event-here};
\end{align} 
where~\eqref{eq:frob-triangle} uses triangle inequality for the Frobenius norm;~\eqref{eq:Q-inv} uses~\eqref{eq:cosine/sine-bd}; and finally~\eqref{eq:use-event-here} uses the fact that on the event~\eqref{eq:frobenius-concentration},  $\|\M_i\|_F\le 6\sqrt{Mn}$ for $0\le i\le T$.

We next observe that for any fixed $1\le i\le T$ and $0\le k\le Q-1$; each of $\M_i(\tau_k)$ and  $\M_i(\tau_{k+1})$ has i.i.d.\, $\cN(0,1)$ entries subject to
\begin{align}\label{eq:correlation-rho}
\mathbb{E}\left[\bigl(\M_i(\tau_{k})\bigr)_{\ell,j}\bigl(\M_i(\tau_{k+1})\bigr)_{\ell,j}\right]=\cos\bigl(\tau_{k+1}-\tau_k\bigr)=\cos\left(\frac{\pi}{2Q}\right).
\end{align}
for $1\le \ell\le M$ and $1\le j\le n$. Recall now by Lemma~\ref{lemma:reduce-to-deterministic} that $\A^*$ is stable with stability probability $1-3p_{\rm st}$. Taking thus a union bound, we find
\begin{align}
   & \mathbb{P}\left[d_H\Bigl(\sigma_i(\tau_k),\sigma_i(\tau_{k+1})\Bigr)\le Cn + L\bigl\|\M_i(\tau_k)-\M_i(\tau_{k+1})\bigr\|_F,1\le i\le T,0\le k\le Q-1\right]\nonumber\\
   &\ge 1-3(T+1)Qp_{\rm st}\label{eq:event-stable}.
\end{align}
We now combine~\eqref{eq:use-event-here} (valid on event~\eqref{eq:frobenius-concentration}) and the event~\eqref{eq:event-stable} by a union bound. We find that
\begin{equation}\label{eq:prob-event-towards-small-overlap}
    \mathbb{P}[\mathcal{E}]\ge 1-3(T+1)Qp_{\rm st}-(T+1)\exp\bigl(-\Theta(n^2)\bigr),
\end{equation}
where
\begin{equation}\label{eq:event-towards-small-overlap}
\mathcal{E} \triangleq \bigcap_{1\le i\le T}\bigcap_{0\le k\le Q-1}\left\{d_H\bigl(\sigma_i(\tau_k),\sigma_i(\tau_{k+1})\Bigr)\le Cn+\frac{3L\pi}{Q}\sqrt{Mn}\right\}.
\end{equation}
In the remainder of the proof, assume we operate on the event $\mathcal{E}$~\eqref{eq:event-towards-small-overlap}. 

Observe that for any $\sigma,\sigma'\in\bincube$, $\|\sigma-\sigma'\|_2 = 2\sqrt{d_H(\sigma,\sigma')}$; and recall from~\eqref{eq:sigma-i-tau} the notation, $\sigma_i(\tau_k)$. We have that for any $1\le i\le T$ and $0\le k\le Q-1$, 
\begin{align}
    \Bigl\|\sigma_i(\tau_k) -\sigma_i(\tau_{k+1})\Bigr\|_2 &=2\sqrt{d_H\bigl(\sigma_i(\tau_k) ,\sigma_k(\tau_{k+1})\bigr)} \nonumber\\
    &\le 2\sqrt{Cn+3L\pi Q^{-1}\sqrt{Mn}}\label{eq:use-event-E-heree}\\
    &\le \sqrt{n}\Bigl(2\sqrt{C}+2\sqrt{3L\pi Q^{-1}}\alpha^{\frac14}\Bigr)\label{eq:important};
\end{align}
where~\eqref{eq:use-event-E-heree} follows from the fact we are on event $\mathcal{E}$~\eqref{eq:event-towards-small-overlap}; and~\eqref{eq:important} uses the fact $M\le n\alpha$ and the trivial inequality $\sqrt{u+v}\le \sqrt{u}+\sqrt{v}$ valid for all $u,v\ge 0$. 

Equipped with~\eqref{eq:important}, we are now in a position to conclude. Fix any $1\le i<j\le T$ and $0\le k\le Q-1$. We have the following chain of inequalities:
\begin{align}
    \Bigl|\mathcal{O}^{(ij)}(\tau_k)-\mathcal{O}^{(ij)}(\tau_{k+1})\Bigr|&= \frac1n\Bigl|\ip{\sigma_i(\tau_k)}{\sigma_j(\tau_{k})} -\ip{\sigma_i(\tau_{k+1})}{\sigma_j(\tau_{k+1})} \Bigr|\nonumber \\
    &\le \frac1n\Bigl(\Bigl|\ip{\sigma_i(\tau_k)-\sigma_i(\tau_{k+1})}{\sigma_j(\tau_{k})}\Bigr|+\Bigl|\ip{\sigma_i(\tau_{k+1})}{\sigma_j(\tau_{k})-\sigma_j(\tau_{k+1})}\Bigr|\Bigr)\label{eq:step-triangle-ineq}\\
    &\le \frac{1}{\sqrt{n}}\Bigl(\Bigl\|\sigma_i(\tau_k)-\sigma_i(\tau_{k+1})\Bigr\|_2+\Bigl\|\sigma_j(\tau_k)-\sigma_j(\tau_{k+1})\Bigr\|_2\Bigr)\label{eq:step-Cauchy-Schw}\\
    &\le 4\sqrt{C} + 4\sqrt{3L \pi Q^{-1}}\alpha^{\frac14}\label{eq:step-final}.
\end{align}
Indeed,~\eqref{eq:step-triangle-ineq} follows from the triangle inequality;~\eqref{eq:step-Cauchy-Schw} uses Cauchy-Schwarz inequality with the fact $\|\sigma\|_2 = \sqrt{n}$ for any $\sigma\in\bincube$; and~\eqref{eq:step-final} uses~\eqref{eq:important}. Recalling the probability bound~\eqref{eq:prob-event-towards-small-overlap} on the event $\mathcal{E}$ that we operated under, the proof of Proposition~\ref{prop:overlaps-are-stable} is complete.
\end{proof}
\paragraph{$\A^*$ is Successful along Each Trajectory.} We next study the event that $\A^*$ is \emph{successful} along each interpolation trajectory and across times. We have
\begin{lemma}\label{lemma:correctness}
\[
\mathbb{P}\bigl[\mathcal{E}_{\rm Suc}\bigr]\ge 1-3T(Q+1)p_f,
\]
where
\begin{equation}\label{eq:lemma-success}
   \mathcal{E}_{\rm Suc}\triangleq \bigcap_{1\le i\le T}\bigcap_{0\le k\le Q}\Bigl\{ \bigl\|\M_i(\tau_k)\sigma_i(\tau_k)\bigr\|_\infty \le \kappa\sqrt{n}\Bigr\}.
\end{equation}
\end{lemma}
\begin{proof}[Proof of Lemma~\ref{lemma:correctness}]
The results follows immediately by (a) recalling, from Lemma~\ref{lemma:reduce-to-deterministic}, that 
\[
\mathbb{P}_{\M}\Bigl[\bigl\|\M \A^*(\M)\bigr\|_\infty\le \kappa\sqrt{n}\Bigr]\ge 3p_f,
\]
(b) observing $\M_i(\tau_k)\distr \M_0$ for all $i$ and $k$; and
(c) taking a union bound over $1\le i\le T$ and $0\le k\le Q$.
\paragraph{Combining Everything.} Fix any subset $A\subset[T]$ with $|A|=m$, and let $\mathcal{E}_A$ be
\[
\mathcal{E}_A\triangleq \left\{\exists \bigl(\sigma^{(i)}\in\bincube,i\in A\bigr):\max_{i\in A}\bigl\|\M_i(1)\sigma^{(i)}\bigr\|_\infty\le \kappa\sqrt{n},\,\beta-\eta\le n^{-1}\ip{\sigma^{(i)}}{\sigma^{(j)}}\le \beta, i,j\in A,i\ne j\right\}.
\]
Namely, $\mathcal{E}_A$ is nothing but the chaos event in the sense of Lemma~\ref{lemma:chaos}, where the indices are restricted to $A\subset[T]$. In particular, $\mathbb{P}[\mathcal{E}_A]\ge \exp(-\Theta(n))$ due to Lemma~\ref{lemma:chaos}. Taking a union bound over $A\subset [T]$, we obtain
\begin{equation}\label{eq:chaos-union-bd}
    \mathbb{P}[\mathcal{E}_{\rm Ch}]\triangleq \mathbb{P}\left[\bigcap_{A\subset[T],|A|=m} \mathcal{E}_A^c\right]\ge 1-\binom{T}{m}e^{-\Theta(n)}=1-\exp\bigl(-\Theta(n)\bigr),
\end{equation}
where we used the fact $\binom{T}{m}=O(1)$ (as $n\to\infty$) since $T=O(1)$ per~\eqref{eq:Q-and-T} and $m=O(1)$. Let
\begin{equation}\label{eq:event-F}
\mathcal{F}\triangleq \mathcal{E}_{\rm St}\cap \mathcal{E}_{\rm Suc}\cap \mathcal{E}_{\rm Ch},
\end{equation}
where $\mathcal{E}_{\rm St}$, $\mathcal{E}_{\rm Suc}$, and $\mathcal{E}_{\rm Ch}$ are defined, respectively, in~\eqref{eq:prop-stable},~\eqref{eq:lemma-success}, and~\eqref{eq:chaos-union-bd}. We then have
\begin{align}
    \mathbb{P}[\mathcal{F}]&\ge 1- \mathbb{P}\bigl[\mathcal{E}_{\rm St}^c\bigr]-\mathbb{P}\bigl[\mathcal{E}_{\rm Suc}^c\bigr]-\mathbb{P}\bigl[\mathcal{E}_{\rm Ch}^c\bigr]\label{eq:union-bound-step}\\
    &\ge 1-3(T+1)Qp_{\rm st}-(T+1)e^{-\Theta(n^2)}-3T(Q+1)p_f-e^{-\Theta(n)}\label{eq:step-lemmas}\\
    &\ge \frac13-\exp\bigl(-\Theta(n)\bigr)\label{eq:step-prob-bds-early},
\end{align}
where~\eqref{eq:union-bound-step} follows from a union bound;~\eqref{eq:step-lemmas} uses Proposition~\ref{prop:overlaps-are-stable}, Lemma~\ref{lemma:correctness} and~\eqref{eq:chaos-union-bd}; and~\eqref{eq:step-prob-bds-early} recalls~\eqref{eq:pf-pst} for $p_f$ and $p_{\rm st}$. We operate on the event $\mathcal{F}$ in the remainder of the proof.

Now, inserting into~\eqref{eq:prop-stable} the choice of $C$ per~\eqref{eq:stable-f-C} and $Q$ per~\eqref{eq:Q-and-T}; it is the case that on $\mathcal{F}$, 
\begin{equation}\label{eq:eta-over-five}
\Bigl|\mathcal{O}^{(ij)}(\tau_k)-\mathcal{O}^{(ij)}(\tau_{k+1})\Bigr| \le \frac{\eta}{5}
\end{equation}
for every $1\le i<j\le T$ and $0\le k\le Q-1$. Fix next any $A\subset[T]$ with $|A|=m$. We establish the following proposition.
\begin{proposition}\label{prop:eventual-trap}
For every $A\subset[T]$ with $|A|=m$, there exists $1\le i_A<j_A\le m$ and $\tau_A\in\{\tau_1,\dots,\tau_Q\}$ such that for $\delta=\frac{\eta}{100}$,
\[
\mathcal{O}^{(i_A,j_A)}(\tau_A) \in \bigl(\beta-\eta+3\delta,\beta-3\delta)\subsetneq (\beta-\eta,\beta).
\]
\end{proposition}
\begin{proof}[Proof of Proposition~\ref{prop:eventual-trap}]
A consequence of $\mathcal{E}_{\rm Ch}$ (part of $\mathcal{F}$) is that there exists distinct $i_A,j_A\in A$ such that
\[
\mathcal{O}^{(i_A,j_A)}(\tau_Q) \le 1-5\kappa^2,
\]
where we utilized~\eqref{eq:chaos-main}. Recall now the choice of $\beta=1-4\kappa^2$ and $\eta$ such that $\beta-\eta >1-5\kappa^2$. In particular, $\mathcal{O}^{(i_A,j_A)}(\tau_Q)<\beta-\eta$. We now claim for $\delta=\eta/100$, there exists a $k'\in\{1,2,\dots,Q\}$ such that
\[
\mathcal{O}^{(i_A,j_A)}(\tau_{k'})\in \bigl(\beta-\eta+3\delta,\beta-3\delta\bigr).
\]
To that end, take $K_0\in \{1,2,\dots,Q\}$ to be the \emph{last time}  such that $\mathcal{O}^{(i_A,j_A)}(\tau_{K_0})\ge \beta-3\delta$. Note that such a $K_0$ must exist as $\mathcal{O}^{(ij)}(0)=1$ for every $1\le i<j\le T$. Then if $\mathcal{O}^{(i_A,j_A)}(\tau_{K_0+1})\le \beta-\eta+3\delta$, we obtain
\[
\Bigl|\mathcal{O}^{(i_A,j_A)}(\tau_{K_0})-\mathcal{O}^{(i_A,j_A)}(\tau_{K_0+1})\Bigr|\ge \eta-6\delta,
\]
contradicting~\eqref{eq:eta-over-five} for sufficiently large $n$. That is, 
\[
\mathcal{O}^{(i_A,j_A)}(\tau_{K_0+1})\in (\beta-\eta+3\delta,\beta-3\delta).
\]
Since $A\subset[T]$ was arbitrary, Proposition~\ref{prop:eventual-trap} is established.
\end{proof}
\paragraph{Constructing an Appropriate Graph, and Applying Ramsey Theory.} We now construct a certain graph $\mathbb{G}=(V,E)$ satisfying the following properties.
\begin{itemize}
    \item The vertex set $V$ coincides with $[T]$. That is, $V=\{1,2,\dots,T\}$, where each vertex $i$ corresponds to the interpolation trajectory $i$, $1\le i\le T$.
    \item For any $1\le i<j\le T$, we add $(i,j)\in E$ iff there exists a time $\tau\in[0,1]$ such that $\mathcal{O}^{(ij)}(\tau)\in(\beta-\eta,\beta)$.
\end{itemize}
Namely, $\mathbb{G}$ is a graph with a potentially large number of vertices, and a certain number of edges. 

We next \emph{color} each $(i,j)\in E$ with one of $Q$ colors. Specifically, for any $1\le i<j\le T$ with $(i,j)\in E$; we color the edge $(i,j)\in E$ with color $t$, $1\le t\le Q$, where $\tau_t\in\{\tau_1,\dots,\tau_Q\}$ is the first time such that 
\[
\mathcal{O}^{(ij)}(\tau_t)\in (\beta-\eta,\beta).
\]
Having done this coloring, $\mathbb{G}=(V,E)$ satisfies the following properties:e
\begin{itemize}
    \item[(a)] We have $|V|=T$; and for every $A\subset V$ with $|A|=m$, there exists distinct $i_A,j_A\in A$ such that $(i_A,j_A)\in E$. Namely, $\mathbb{G}$ contains no independent sets of size larger than $m-1$.
    \item[(b)] Any $(i,j)\in E$ is colored with one of colors $\{1,2,\dots,Q\}$.
\end{itemize}
We claim 
\begin{proposition}\label{prop:mono-chrom}
$\mathbb{G}=(V,E)$ defined above contains a monochromatic $m-$clique, $K_m$.
\end{proposition}
\begin{proof}[Proof of Proposition~\ref{prop:mono-chrom}]
Recall from~\eqref{eq:Q-and-T} that $\mathbb{G}$ contains  $T=\exp_2\bigl(\exp_2\bigl(4mQ\log_2 Q\bigr)\bigr)$ vertices. Set
\begin{equation}\label{eq:capital-M}
    M\triangleq Q^{mQ}=2^{mQ\log_2 Q}.
\end{equation}
\paragraph{Extracting a Large Clique $K_M$.} Recall from Theorem~\ref{thm:2-color-ramsey} that
\[
R_2(M,M)\le \binom{2M-2}{M-1}.
\]
As a result, any graph with at least $\binom{2M-2}{M-1}$ vertices contains either an independent set of cardinality $M$, or an $M-$clique, $K_M$. Now, from property ${\rm (a)}$ above, the largest independent set of $\mathbb{G}$ is of size at most $m-1$, which is less than $M$. Since
\[
T=\exp_2\left(2^{4mQ\log_2 Q}\right)\ge 2^{2M} = 4^M\ge \binom{2M-2}{M-1}
\]
for $M$ defined in~\eqref{eq:capital-M}, it follows that $\mathbb{G}$ contains a $K_M$, where $M=Q^{Qm}$, each of whose edges is colored with one of $Q$ colors.
\paragraph{Further Extracting a Monochromatic $K_m$.} Now that we extracted a $K_M$ with $M=Q^{Qm}$. Since $R_Q(m)\le Q^{Qm}$ per Theorem~\ref{thm:q-color-ramsey}; we obtain, by applying the multicolor version of Ramsey Theory, that $K_M$ contains a monochromatic $K_m$. Since $K_M$ is a subgraph of $\mathbb{G}$, this establishes Proposition~\ref{prop:mono-chrom}.
\end{proof}
We now finalize the proof of Theorem~\ref{thm:stable-hardness}. We interpret $K_m$ of $\mathbb{G}$ extracted in Proposition~\ref{prop:mono-chrom}: there exists an $m-$tuple, $1\le i_1<i_2<\cdots<i_m\le T$ and a color $t\in\{1,2,\dots,Q\}$ such that
\[
\mathcal{O}^{(i_k,i_\ell)}(\tau_t) \in(\beta-\eta,\beta),1\le k<\ell\le m.
\]
Now, set $\sigma^{(k)}\triangleq \A^*\bigl(\M_{i_k}(\tau_t)\bigr)\in\bincube$, $1\le k\le m$. Observe the following for this $m-$tuple:
\begin{itemize}
    \item Noting we are on $\mathcal{F}$~\eqref{eq:event-F}, and in particular $\mathcal{F}\subset \mathcal{E}_{\rm Suc}$ defined in~\eqref{eq:lemma-success}; we have
    \[
    \Bigl\|\M_{i_k}(\tau_t)\sigma^{(k)}\Bigr\|_\infty \le \kappa\sqrt{n},\quad 1\le k\le m.
    \]
    \item For $1\le k<\ell\le m$, 
    \[
    \beta-\eta<\frac1n\ip{\sigma^{(k)}}{\sigma^{(\ell)}}<\beta.
    \]
\end{itemize}
In particular, for the choice $\zeta=\{i_1,i_2,\dots,i_m\}$ of the $m-$tuple of distinct indices, the set $\mathcal{S}_\zeta\triangleq \mathcal{S}_\kappa\bigl(\beta,\eta,m,\alpha,\mathcal{I}\bigr)$ with $\mathcal{I}=\{\tau_i:0\le i\le Q\}$ is non-empty. That is,
\[
\mathbb{P}\Bigl[\exists\zeta \in [T]:|\zeta|=m,\mathcal{S}_\zeta\ne\varnothing\Bigr]\ge \mathbb{P}[\mathcal{F}]\ge \frac13-\exp\bigl(-\Theta(n)\bigr).
\]
Notice, on the other hand, that using the $m-$OGP result, Theorem~\ref{thm:m-OGP-small-kappa}, we have 
\[
\mathbb{P}\Bigl[\exists\zeta \in [T]:|\zeta|=m,\mathcal{S}_\zeta\ne\varnothing\Bigr] \le \binom{T}{m}e^{-\Theta(n)}=\exp\bigl(-\Theta(n)\bigr),
\]
by taking a union bound and recalling $\binom{T}{m}=O(1)$. Combining these, we therefore obtain
\[
\exp\bigl(-\Theta(n)\bigr)\ge \frac13-\exp\bigl(-\Theta(n)\bigr),
\]
which is clearly a contradiction for all $n$ large enough, establishing the result.
\end{proof}

\subsection{Proof of Theorem~\ref{thm:online-alg-fail}}\label{sec:pf-online-alg-fail}
We first provide an auxiliary result.
\begin{proposition}\label{prop:online-alg-fail}
Fix $\Delta\in(0,\frac12)$. Let $\M\in\R^{M\times n}$ be a matrix with i.i.d. $\mathcal{N}(0,1)$ coordinates; and let $\M_\Delta\in\R^{M\times n}$ be the matrix obtained from $\M$ by resampling its last $\Delta\cdot n$ columns independently from $\mathcal{N}(0,1)$. Let $\Xi(\Delta)\subset \bincube\times\bincube$ be the set of all $(\sigma,\sigma_\Delta)\in \bincube\times\bincube$ satisfying the following conditions.
\begin{itemize}
    \item $\bigl\|\M\sigma\bigr\|_\infty\le \sqrt{n}$ and $\bigl\|\M_\Delta\sigma_\Delta\bigr\|_\infty\le \sqrt{n}$.
    \item $n^{-1}\ip{\sigma}{\sigma_\Delta}\in [1-2\Delta,1]$.
\end{itemize}
Then, there is a $\Delta>0$ such that
\[
\mathbb{P}\bigl[\Xi(\Delta)=\varnothing\bigr]\ge 1-\exp\bigl(-\Theta(n)\bigr).
\]
\end{proposition}
Assuming Proposition~\ref{prop:online-alg-fail}, we now show how to establish Theorem~\ref{thm:online-alg-fail}. Suppose such an $\A$ that is $p_f-$online for $p_f<\frac12-\exp(-c_f n)$ exists. Let $\M\in\R^{M\times n}$ with i.i.d. $\mathcal{N}(0,1)$ entries,  $\M_\Delta\in\R^{M\times n}$ be the matrix obtained from $\M$ by independently resampling its last $\Delta n$ columns; and set
\[
\sigma\triangleq \A\bigl(\M\bigr)\in\bincube\qquad\text{and}\qquad \sigma_\Delta\triangleq \A\bigl(\M_\Delta\bigr)\in\bincube.
\]
By a union bound, it is the case that w.p.\,at least $1-2p_f$, $\bigl\|\M\sigma\bigr\|_\infty\le \sqrt{n}$ and $\bigl\|\M_\Delta\sigma_\Delta\bigr\|_\infty\le \sqrt{n}$. Moreover, since the algorithm is online per Definition~\ref{def:online-alg}, it follows that $\sigma(i)=\sigma_\Delta(i)$ for $1\le i\le n-\Delta n$. Hence,
\[
\frac1n \ip{\sigma}{\sigma_\Delta} = \frac1n\sum_{1\le i\le n-\Delta n}\sigma(i)\sigma_\Delta(i) + \frac1n\sum_{ n-\Delta n+1\le i\le n}\sigma(i)\sigma_\Delta(i)\ge 1-2\Delta.
\]
As $(\sigma,\sigma_\Delta)\in \Xi(\Delta)$, we have $\mathbb{P}\bigl[\Xi(\Delta)\ne\varnothing\bigr]\ge 1-2p_f\ge 2\exp(-c_f n)$. On the other hand, $\mathbb{P}\bigl[\Xi(\Delta)\ne\varnothing\bigr]\le\exp\bigl(-\Theta(n)\bigr)$. This is a clear contradiction if $c_f>0$ is small enough. Therefore, it suffices to prove Proposition~\ref{prop:online-alg-fail}.
 \begin{proof}[Proof of Proposition~\ref{prop:online-alg-fail}]
The proof is similar to that of $2-$OGP result, Theorem~\ref{thm:2-OGP}; and is based, in particular, on the first moment method. Let 
\[
N=\sum_{\sigma,\sigma_\Delta:n^{-1}\ip{\sigma}{\sigma_\Delta}\in [1-2\Delta,1]} \ind\left\{\bigl\|\mathcal{M}\sigma\bigr\|_\infty\le\sqrt{n},\bigl\|\mathcal{M}_\Delta\sigma_\Delta\bigr\|_\infty\le\sqrt{n}\right\}
\]
Clearly, $N=|\Xi(\Delta)|$. By Markov's inequality,
\begin{equation}\label{eq:online-markof}
\mathbb{P}\bigl[\Xi(\Delta)\ne\varnothing\bigr] = \mathbb{P}\bigl[N\ge 1\bigr]\le \mathbb{E}[N].
\end{equation}
Thus, it suffices to show $\mathbb{E}[N]=\exp(-\Theta(n))$ for $\Delta>0$ small.
\paragraph{Counting term.} There are $2^n$ choices for $\sigma\in\bincube$. Note that $n^{-1}\ip{\sigma}{\sigma_\Delta}\in[1-2\Delta,1]\iff d_H\bigl(\sigma,\sigma_\Delta\bigr)\le\Delta n$. Thus, having fixed a $\sigma$, there are 
\[
\sum_{k\in\mathbb{N}\cap [0,\Delta n]}\binom{n}{k}\le (1+\Delta n)\cdot \binom{n}{\Delta n}=\exp_2\Bigl(nh(\Delta) + O(\log_2 n)\Bigr)
\]
choices for $\sigma_\Delta\in\bincube$, where we used the fact $\binom{n}{k}\le \binom{n}{\Delta n}$ for any $k\le \Delta n$ (as $\Delta<1/2$) and Stirling's approximation. Thus, 
\begin{equation}\label{eq:prop-counting-term}
    \Bigl|\Bigl\{(\sigma,\sigma_\Delta)\in\bincube\times\bincube:n^{-1}\ip{\sigma}{\sigma_\Delta}\ge 1-2\Delta\Bigr\}\Bigr\}\le \exp_2\Bigl(n + nh(\Delta) +O(\log_2 n)\Bigr).
\end{equation}
\paragraph{Probability term.} Now, fix $\sigma,\sigma_\Delta$ with $n^{-1}\ip{\sigma}{\sigma_\Delta}\ge 1-2\Delta$. Let $R=(Z_1,Z_2,\dots,Z_n)\in\R^n$ and $R_\Delta=(Z_1,Z_2,\dots,Z_{n-\Delta n},Z'_{n-\Delta n+1},\cdots,Z'_n)\in\R^n$ respectively be the first rows of $\M$ and $\M_\Delta$, where $Z_1,\dots,Z_n,Z_{n-\Delta n+1}',\dots,Z_n'$ are i.i.d. standard normal. Using the independence of rows of $\M$ and $\M_\Delta$, we have
\[
\mathbb{P}\Bigl[\bigl\|\M\sigma\bigr\|_\infty\le\sqrt{n},\bigl\|\M_\Delta\sigma_\Delta\bigr\|_\infty\le\sqrt{n}\Bigr] = \mathbb{P}\Bigl[n^{-\frac12}\bigl|\ip{R}{\sigma}\bigr|\le 1,n^{-\frac12}\bigl|\ip{R_\Delta}{\sigma_\Delta}\bigr|\le 1\Bigr]^{\alpha n}.
\]
We next study bivariate normal variables $n^{-\frac12}\ip{R}{\sigma}\distr \mathcal{N}(0,1)$ and $n^{-\frac12}\ip{R_\Delta}{\sigma_\Delta}\distr \mathcal{N}(0,1)$. Note that
\begin{align*}
    \frac1n\mathbb{E}\Bigl[\ip{R}{\sigma}\ip{R_\Delta}{\sigma_\Delta}\Bigr] &=\frac1n\sum_{1\le i\le n-\Delta n}\mathbb{E}\bigl[Z_i^2\sigma(i)\sigma_\Delta(i)\bigr] + \frac1n\sum_{n-\Delta n+1\le i\le n}\underbrace{\mathbb{E}\bigl[Z_iZ_i'\sigma(i)\sigma_\Delta(i)\bigr]}_{=0}\\
    &=\frac1n\sum_{1\le i\le n-\Delta n}\sigma(i)\sigma_\Delta(i) \in[1-2\Delta,1-\Delta]
\end{align*}
since $d_H\bigl(\sigma,\sigma_\Delta\bigr)\le \Delta n$. Letting $\lambda\triangleq n^{-1}\mathbb{E}\Bigl[\ip{R}{\sigma}\ip{R_\Delta}{\sigma_\Delta}\Bigr]/(1-\Delta) \in\left[\frac{1-2\Delta}{1-\Delta},1\right]$, we therefore obtain that $n^{-\frac12}\ip{R}{\sigma},n^{-\frac12}\ip{R_\Delta}{\sigma_\Delta}$ is bivariate normal with parameter $\lambda(1-\Delta)$. Let $p(\rho)\triangleq \mathbb{P}\bigl[(Z_1,Z_2)\in[-1,1]^2\bigr]$
where
$\rho\in[0,1]$ and $(Z_1,Z_2)$ bivariate normal with parameter $\rho$. Using Theorem~\ref{thm:sidak} with $k=2$ and $\lambda_1=\lambda_2=\sqrt{\lambda}$, we thus obtain $\max_{0\le \lambda\le 1}p\bigl(\lambda (1-\Delta)\bigr)=p\bigl(1-\Delta\bigr)$. Hence,
\begin{align}\label{eq:prop-prob-bounds}
    \mathbb{P}\Bigl[\bigl\|\M\sigma\bigr\|_\infty\le\sqrt{n},\bigl\|\M_\Delta\sigma_\Delta\bigr\|_\infty\le\sqrt{n}\Bigr] \le p\bigl(1-\Delta\bigr)^{\alpha n}.
\end{align}
\paragraph{Upper bounding $\mathbb{E}[N]$.} Combining~\eqref{eq:online-markof},~\eqref{eq:prop-counting-term} and~\eqref{eq:prop-prob-bounds}, we obtain
\begin{align*}
\mathbb{P}\bigl[\Xi(\Delta)\ne\varnothing\bigr]\le \mathbb{E}[N]&\le \exp_2\Bigl(n\Bigl(1+h(\Delta)+\alpha\log_2 p\bigl(1-\Delta\bigr)\Bigr)+O(\log_2 n)\Bigr) \\
&\le \exp_2\bigl(nf_1(\Delta,\alpha)+O(\log_2 n)\bigr),
\end{align*}
where $f_1(\Delta,\alpha)=1+h(\Delta)+\alpha\log_2 p(1-\Delta)$, per Lemma~\ref{asm:negativity}. Since $\log_2 p(1-\Delta)<0$, it suffices to consider $\alpha=1.77$. Setting $\Delta$ such that
\[
f_1(\Delta,1.77) = \inf_{x\in [10^{-5},10^{-1}]}f_1(x,1.77),
\]
Lemma~\ref{asm:negativity} ${\rm (a)}$ implies $f_1(\Delta,1.77)<0$. With this choice of $\Delta$, we complete the proof of Proposition~\ref{prop:online-alg-fail}.
\end{proof}

\subsection{Proof of Theorem~\ref{thm:kr-stable}}\label{sec:proof-of-kr}
In this section, we establish Theorem~\ref{thm:kr-stable}. That is, we show that $\KRA$  is stable in the probabilistic sense. We first set the stage. Recall from~\eqref{eq:matrix-interpolate} the interpolation
\[
\overline{\M}(\tau)\triangleq \cos(\tau)\M+\sin(\tau)\M'\in\R^{k\times n},\qquad \tau\in\left[0,\frac{\pi}{2}\right],
\]
where $\M,\M'\in\R^{k\times n}$ are two i.i.d.\,random matrices each with i.i.d. $\mathcal{N}(0,1)$ entries. In particular, $\overline{\M}(\tau)$ has i.i.d. $\mathcal{N}(0,1)$ coordinates for each $\tau\in[0,\pi/2]$. 

Next, denote by $R_1,\dots,R_k\in\R^n$ the rows of $\M$;  and by $C_1,\dots,C_n\in\R^k$ the columns of $\M$. Likewise, let $\overline{R}_1,\dots,\overline{R}_k\in\R^n$ and $\overline{C}_1,\dots,\overline{C}_n\in\R^k$ be the rows and columns of $\overline{\M}(\tau)$, respectively. (Whenever appropriate, we drop $\tau$ for convenience.) As in Theorem~\ref{thm:kr-stable}, set
\[
\sigma =\KRA\bigl(\M\bigr)\in\bincube\qquad\text{and}\qquad \overline{\sigma}=\KRA\Bigl(\overline{\M}(\tau)\Bigr)\in\bincube.
\]
We first establish the following proposition which pertains the $L$ round implementation of Kim-Roche algorithm, where $L\le c\log_{10}\log_{10} n$ for $c>0$ sufficiently small (as opposed to its full implementation).
\begin{proposition}\label{prop:kr-stable}
Let $c>0$ be a sufficiently small constant, and $L\le c\log_{10} \log_{10} n$  be an arbitrary non-negative integer. Define 
\begin{equation}\label{eq:alpha-ell}
\alpha_\ell = \alpha_0\cdot 10^{-\ell},\qquad 1\le \ell \le L\quad\text{with}\quad \alpha_0=0.01;
\end{equation}
 and set $\tau = n^{-2\alpha_0}$.
 Let $\sigma\in\{-1,1\}^{\sum_{0\le \ell \le L}n_\ell}$ and $\overline{\sigma}\in\{-1,1\}^{\sum_{0\le \ell \le L}n_\ell}$ respectively be the outputs generated  by running $L$ rounds of Kim-Roche algorithm on $\mathcal{M}$ and $\overline{\mathcal{M}}(\tau)$ defined in~\eqref{eq:matrix-interpolate}. Define
\begin{equation}\label{eq:J-ell}
J_\ell\triangleq \Bigl\{i\in[n_0+n_1+\cdots+n_\ell]:\sigma_i\ne\overline{\sigma}_i\Bigr\}, \quad 0\le \ell\le L.
\end{equation}
Then, for any $0\le \ell \le L$,
\[
\mathbb{P}\Bigl[\bigl|J_\ell\bigr|\le n^{1-\alpha_\ell}\Bigr]\ge 1-O\left(n^{-\frac{1}{40}+\epsilon}\right),
\]
where $\epsilon>0$ is arbitrary.
\end{proposition}

\subsubsection{Proof of Proposition~\ref{prop:kr-stable}} 
This section is devoted to the proof of
Proposition~\ref{prop:kr-stable}. 
We proceed by establishing several auxiliary results.

\paragraph{Majority is stable.} As a first step, we establish the stability of the majority algorithm.  This algorithm  assigns, to each column $C_j\in\R^k$ of $\mathcal{M}$, the sign of entries in $C_j$. That is, \[
\sigma_j ={\rm sgn}\left(\sum_{1\le  i\le k}\mathcal{M}_{ij}\right)\in\{-1,1\}.
\]
Namely, this algorithm is simply the very first round of $\KRA$ assigning $n_0\le n$ entries of $\sigma\in\bincube$, where $n_0\approx n$.

\begin{lemma}\label{lemma:gauss-cont-maj}
Let $\A_{\rm maj}:\R^{k\times n}\to\bincube$ be the ``majority" algorithm defined above. Recall $\overline{\M}(\tau)$ from~\eqref{eq:matrix-interpolate}. Then,
\[
d_H\Bigl(\A_{\rm maj}(\M),\A_{\rm maj}\bigl(\overline{\M}(\tau)\bigr)\Bigr) \distr {\rm Bin}\Bigl(n,\frac{\tau}{\pi}\Bigr).
\]
\end{lemma}
\begin{proof}[Proof of Lemma~\ref{lemma:gauss-cont-maj}]
Define $I_j$, $1\le j\le n$ by
\[
I_j = \ind\Bigl\{\A_{\rm maj}(\M)_j \ne \A_{\rm maj}\bigl(\overline{\M}(\tau)\bigr)_j\Bigr\}.
\]
Then, $I_j$ are i.i.d.\,Bernoulli. In particular, it suffices to show $I_j\sim{\rm Ber}(\tau/\pi)$. To that end, we study $I_1$. Let $(X_1,\dots,X_k)$ be the first column of $\M$ and $(Y_1,\dots,Y_k)$ be the first column of $\M'$. Furthermore, set 
\[
Z_i\triangleq \cos(\tau)X_i+\sin(\tau)Y_i,\quad 1\le i\le k.
\]
Note that, $I_i=1$ if and only if
\[
{\rm sgn}\left(\sum_{1\le i\le k}X_i\right)\ne {\rm sgn}\left(\sum_{1\le i\le k}Z_i\right).
\]
From symmetry,
\begin{equation}\label{eq:main-pb-bd}
\mathbb{P}\left({\rm sgn}\left(\sum_{1\le i\le k}X_i\right)\ne {\rm sgn}\left(\sum_{1\le i\le k}Z_i\right)\right)=2\mathbb{P}\left(k^{-\frac12}\sum_{1\le i\le k}X_i>0,-k^{-\frac12}\sum_{1\le i\le k}Z_i>0\right).
\end{equation}
Observe that $\mathbb{E}[X_iZ_j] = \cos(\tau)\ind\{i=j\}$.
Hence,
\[
\left(k^{-\frac12}\sum_{1\le i\le k}X_i,-k^{-\frac12}\sum_{1\le i\le k}Z_i\right)\distr\mathcal{N}\left(\begin{bmatrix}0 \\ 0\end{bmatrix},\begin{bmatrix} 1 &-\cos(\tau) \\-\cos(\tau) & 1 \end{bmatrix}\right).
\]

Next, applying Lemma~\ref{lemma:gaussian-orthant}, the probability in~\eqref{eq:main-pb-bd} evaluates to
\[
2\cdot\left(\frac14+\frac{1}{2\pi}\sin^{-1}\left(-\cos(\tau)\right)\right) = \frac12 + \frac{1}{\pi}\sin^{-1}\left(-\sin\left(\frac{\pi}{2}-\tau\right)\right) = \frac{\tau}{\pi}.
\]
 Hence $I_j\distr {\rm Ber}(\tau/\pi)$, $1\le j\le n$ i.i.d. Finally, since $d_H\Bigl(\A_{\rm maj}(\M),\A_{\rm maj}\bigl(\overline{\M}(\tau)\bigr)\Bigr) = \sum_{1\le j\le n}I_j$, the proof of Lemma~\ref{lemma:gauss-cont-maj} is complete. 
\end{proof}

\paragraph{Correlated ensemble is close to the  original.}  Next, assume that for some $T\in\mathbb{N}$, $T$ rounds (of the algorithm) are completed so far. In particular, the algorithm produced $\sigma,\overline{\sigma}\in\{\pm 1\}^{\sum_{0\le j\le T}n_j}$. Recall the variables from Section~\ref{sec:kim-roche1}:
    \[
    \ip{R_i}{\sigma},\quad 1\le i\le k\qquad\text{and}\qquad \ip{\overline{R_i}}{\overline{\sigma}},\quad 1\le i\le k, 
    \]
    where the inner products are defined in $\R^{\sum_{0\le j\le T}n_j}$ and $\overline{R_i}$, $1\le i\le k$ are the rows of $\overline{\M}(\tau)$ appearing in~\eqref{eq:matrix-interpolate}. We show that these ensembles are ``close" to each other in the following sense.
    \begin{lemma}\label{lemma:ip-and-bar-ip-close}
    Let $\alpha>0$ satisfy
    \begin{equation}\label{eq:alpha-lower-bd}
        \alpha\ge 10^{-c\log_{10}(\log_{10} n)}
    \end{equation}
    for a sufficiently small constant $c>0$ and $k=\Theta(n)$. Then with probability at least $1-\exp(-k/3)$,
    \begin{equation}\label{eq:stability-main}
    \sup \sum_{1\le i\le k}\Bigl|\ip{R_i}{\sigma}-\ip{\overline{R}_i}{\overline{\sigma}}\Bigr|\le Ck\sqrt{n}\Bigl(\tau \log_{10} n+n^{-\alpha/2}\Bigr)
    \end{equation}
    for all large enough $n$, where the supremum is over all pairs $(\sigma,J)$, $\sigma\in\bincube$ and $J=\{i\in[n]:\sigma_i\ne\overline{\sigma}_i\}\subset [n]$ with $|J|\le n^{1-\alpha}$. Here, $C>0$ is an absolute constant.
    \end{lemma}
    It is worth noting that due to $\sup$ term, 
    Lemma~\ref{lemma:ip-and-bar-ip-close} provides a uniform control for {\bf any} pair $(\sigma,\overline{\sigma})\in\bincube\times \bincube$ with $d_H(\sigma,\overline{\sigma})\le n^{1-\alpha}$. We will show that this, in particular, captures the outputs of Kim-Roche algorithm.
\begin{proof}[Proof of Lemma~\ref{lemma:ip-and-bar-ip-close}]
We start by observing that for any $\sigma\in\bincube$, if $J=\{i:\sigma_i\ne \overline{\sigma}_i\}$, then
\begin{equation}\label{eq:inner-to-sum}
\ip{R_i}{\sigma}-\ip{\overline{R}_i}{\overline{\sigma}}=\sum_{j\in J}\bigl(R_{ij}+\overline{R}_{ij}\bigr)\sigma_j +\sum_{j\in J^c}\bigl(R_{ij}-\overline{R}_{ij}\bigr)\sigma_j.
\end{equation}
Next, for any {\bf fixed} $\sigma$ and $J$, we have by~\eqref{eq:matrix-interpolate} that
\begin{align}
    \sum_{j\in J}\bigl(R_{ij}+\overline{R}_{ij}\bigr)\sigma_j& \distr  \mathcal{N}\Bigl(0,\Bigl(\sin^2\tau  + \bigl(1+\cos(\tau)\bigr)^2\Bigr)|J|\Bigr)\label{eq:dist-sum-J}\\
    \sum_{j\in J^c}\bigl(R_{ij}-\overline{R}_{ij}\bigr)\sigma_j  &\distr \mathcal{N}\Bigl(0,\Bigl(\sin^2\tau+\bigl(1-\cos(\tau)\bigr)^2\Bigr)|J^c|\Bigr)\label{eq:dist-sum-J^c}.
\end{align}
Using simple bounds, $1-\frac{\tau^2}{2}\le \cos(\tau)\le 1$ and $\sin(\tau)\le \min\{1,\tau\}$, we obtain the following upper bounds
\begin{align*}
    \left(\sin^2\tau+ \left(1+\cos\tau\right)^2\right)|J|&\le 5|J| \\
     \left(\sin^2 \tau +\left(1-\cos\tau\right)^2\right)|J^c| &\le 100n\tau^2
\end{align*}
on the variances of variables appearing in~\eqref{eq:dist-sum-J} and~\eqref{eq:dist-sum-J^c}.

We next set the stage to apply Bernstein's inequality~\cite[Proposition 5.16]{vershynin2010introduction}: for i.i.d. $Z_i\sim\mathcal{N}(0,1)$, $1\le i\le n$, there exists an absolute constant  $A>0$ such that for all $t>0$, 
\begin{equation}\label{eq:bersntein-from-ver}
\mathbb{P}\left(\sum_{1\le i\le n}|Z_i|\ge n\mathbb{E}[|Z_1|]+t\right)\le \exp\left(-\min\left(\frac{t^2}{4nA^2},\frac{t}{2A}\right)\right).
\end{equation}
Fix any absolute constant $C>0$. Note that using the variance upper bound above
\begin{equation}\label{eq:proxy-bern1}
\mathbb{P}\left(\sum_{1\le i\le k}\left|\sum_{j\in J}\left(R_{ij}+\overline{R}_{ij}\right)\sigma_j\right|\ge Ck\sqrt{|J|}\right)\le \mathbb{P}\left(\sum_{1\le i\le k}|Z_i|\ge \frac{Ck\sqrt{|J|}}{\sqrt{5|J|}}\right),
\end{equation}
where $Z_i\distr\mathcal{N}(0,1)$, $1\le i\le k$ i.i.d. Recall that $\mathbb{E}[|Z_i|] = \sqrt{2/\pi}$. Furthermore, for 
\[
t = \left(\frac{C}{\sqrt{5}}-\sqrt{\frac{2}{\pi}}\right)k,
\]
$k=\Theta(n)$ implies that
\[
\min\left(\frac{t^2}{4nA^2},\frac{t}{2A}\right)\ge k
\]
for $C>0$ large enough. Likewise, 
\begin{equation}\label{eq:proxy-bern2}
\mathbb{P}\left(\sum_{1\le i\le k}\left|\sum_{j\in J^c}\left(R_{ij}-\overline{R}_{ij}\right)\sigma_j\right|\ge Ck\sqrt{n\tau^2}\log_{10} n\right)\le\mathbb{P}\left(\sum_{1\le i\le k}|Z_i|\ge \frac{Ck\sqrt{n\tau^2}\log_{10} n}{10\sqrt{n\tau^2}}\right).
\end{equation}
This time, choosing
\[
t = \frac{C}{10}k\log_{10}n - k\sqrt{\frac{2}{\pi}},
\]
and recalling $k=\Theta(n)$, we have
\[
\min\left(\frac{t^2}{4nA^2},\frac{t}{2A}\right)\ge k\log_{10}n
\]
provided $C>0$ is large. Consequently, applying Bernstein's inequality~\eqref{eq:bersntein-from-ver}  to~\eqref{eq:proxy-bern1} and~\eqref{eq:proxy-bern2}, we obtain
\begin{align}
   & \mathbb{P}\left(\sum_{1\le i\le k}\left|\sum_{j\in J}\left(R_{ij}+\overline{R}_{ij}\right)\sigma_j\right|\ge Ck\sqrt{|J|}\right)\le \exp(-k)\label{eq:bernstein-J}\\
   & \mathbb{P}\left(\sum_{1\le i\le k}\left|\sum_{j\in J^c}\left(R_{ij}-\overline{R}_{ij}\right)\sigma_j\right|\ge Ck\sqrt{n\tau^2}\log_{10} n\right)\le \exp(-k\log_{10} n)\label{eq:bernstein-J^c}
\end{align}
for any sufficiently large constant $C>0$.

The bounds above are valid for any such $(\sigma,J)$. We next upper bound the number of all such pairs. Note that, 
\begin{align}
      \Bigl|(\sigma,J):\sigma\in\{-1,1\}^{J},|J|\le n^{1-\alpha}\Bigr|&=\sum_{m\le n^{1-\alpha}}2^m\binom{n}{m}\nonumber\\ 
      &\le \sum_{m\le n^{1-\alpha}}(2n)^m\nonumber \\
      &\le n^{1-\alpha}(2n)^{n^{1-\alpha}}\nonumber \\
      &\le \exp\bigl(C'n^{1-\alpha}\log_{10} n\bigr)\label{eq:u-bd-J},
\end{align}
for some absolute $C'>0$. Likewise,
\begin{align}
     \Bigl|(\sigma,J):\sigma\in\{-1,1\}^{J^c},|J|\le n^{1-\alpha}\Bigr| &=\sum_{m\le n^{1-\alpha}}2^{n-m}\binom{n}{m}\nonumber \\
     &\le 2^n\sum_{0\le m\le n}2^{-m}\binom{n}{m}\nonumber\\
     &=\exp\bigl(n\ln 3\bigr)\label{eq:u-bd-J^c},
\end{align}
where we used the binomial theorem, $\sum_{0\le m\le n}2^{-m}\binom{n}{m} = 3^n/2^n$.

We now prepare the stage to take union bounds. Note that since $k=\Theta(n)$, $k\log_{10}n = \Theta\bigl(n\log n\bigr) = \omega\bigl(n\ln 3\bigr)$. In particular, the cardinality term appearing in~\eqref{eq:u-bd-J^c} is dominated by the corresponding probability term~\eqref{eq:bernstein-J^c}. Next, we compare the (order of) cardinality term~\eqref{eq:u-bd-J} with the corresponding probability term~\eqref{eq:bernstein-J}. Note that
\[
10^{-c\log_{10}(\log_{10}n)} = \bigl(\log_{10}n\bigr)^{-c}.
\]
Employing this and the lower bound~\eqref{eq:alpha-lower-bd} on $\alpha$, we obtain
\begin{align*}
    n^{1-\alpha}\log_{10}n &\le n^{1-10^{-c\log_{10}(\log_{10} n)}}\log_{10}n \\
    &=n\cdot \underbrace{\exp\left(\frac{1}{\log_{10} e}\log_{10}(\log_{10}n) -\frac{1}{\log_{10}e}\bigl(
    \log_{10}n\bigr)^{1-c}\right)}_{\text{$=o(1)$, provided $c<1$}}\\
    &=o(n).
\end{align*}
Since $k=\Theta(n)$ and $c>0$ is sufficiently small, it follows that the probability term appearing in~\eqref{eq:bernstein-J} dominates the cardinality term~\eqref{eq:u-bd-J}.

Taking union bounds, we obtain
\begin{equation}\label{eq:sup-1}
     \mathbb{P}\left(\sup_{\sigma,J}\sum_{1\le i\le k}\left|\sum_{j\in J}\left(R_{ij}+\overline{R}_{ij}\right)\sigma_j\right|\ge Ck\sqrt{n^{1-\alpha}}\right)\le \exp(-k/2)
\end{equation}
and
\begin{equation}\label{eq:sup-2}
    \mathbb{P}\left(\sup_{\sigma,J}\sum_{1\le i\le k}\left|\sum_{j\in J^c}\left(R_{ij}-\overline{R}_{ij}\right)\sigma_j\right|\ge Ck\sqrt{n\tau^2}\log_{10} n\right)\le \exp\bigl(-k\log_{10} n/2\bigr).
\end{equation}
Finally, combining~\eqref{eq:sup-1} and~\eqref{eq:sup-2} via a union bound, we conclude that~\eqref{eq:stability-main}  holds with probability at least $1-\exp(-k/3)$, completing the proof of Lemma~\ref{lemma:ip-and-bar-ip-close}.
\end{proof}
    
    \paragraph{Distribution of inner products.}
Next, as an auxiliary step, we study the parameters of distribution of $\ip{R_i}{\sigma}$, where $\sigma$ is generated by the application of majority protocol: $\sigma_j = {\rm sgn}\left(\ip{C_j}{e}\right)$, where $e$ is the vector of all ones. Note that 
\[
\ip{R_i}{\sigma} = \sum_{1\le j\le n}X_{ij}\sigma_j = \sum_{1\le j\le n}X_{ij}{\rm sgn}\left(\sum_{1\le i\le k}X_{ij}\right) = \sum_{1\le j\le n}X_{ij}{\rm sgn}\left(\frac{1}{\sqrt{k}}\sum_{1\le i\le k}X_{ij}\right).
\]
Notice that for any fixed row index $i$, the collection $X_{ij}\sigma_j$, $1\le j\le n$ is i.i.d.\,We now compute the relevant statistics. 
\begin{lemma}\label{thm:mean-Xij-sigma-j}
For any  $1\le i\le k$ and $1\le j\le n$, 
\[
\mathbb{E}\bigl[X_{ij} \sigma_j\bigr] = \sqrt{\frac{2}{\pi k}}.
\]
Consequently, for any distinct $(i,j),(i',j')\in[k]\times [n]$,
\[
\mathbb{E}\bigl[X_{ij}\sigma_j X_{i'j'}\sigma_{j'}\bigr] = \ind\{j\ne j'\}\frac{2}{\pi k}
\]
\end{lemma}
\begin{proof}[Proof of Lemma~\ref{thm:mean-Xij-sigma-j}]
For simplicity, we drop the index $i$ below whenever convenient. Observe that $\mathbb{P}(\sigma_1=+1)=\mathbb{P}(\sigma_1=-1)=\frac12$. We then have
\begin{align*}
    \mathbb{E}\bigl[X_1\sigma_1\bigr] &= \frac12\Bigl(\mathbb{E}\bigl[X_1|\sigma_1=+1\bigr] - \mathbb{E}\bigl[X_1|\sigma_1=-1\bigr]\Bigr) \\
    &=\frac12\left(\mathbb{E}\left[X_1\Bigg|\frac{1}{\sqrt{k}}\sum_{1\le j\le k}X_j\ge 0\right]+\mathbb{E}\left[-X_1\Bigg|-\frac{1}{\sqrt{k}}\sum_{1\le j\le k}X_j\ge 0\right]\right) \\
    &=\mathbb{E}\left[X_1\Bigg|\frac{1}{\sqrt{k}}\sum_{1\le j\le k}X_j\ge 0\right] = \sqrt{\frac{2}{\pi k}},
\end{align*}
where we applied Lemma~\ref{lemma:bivariate-cond-quadrant} for the bivariate normal
\[
\left(X_1,\frac{1}{\sqrt{k}}\sum_{1\le j\le k}X_j\right)\distr\mathcal{N}\left(\begin{bmatrix}0 \\ 0\end{bmatrix} ,\begin{bmatrix} 1&\frac{1}{\sqrt{k}} \\ \frac{1}{\sqrt{k}} & 1\end{bmatrix}\right).
\]
Having established the claim for $\mathbb{E}[X_{ij}\sigma_j]$, the rest is straightforward. Take $(i,j)\ne (i',j')$. Note that if $j=j'$, we are done since $X_{ij}$ and $X_{i'j}$ are independent with mean zero. Assume $j\ne j'$. Then, $X_{ij}\sigma_j$ and $X_{i'j'}\sigma_{j'}$ are i.i.d. This completes the proof of Lemma~\ref{thm:mean-Xij-sigma-j}.
\end{proof}
\paragraph{Thresholding suffices to find $k_j$ indices.} We now establish that for finding the $k_j$ (row) indices to be used in round $j$ of the algorithm, it suffices to threshold the inner products. This is a consequence of the following concentration result. 
\begin{lemma}\label{lemma:main-card}
Suppose that $0<c<\log_{10} 2$ is an arbitrary constant, and $1\le T\le c\log_{10}\log_{10} n$ is an arbitrary integer. Let $\sigma\in\R^{S_T}$ for $S_T= \sum_{0\le s\le T}n_s$ be the output of $\KRA$ at the end of $T{\rm th}$ round. Then, for  any $x\in\mathbb{R}$, and $\epsilon>0$,
\[
\mathbb{E}\left[\left(\Bigl|\bigl\{1\le i\le k:\ip{R_i}{\sigma}<x\bigr\}\Bigr| - k\Phi\left(\frac{1}{\sqrt{S_T}}\left(x-\sum_{0\le s\le T}n_s\sqrt{\frac{2k_s}{\pi k^2}}\right) 
\right)\right)^2\right]\le O\left(n^{\frac{39}{20}+\epsilon}\right),
\]
where $\Phi(t) = \mathbb{P}(Z\le t)$ for $Z\sim\mathcal{N}(0,1)$,  $k_s$ is defined in~\eqref{eq:k-j} and $n_s$ is defined  in~\eqref{eq:n-j}.
\end{lemma}
\begin{proof}[Proof of Lemma~\ref{lemma:main-card}]
We consider
\begin{equation}\label{eq:main-quantities}
\mathbb{P}\Bigl(\ip{R_i}{\sigma}<x\Bigr)\quad\text{and}\quad \mathbb{P}\Bigl(\ip{R_i}{\sigma}<x,\ip{R_{i'}}{\sigma}<x\Bigr)
\end{equation}
for $1\le i,i'\le k$ and $i\ne i'$. Define the running sums
\begin{equation}\label{eq:running-sum}
    S_t = \sum_{0\le j\le t}n_j,\quad 0\le t\le T.
\end{equation}
That is, $S_t$ is the number of entries of $\sigma$ assigned at  the end of round $t$. Suppose that  the algorithm run for $T$ rounds, where $1\le T\le c\log_{10}\log_{10} n$ for $c>0$ sufficiently small. With this notation,
\[
\ip{R_i}{\sigma} = \sum_{1\le j\le S_T}X_{ij}\sigma_j  = \sum_{1\le j\le n_0}X_{ij}\sigma_j + \sum_{1\le t\le T} \sum_{S_{t-1}+1\le j\le S_t}X_{ij}\sigma_j.
\]
To analyze the distribution of this value, let us define
\begin{equation}\label{eq:U-i-i'}
U(i,i')\triangleq  \sum_{0\le t\le T} \sum_{S_{t-1}+1\le j\le S_t}X_{ij}\widetilde{\sigma}_j + \sum_{0\le t\le T}n_t \mu_t,
\end{equation}
where for $S_{t-1}+1\le j\le S_t$, and  $0\le t\le T$ (with the convention $S_{-1}\triangleq 0$, $\mathcal{I}_0\triangleq[k]$ and $k_0\triangleq k$) 
\[
\widetilde{\sigma}_j \triangleq   {\rm sgn}\left(\sum_{\ell\in\mathcal{I}_t\setminus\{i,i'\} }X_{\ell j}\right) \qquad\text{and}\qquad \mu_t=\mathbb{E}\bigl[X_{ij}\sigma_j\bigr].
\]
We suppress the dependence of $\widetilde{\sigma}$ on $i,i'$ for convenience. Observe that $\widetilde{\sigma}_j$ is independent of all $X_{ij}$ and $X_{i'j}$.

We now compute $\mu_t$ appearing above. To that end, we remind the reader the (index) set $\mathcal{I}_t$ for convenience: for any $1\le i\le k$, $i\in\mathcal{I}_t$ iff the (partial) inner product, $\ip{R_i}{\sigma}$, is among the smallest $k_t$ (partial) inner products $\ip{R_j}{\sigma}$, $1\le j\le k$. 

Next, note that $\sigma_j$ (the sign assigned to column $j$) is obtained by taking a majority vote in $k_t\times n_t$ submatrix with row indices prescribed by $\mathcal{I}_t$. Note that if $i\notin\mathcal{I}_t$, $X_{ij}$ and $\sigma_j$ are independent. Furthermore, given any $i$, 
\[
\mathbb{P}(i\in \mathcal{I}_t) = \binom{k}{k_t}^{-1}\binom{k-1}{k_t-1} = \frac{k_t}{k}. 
\]
from symmetry. Consequently,
\begin{align}
    \mu_t &= \mathbb{E}\bigl[X_{ij}\sigma_j\mid i\in\mathcal{I}_t\bigr]\mathbb{P}(i\in\mathcal{I}_t) + 
    \underbrace{\mathbb{E}\bigl[X_{ij}\sigma_j\mid i\notin\mathcal{I}_t\bigr]}_{=0}\mathbb{P}(i\notin\mathcal{I}_t)\LinesNotNumbered \\
    &=\sqrt{\frac{2}{\pi k_t}}\cdot \frac{k_t}{k} = \sqrt{\frac{2k_t}{\pi k^2}},\label{eq:eq-mu-t}
\end{align}
where we used Lemma~\ref{thm:mean-Xij-sigma-j} to invoke $\mathbb{E}\bigl[X_{ij}\sigma_j\mid i\in\mathcal{I}_t\bigr]=\sqrt{\frac{2}{\pi k_t}}$.

Define now \[
\Delta_{i,i'}\triangleq \ip{R_i}{\sigma} - U(i,i').
\]
Since $\mathbb{E}\bigl[X_{ij}\widetilde{\sigma}_j\bigr]=0$, we obtain $\mathbb{E}[\Delta_{i,i'}]=0$  from the choice of $\mu_t$, $0\le t\le T$. 
We now claim
\begin{lemma}\label{lemma:variance-up-bd}
\begin{equation}\label{eq:var-star}
{\rm Var}\Bigl(\Delta_{i,i'}\Bigr) = O\left(\sum_{0\le t\le T}n_t\cdot k_t^{-\frac14}\right).
\end{equation}
Moreover, if $T\le c\log_{10}\log_{10} n$ with $c<\log_{10} 2$, then
\begin{equation}\label{eq:var-star-star}
{\rm Var}\Bigl(\Delta_{i,i'}\Bigr) = O\left(n^{\frac34+\epsilon}\right)
\end{equation}
for any $\epsilon>0$. 
\end{lemma}
\begin{proof}[Proof of Lemma~\ref{lemma:variance-up-bd}]
Note that for $S_{t-1}+1\le j  \le S_t$,  $\sigma_j$ is a function of a $k_t\times n_t$ submatrix with i.i.d. $\mathcal{N}(0,1)$ entries that has not been inspected yet. Hence
\begin{align*}
    {\rm Var}\Bigl(\Delta_{i,i'}\Bigr) &={\rm Var}\left(\sum_{0\le t\le T}\sum_{S_{t-1}+1\le j\le S_t} X_{ij}\left(\sigma_j-\widetilde{\sigma}_j\right)\right)\\
    &=\sum_{0\le t\le T}n_t{\rm Var}\Bigl(X_{ij}\left(\sigma_j-\widetilde{\sigma}_j\right)\Bigr)\\
    &\le \sum_{0\le  t\le T}n_t\mathbb{E}\Bigl[X_{ij}^2\left(\sigma_j-\widetilde{\sigma}_j\right)^2\Bigr]\\
    &\le \sum_{0\le t\le T} n_t\sqrt{\mathbb{E}\left[X_{ij}^4\right]\mathbb{E}\left[\left(\sigma_j-\widetilde{\sigma}_j\right)^4\right]}\\
    &=O\left(\sum_{0\le t\le T}n_t \sqrt{\mathbb{P}\left(\sigma_j\ne\widetilde{\sigma}_j\right)}\right),
\end{align*}
where  the second line uses the fact that for any fixed $t$ and $S_{t-1}+1\le j\le S_t$ the distributions of $X_{ij}\bigl(\sigma_j-\widetilde{\sigma}_j\bigr)$ are identical;
     the third line uses ${\rm Var}(U)\le \mathbb{E}[U^2]$; and
   the fourth line uses Cauchy-Schwarz inequality.

We now show that for $S_{t-1}+1\le j\le S_t$, $0\le t\le T$, 
\[
\mathbb{P}\Bigl(\sigma_j\ne\widetilde{\sigma}_j\Bigr) =O\left(\frac{1}{\sqrt{k_t}}\right)
\]
which will establish Lemma~\ref{lemma:variance-up-bd}. We have that
\begin{align*}
\mathbb{P}\Bigl(\sigma_j\ne\widetilde{\sigma}_j\Bigr)&\le \underbrace{\mathbb{P}\Bigl(\sigma_j\ne\widetilde{\sigma}_j\bigl\vert i,i'\notin \mathcal{I}_t\Bigr)}_{=0} +\mathbb{P}\Bigl(\sigma_j\ne\widetilde{\sigma}_j\Bigl\vert |\mathcal{I}_t\cap \{i,i'\}|=1\Bigr)+\mathbb{P}\Bigl(\sigma_j\ne\widetilde{\sigma}_j\bigl\vert i,i'\in  \mathcal{I}_t\Bigr). 
\end{align*}
Recall now from Lemma~\ref{lemma:gaussian-orthant} that for a pair $(S_1,S_2)$ of bivariate normal random variables $S_1,S_2\distr \mathcal{N}(0,1)$ with parameter $\rho$, \[
\mathbb{P}\Bigl({\rm sgn}(S_1)\ne {\rm sgn}(S_2)\Bigr) 
= \frac12-\frac1\pi \sin^{-1}(\rho),
\]
which, in particular, is a decreasing function of $\rho$. Now, 
\[
\mathbb{P}\Bigl(\sigma_j\ne\widetilde{\sigma}_j\Bigl\vert |\mathcal{I}_t\cap \{i,i'\}|=1\Bigr) = \mathbb{P}\left({\rm sgn}\left(\frac{1}{\sqrt{k_t-1}}\sum_{1\le i\le k_t-1}Z_i\right)\ne {\rm sgn}\left(\frac{1}{\sqrt{k_t}}\sum_{1\le i\le k_t}Z_i\right)\right),
\]
where $Z_i$, $1\le i\le k_t$ are i.i.d. $\mathcal{N}(0,1)$. Setting $S_1=(k_t-1)^{-\frac12}\sum_{1\le i\le k_t-1}Z_i$ and $S_2=k_t^{-\frac12}\sum_{1\le i\le k_t}Z_i$, we find that $(S_1,S_2)$ is a bivariate normal with parameter $\sqrt{1-\frac{1}{k_t}}$. Likewise, a similar argument yields that $\mathbb{P}\Bigl(\sigma_j\ne\widetilde{\sigma}_j\lvert i,i'\in\mathcal{I}_t\Bigr)= \mathbb{P}\bigl({\rm sgn}(S_1')\ne{\rm sgn}(S_2')\bigr)$, where $S_1',S_2'\distr \mathcal{N}(0,1)$ is a bivariate normal with parameter $\sqrt{1-\frac{2}{k_t}}$. Consequently,
\begin{align*}
\mathbb{P}\Bigl(\sigma_j\ne\widetilde{\sigma}_j\Bigr)&\le  \mathbb{P}\Bigl(\sigma_j\ne\widetilde{\sigma}_j\Bigl\vert |\mathcal{I}_t\cap \{i,i'\}|=1\Bigr)+\mathbb{P}\Bigl(\sigma_j\ne\widetilde{\sigma}_j\bigl\vert i,i'\in  \mathcal{I}_t\Bigr) \le 2\mathbb{P}\Bigl(\sigma_j\ne\widetilde{\sigma}_j\bigl\vert i,i'\in  \mathcal{I}_t\Bigr).
\end{align*}
Next, for fixed $i\ne i'$;  set
$S\triangleq\sum_{\ell\in\mathcal{I}_t\setminus\{i,i'\}}X_{\ell j}$. We then have,
\begin{align*}
\mathbb{P}\Bigl(\sigma_j\ne\widetilde{\sigma}_j\bigl\vert i,i'\in  \mathcal{I}_t\Bigr) &= 2\mathbb{P}\bigl(S+X_i+X_{i'}\ge 0,S\le 0\bigr) \\
&= 2\mathbb{P}\left(\frac{1}{\sqrt{k_j}}(S+X_i+X_{i'})\ge 0,-\frac{1}{\sqrt{k_j-2}}S\ge 0\right) \\
& = \frac12-\frac1\pi\sin^{-1}\left(\sqrt{1-\frac{2}{k_j}}\right)\\
& = \frac12-\frac1\pi \sin^{-1}\left(1-\frac{1}{k_j}+O\left(\frac{1}{k_j^2}\right)\right)\\
&=\frac12-\frac1\pi\left(\frac{\pi}{2} - O\left(\frac{1}{\sqrt{k_j}}\right)\right) = O\left(\frac{1}{\sqrt{k_j}}\right),
\end{align*}
where the first line uses symmetry; the third line uses Lemma~\ref{lemma:gaussian-orthant}; and the fourth line uses $\sqrt{1-x} = 1-\frac{x}{2} + O(x^2)$ and $\sin^{-1}(1-x) =  \frac{\pi}{2} - \sqrt{2x}+ O(x^{3/2})$. 
Hence, we  established
\[
{\rm Var}\Bigl(\Delta_{i,i'}\Bigr)  = O\left(\sum_{0\le t\le T}n_t\cdot k_t^{-\frac14}\right),
\]
where $O(\cdot)$ only hides absolute constants. Namely,~\eqref{eq:var-star} holds. Next, let $c<\log_{10} 2$. Then we claim 
\[
\sum_{0\le j\le c\log_{10}\log_{10} n} \frac{n_j}{\sqrt{k_j}} = O\left(n^{\frac34+\epsilon}\right)
\]
for any $\epsilon>0$, which will yield~\eqref{eq:var-star-star}.

In the remainder, we omit floor/ceiling operators whenever convenient. Note that for $k_j$ defined in~\eqref{eq:k-j},
\[
k_j = 2\lfloor(1/2)f_j^3\cdot n\rfloor +1\ge f_j^3 n -1.
\] 
Moreover, for $n_j$ appearing in~\eqref{eq:n-j},
\[
n_j = \left\lfloor \frac{n}{A}\sum_{0\le i\le j}f_i\right\rfloor - \left\lfloor \frac{n}{A}\sum_{0\le i\le j-1}f_i\right\rfloor \le \frac{n}{A}f_j+1\le nf_j+1,
\]
using the fact that $A\ge 1$ per~\eqref{eq:kr-A}. Next, using $f_j\le 1$, we have
\[
    nf_j\ge n f_j^3 \ge n\cdot \exp_{10}\Bigl(-3\cdot 2^{c\log_{10} \log_{10} n}\Bigr) = \exp_{10}\left(\log_{10} n -3\cdot\left(\log_{10} n\right)^{c'}\right)=\omega(1)
\]
where $c'=c\log_{10}2<(\log_{10}2)^2<1$. Here, we used the fact that $j\le L$, where $L=c\log_{10}\log_{10}n$ appears in Proposition~\ref{prop:kr-stable} with $c>0$ small enough.

Namely,  $nf_j,nf_j^3=\omega(1)$. Employing this, together with $n_j\le nf_j$ and $k_j\ge nf_j^3-1$ established above, we have
\begin{align*}
 O\left(\sum_{0\le j\le c\log_{10}\log_{10} n}n_j\cdot k_j^{-\frac14} \right) &= O\left(\sum_{0\le j\le  c\log_{10}\log_{10}n}nf_j\cdot (nf_j^3)^{-\frac14}\right)\\
  &=O\left(\sum_{0\le j\le c\log_{10}\log_{10} n}n^{\frac34}f_j^{\frac14}\right)\\
  &=O\left(\log_{10}\log_{10} n\cdot n^{\frac34}\right)\\
  &=O\left(n^{\frac34+\epsilon}\right),\quad\forall \epsilon>0,
\end{align*}
where the first line uses the bounds $n_j\le nf_j+1$ and $k_j\ge nf_j^3-1$ and the  third line uses the fact $f_j=o(1)$. 
This concludes the proof of Lemma~\ref{lemma:variance-up-bd}.
\end{proof}

Lemma~\ref{lemma:variance-up-bd} yields
\begin{equation}\label{eq:var-up-bd2}
\mathbb{E}\bigl[\Delta_{i,i'}^2\bigr] = {\rm Var}\bigl(\Delta_{i,i'}\bigr)  = O\left(n^{\frac34+\epsilon}\right)
\end{equation}
for any $\epsilon>0$. 
Now, recall that 
\[
\ip{R_i}{\sigma} = U(i,i')+\Delta_{i,i'}\qquad\text{and}\qquad \ip{R_{i'}}{\sigma} = U(i',i)+\Delta_{i',i}.
\]
In particular, using Chebyshev's inequality and~\eqref{eq:var-up-bd2}, we obtain
\begin{equation}\label{eq:chebyshevv}
\mathbb{P}\Bigl(\bigl|\Delta_{i',i}\bigr|>k^{\frac25}\Bigr)=\mathbb{P}\Bigl(\bigl|\Delta_{i,i'}\bigr|>k^{\frac25}\Bigr) \le k^{-\frac45}\mathbb{E}\Bigl[\Delta_{i,i'}^2\Bigr] = O\left(n^{-\frac{1}{20}+\epsilon}\right),
\end{equation}
for any $\epsilon>0$, as $k=\Theta(n)$. Equipped with these,
\begin{align*}
    \mathbb{P}\Bigl(\ip{R_i}{\sigma}<x,\ip{R_{i'}}{\sigma}<x\Bigr) & = \mathbb{P}\Bigl(U(i,i')+\Delta_{i,i'}<x,U(i',i)+\Delta_{i',i}<x\Bigr)\\
    &\le \mathbb{P}\Bigl(U(i,i')<x+k^{\frac25} ,U(i',i)<x+k^{\frac25}\Bigr) + \mathbb{P}\Bigl(\bigl|\Delta_{i,i'}\bigr|>k^{\frac25}\Bigr) + \mathbb{P}\Bigl(\bigl|\Delta_{i',i}\bigr|>k^{\frac25}\Bigr)\\
    & = \mathbb{P}\Bigl(U(i,i')<x+k^{\frac25} ,U(i',i)<x+k^{\frac25}\Bigr)+ O\left(n^{-\frac{1}{20}+\epsilon}\right),
\end{align*}
where the last line uses~\eqref{eq:chebyshevv}. Next, we establish an auxiliary lemma. 
\begin{lemma}\label{lemma:iid-gaussianity}
The random variables
\[
\frac{1}{\sqrt{S_T}}\sum_{1\le j\le S_T}X_{ij}\widetilde{\sigma}_j\qquad\text{and}\qquad \frac{1}{\sqrt{S_T}}\sum_{1\le j\le S_T}X_{i'j}\widetilde{\sigma}_j
\]
are i.i.d.\,standard normal.
\end{lemma}
\begin{proof}[Proof of Lemma~\ref{lemma:iid-gaussianity}]
Note that $X_{ij}$, $1\le j\le S_T$ and $X_{i'j}$, $1\le j\le S_T$ are i.i.d. $\cN(0,1)$. Moreover, $\widetilde{\sigma}_j$, $1\le j\le S_T$ is an i.i.d.\,collection with
\begin{equation}\label{eq:rademacher}
\mathbb{P}\bigl[\widetilde{\sigma}_j = 1\bigr] = \frac12=\mathbb{P}\bigl[\widetilde{\sigma}_j = -1\bigr]
\end{equation}
and that $\widetilde{\sigma}_j$ are independent of $X_{ij}$ and $X_{i'j}$. Next, we show if $X\distr \cN(0,1)$ and $\widetilde{\sigma}$ has the distribution~\eqref{eq:rademacher} and is independent of $X$, then $\widetilde{\sigma} X\distr\cN(0,1)$. To see this, we rely on characteristic functions: for any $t\in\R$, 
\begin{align*}
    \mathbb{E}\bigl[e^{it\widetilde{\sigma}X}\bigr] &=  \mathbb{E}\bigl[e^{it\widetilde{\sigma}X}\big\lvert\widetilde{\sigma}=1\bigr]\mathbb{P}\bigl[\widetilde{\sigma}=1\bigr] + \mathbb{E}\bigl[e^{it\widetilde{\sigma}X}\big\lvert\widetilde{\sigma}=-1\bigr]\mathbb{P}\bigl[\widetilde{\sigma}=-1\bigr]\\
    &=\frac12\left(\mathbb{E}\bigl[e^{itX}\bigr] + \mathbb{E}\bigl[e^{-itX}\bigr]\right)\\
    &=\exp\bigl(-\frac{t^2}{2}\bigr).
\end{align*}
Using L{\'e}vy's inversion theorem~\cite[Section~16.6]{williams1991probability}, it follows that $\widetilde{\sigma}X\distr\cN(0,1)$. Applying this fact, since $X_{ij}\widetilde{\sigma}_j$, $1\le j\le S_T$ is an i.i.d. $\cN(0,1)$  collection, we deduce
\[
\frac{1}{\sqrt{S_T}}\sum_{1\le j\le S_T}X_{ij}\widetilde{\sigma}_j\distr \cN(0,1)\qquad\text{and}\qquad \frac{1}{\sqrt{S_T}}\sum_{1\le j\le S_T}X_{i'j}\widetilde{\sigma}_j\distr \cN(0,1).
\]
Finally, we show $ X_{ij}\widetilde{\sigma}_j\perp X_{i'j}\widetilde{\sigma}_j$, which will yield Lemma~\ref{lemma:iid-gaussianity}. Once again, we rely on L{\'e}vy's inversion theorem. Let \[
\boldsymbol{Z} = \left(X_{ij}\widetilde{\sigma}_j,X_{i'j}\widetilde{\sigma}_j\right)\qquad\text{and}\qquad
\boldsymbol{t} = (t_1,t_2).
\]
We have
\begin{align*}
    \mathbb{E}\bigl[e^{i\boldsymbol{t}^T\boldsymbol{Z}}\bigr] & = \mathbb{E}\left[e^{i\widetilde{\sigma}_j\bigl(t_1X_{ij}+t_2X_{i'j}\bigr)}\right] \\
    &=\mathbb{E}\left[e^{i\widetilde{\sigma}_j\bigl(t_1X_{ij}+t_2X_{i'j}\bigr)}\Big\lvert \widetilde{\sigma}_j=1\right]\mathbb{P}\bigl[\widetilde{\sigma}_j=1\bigr] + \mathbb{E}\left[e^{i\widetilde{\sigma}_j\bigl(t_1X_{ij}+t_2X_{i'j}\bigr)}\Big\lvert \widetilde{\sigma}_j=-1\right]\mathbb{P}\bigl[\widetilde{\sigma}_j=-1\bigr]\\
    &=\exp\left(-\frac{t_1^2+t_2^2}{2}\right),
\end{align*}
where we used the fact $t_1X_{ij}+t_2X_{i'j}\distr \cN\bigl(0,t_1^2+t_2^2\bigr)$. Clearly, 
\[
\mathbb{E}\bigl[e^{i\boldsymbol{t}^T\boldsymbol{Z}}\bigr] = \mathbb{E}\bigl[e^{it_1X_{ij}\widetilde{\sigma}_j}\bigr]\mathbb{E}\bigl[e^{it_2X_{i'j}\widetilde{\sigma}_j}\bigr],
\]
since $t_1X_{ij}\widetilde{\sigma}_j\distr \cN(0,t_1^2)$ and $t2_2X_{i'j}\widetilde{\sigma}_j\distr \cN(0,t_2^2)$. Since $t_1,t_2\in\R$ are arbitrary, we complete the proof of Lemma~\ref{lemma:iid-gaussianity} by appealing once again to L{\'e}vy's inversion theorem.
\end{proof}
Denote  $\Phi(t)=\mathbb{P}[\mathcal{N}(0,1)\le t]$. For $S_T=\sum_{0\le t\le T}n_t$, we have
\begin{align*}
&\mathbb{P}\Bigl(U(i,i')<x+k^{\frac25} ,U(i',i)<x+k^{\frac25}\Bigr)\\
&= \mathbb{P}\left(\frac{1}{\sqrt{S_T}}\sum_{1\le j\le S_T}X_{ij}\widetilde{\sigma}_j,\frac{1}{\sqrt{S_T}}\sum_{1\le j\le S_T}X_{i'j}\widetilde{\sigma}_j<\frac{x}{\sqrt{S_T}}-\frac{1}{\sqrt{S_T}}\sum_{0\le t\le T}n_t\sqrt{\frac{2k_t}{\pi k^2}}+\frac{k^{\frac25}}{\sqrt{S_T}} \right) \\
&=\Phi\left(\frac{x}{\sqrt{S_T}}-\frac{1}{\sqrt{S_T}}\sum_{0\le t\le T}n_t\sqrt{\frac{2k_t}{\pi k^2}}+\frac{k^{\frac25}}{\sqrt{S_T}}\right)^2,
\end{align*}
where the second line uses the expressions for $U(i,i')$ and $U(i',i)$ per~\eqref{eq:U-i-i'} and for $\mu_t$ per~\eqref{eq:eq-mu-t}; and 
 the last line uses Lemma~\ref{lemma:iid-gaussianity}.
 Observe next that $\Phi(\cdot)$ is trivially $1-$Lipschitz: 
\[
\bigl|\Phi(t_1)-\Phi(t_2)\bigr| = \int_{\min\{t_1,t_2\}}^{\max\{t_1,t_2\}}\frac{1}{\sqrt{2\pi}}e^{-\frac{t^2}{2}}\;dt \le \bigl|t_1-t_2\bigr|.
\]
With this, we have
\[
\Phi\left(\frac{x}{\sqrt{S_T}}-\frac{1}{\sqrt{S_T}}\sum_{0\le t\le T}n_t\sqrt{\frac{2k_t}{\pi k^2}}+\frac{k^{\frac25}}{\sqrt{S_T}}\right) - \Phi\left(\frac{x}{\sqrt{S_T}}-\frac{1}{\sqrt{S_T}}\sum_{0\le t\le T}n_t\sqrt{\frac{2k_t}{\pi k^2}}\right)\le \frac{k^{\frac25}}{\sqrt{S_T}}.
\]
Moreover, note that $\Theta(n) = n_0\le S_T\le n$ yields $S_T=\Theta(n)$, hence in particular $k^{2/5}/\sqrt{S_T} = \Theta(n^{-1/10})$ as $k=\Theta(n)$ too. Consequently,
\begin{align*}
\Phi\left(\frac{x}{\sqrt{S_T}}-\frac{1}{\sqrt{S_T}}\sum_{0\le t\le T}n_t\sqrt{\frac{2k_t}{\pi k^2}}+\frac{k^{\frac13}}{\sqrt{S_T}}\right)^2
\le \Phi\left(\frac{x}{\sqrt{S_T}}-\frac{1}{\sqrt{S_T}}\sum_{0\le t\le T}n_t\sqrt{\frac{2k_t}{\pi k^2}}\right)^2 + O\left(n^{-\frac{1}{10}}\right).
\end{align*}
Likewise, using inequality $\mathbb{P}(\mathcal{E}_1\cap\mathcal{E}_2)\ge \mathbb{P}(\mathcal{E}_1)-\mathbb{P}(\mathcal{E}_2^c)$ valid for all events $\mathcal{E}_1,\mathcal{E}_2$, we have
\begin{align*}
\mathbb{P}\Bigl(\ip{R_i}{\sigma}<x,\ip{R_{i'}}{\sigma}<x\Bigr) &\ge \mathbb{P}\Bigl(U(i,i')<x-k^{\frac25},U(i',i)<x-k^{\frac25},\Delta_{i,i'}<k^{\frac13},\Delta_{i',i}<k^{\frac25}\Bigr) \\
&\ge \mathbb{P}\Bigl(U(i,i')<x-k^{\frac25},U(i',i)<x-k^{\frac25}\Bigr)  - \mathbb{P}\Bigl(\Delta_{i,i'}>k^{\frac25}\quad\text{or}\quad \Delta_{i',i}>k^{\frac25}\Bigr) \\ 
&\ge \Phi\left(\frac{x}{\sqrt{S_T}}-\frac{1}{\sqrt{S_T}}\sum_{0\le t\le T}n_t\sqrt{\frac{2k_t}{\pi k^2}}\right)^2 - O\left(n^{-\frac{1}{20}+\epsilon}\right),
\end{align*}
where in the last line we have once again used~\eqref{eq:chebyshevv}.
Combining these, we arrive at
\begin{align}
\left| \mathbb{P}\Bigl(\ip{R_i}{\sigma}<x,\ip{R_{i'}}{\sigma}<x\Bigr) - \Phi\left(\frac{x}{\sqrt{S_T}}-\frac{1}{\sqrt{S_T}}\sum_{0\le t\le T}n_t\sqrt{\frac{2k_t}{\pi k^2}}\right)^2\right| =  O\left(n^{-\frac{1}{20}+\epsilon}\right)\label{eq:joint-Ri-Ri'}.
\end{align}
Similarly, we have
\begin{equation}\label{eq:dist-Ri}
    \left|\mathbb{P}\Bigl(\ip{R_i}{\sigma}<x\Bigr) - \Phi\left(\frac{x}{\sqrt{S_T}}-\frac{1}{\sqrt{S_T}}\sum_{0\le t\le T}n_t\sqrt{\frac{2k_t}{\pi k^2}}\right) \right|= O\left(n^{-\frac{1}{20}+\epsilon}\right).
    \end{equation}
Next, let
\[
\xi \triangleq \Phi\left(\frac{x}{\sqrt{S_T}}-\frac{1}{\sqrt{S_T}}\sum_{0\le t\le T}n_t\sqrt{\frac{2k_t}{\pi k^2}}\right).
\]
Moreover, set
\begin{align*}
    \Delta_1&\triangleq  \mathbb{P}\Bigl(\ip{R_i}{\sigma}<x,\ip{R_{i'}}{\sigma}<x\Bigr) - \Phi\left(\frac{x}{\sqrt{S_T}}-\frac{1}{\sqrt{S_T}}\sum_{0\le t\le T}n_t\sqrt{\frac{2k_t}{\pi k^2}}\right)^2 \\
    \Delta_2&\triangleq \mathbb{P}\Bigl(\ip{R_i}{\sigma}<x\Bigr) - \Phi\left(\frac{x}{\sqrt{S_T}}-\frac{1}{\sqrt{S_T}}\sum_{0\le t\le T}n_t\sqrt{\frac{2k_t}{\pi k^2}}\right).
\end{align*}
Then, $|\Delta_1|,|\Delta_2|=O(n^{-\frac{1}{20}+\epsilon})$. Using~\eqref{eq:joint-Ri-Ri'} and~\eqref{eq:dist-Ri}, we obtain
\begin{align*}
   & \Bigg|\mathbb{E}\Bigg[\Bigl(\ind\bigl\{\ip{R_i}{\sigma}<x\bigr\}-\xi\Bigr)\Bigl(\ind\bigl\{\ip{R_{i'}}{\sigma}<x\bigr\}-\xi\Bigr)\Bigg]\Bigg|\\
    &=\Bigg|\mathbb{P}\Bigl(\ip{R_i}{\sigma}<x,\ip{R_{i'}}{\sigma}<x\Bigr)-2\mathbb{P}\bigl(\ip{R_i}{\sigma}<x\bigr)\xi + \xi^2\Bigg|\\
    &=\Bigl|\Delta_1+\xi^2 - 2\xi\bigl(\Delta_2+\xi\bigr)+\xi^2\Bigr| \\
    &=\Bigl|\Delta_1-2\xi\Delta_2\Bigr| \\
    &=O\Bigl(n^{-\frac{1}{20}+\epsilon}\Bigr)
\end{align*}
as $\xi<1$. Hence, we arrive at
\begin{align*}
    &\mathbb{E}\left[\left(\Bigl|\{1\le i\le k:\ip{R_i}{\sigma}<x\}\Bigr| - k\Phi\left(\frac{x}{\sqrt{S_T}}-\frac{1}{\sqrt{S_T}}\sum_{0\le t\le T}n_t\sqrt{\frac{2k_t}{\pi k^2}}\right)\right)^2\right] \\& = \mathbb{E}\left[\left(\sum_{1\le i\le k}\left(\ind\bigl\{\ip{R_i}{\sigma}<x\bigr\} -\Phi\left(\frac{x}{\sqrt{S_T}}-\frac{1}{\sqrt{S_T}}\sum_{0\le t\le T}n_t\sqrt{\frac{2k_t}{\pi k^2}}\right)\right)\right)^2\right]\\
    &= O(k) + k^2 O\left(n^{-\frac{1}{20}+\epsilon}\right) =O\left(n^{\frac{39}{20}+\epsilon}\right),
\end{align*}
since $k\le n$. This concludes the proof of Lemma~\ref{lemma:main-card}.
\end{proof}

\paragraph{Index sets are nearly identical.}
Next, assume that the algorithm completed $\ell-1$ rounds and generated  $\sigma,\overline{\sigma}\in\{\pm 1\}^{\sum_{0\le t\le \ell-1}n_t}$. Recall that $\mathcal{I}_\ell$ is the set of (row) indices $1\le i\le k$ corresponding to smallest $k_\ell$ elements among $\left(\sum_{0\le t\le \ell-1}n_t\right)^{-\frac12}\ip{R_i}{\sigma}$, $1\le i\le k$. Likewise, $\overline{\mathcal{I}}_\ell$ denotes the corresponding set of indices for $\left(\sum_{0\le t\le \ell-1}n_t\right)^{-\frac12}\ip{\overline{R}_i}{\overline{\sigma}}$, $1\le i\le k$. In particular, Lemma~\ref{lemma:main-card} yields that there is an $x_\ell$ such that w.h.p.,
\[
\mathcal{I}_\ell \approx \Bigl\{1\le i\le k:\ip{R_i}{\sigma}<x_\ell\Bigr\}\qquad\text{and}\qquad \overline{\mathcal{I}}_\ell \approx \Bigl\{1\le i\le k:\ip{\overline{R_i}}{\overline{\sigma}}<x_\ell\Bigr\}.
\]
We now show $\mathcal{I}_\ell$ and $\overline{\mathcal{I}}_\ell$ are nearly identical: $|\mathcal{I}_\ell\cap \overline{\mathcal{I}}_\ell|\ge k_\ell-o(k_\ell)$. 
\begin{lemma}\label{lemma:almost-full-intesect} Recall $J_\ell$ from~\eqref{eq:J-ell} and assume that $|J_{\ell-1}|\le n^{1-\alpha_{\ell-1}}$ for $\alpha_{\ell-1}$ defined in~\eqref{eq:alpha-ell}. Then,
\[
\Bigl|\mathcal{I}_\ell\cap\overline{\mathcal{I}}_\ell\Bigr|\ge k_\ell - O\left(n^{1-\alpha_{\ell-1}/4}\right)
\]
with probability at least $1-O\left(n^{-1/40+\epsilon}\right)$,  where $\epsilon>0$ is arbitrary. 
\end{lemma}

\begin{proof}[Proof of Lemma~\ref{lemma:almost-full-intesect}]
Recall $S_{\ell-1}=\sum_{0\le s\le \ell-1}n_s$ appearing in Lemma~\ref{lemma:main-card}. We find $x_\ell\in\mathbb{R}$ satisfying
\begin{equation}\label{eq:x-ell-k-ell}
k\Phi\left(\frac{1}{\sqrt{S_{\ell-1}}}\left(x_\ell-\sum_{0\le s\le \ell-1}n_s\sqrt{\frac{2k_s}{\pi k^2}}\right)\right) = k_\ell
\end{equation}
as $\Phi(\cdot)$ is continuous. For this choice of $x_\ell$, using Lemma~\ref{lemma:main-card} with $T=\ell-1$ and applying Markov's inequality, we arrive at
\begin{equation}\label{eq:markoff}
\mathbb{P}\left(\Bigl|\bigl|\{1\le i\le k:\ip{R_i}{\sigma}<x_\ell\}\bigr| - k_\ell\Bigr|>n^{\frac{79}{80}}\right)\le O\left(n^{-\frac{1}{40}+\epsilon}\right),
\end{equation}
where $\epsilon>0$ is arbitrary. 

We now claim that as long as $\ell \le c\log_{10}\log_{10}n$ for $c>0$ small enough, 
\[
k_\ell = \omega\Bigl(n^{\frac{79}{80}}\Bigr).
\]
Recall $k_\ell$ from~\eqref{eq:k-j}. We have 
\[
k_\ell =2\left\lfloor \frac12 nf_\ell^3\right\rfloor+1 \ge nf_\ell^3-1 = n\cdot 10^{-3\cdot 2^\ell}-1\ge n\cdot 10^{-3\cdot 2^{c\log_{10}\log_{10}n}}-1.
\]
Above, we used the fact $\ell\le L\le c\log_{10}\log_{10}n$.

Rearranging, we have
\[
n\cdot 10^{-3\cdot 2^{c\log_{10}\log_{10}n}}-1 = \exp_{10}\left(\log_{10}n-3\cdot\bigl(\log_{10}n\bigr)^{c\log_{10}2}\right)-1.
\]
From here, it is evident that if $c\log_{10}2<1$, then indeed $k_\ell=\omega(n^{79/80})$. Consequently, \eqref{eq:markoff} yields
\[
k_\ell+O\left(n^{\frac{79}{80}}\right)\ge \Bigl|\{1\le i\le k:\ip{R_i}{\sigma}<x_\ell\}\Bigr|\ge k_\ell -O\left(n^{\frac{79}{80}}\right),
\]
with probability $1-O(n^{-1/40+\epsilon})$. Define next
the sets
\begin{equation}\label{eq:lg-ell}
    {\rm LG}_j \triangleq \Bigl\{1\le i\le  k:\bigl|\ip{R_i}{\sigma} - \ip{\overline{R}_i}{\overline{\sigma}}\bigr|\ge n^{1/2-\beta_j}\Bigr\},
\end{equation}
where  the inner product appearing in ${\rm LG}_j$ is taken over $\R^{\sum_{0\le s\le j}n_s}$ and
\begin{equation}\label{eq:beta-ell}
    \beta_j = \alpha_j/4 = \alpha_0 \cdot 10^{-j}/4,\quad \text{where}\quad \alpha_0=0.04.
\end{equation}
Note that under the assumption $|J_{\ell-1}|\le n^{1-\alpha_{\ell-1}}$,  Lemma~\ref{lemma:ip-and-bar-ip-close} yields
\[
\sum_{1\le  i\le k}\left|\ip{R_i}{\sigma}-\ip{\overline{R}_i}{\overline{\sigma}}\right|\le Ck\sqrt{n}\left(\tau \log_{10} n+n^{-\alpha_{\ell-1}/2}\right),
\]
(where the inner product is taken over $\R^{S_{\ell-1}}$) with probability at least $1-\exp(-k/3)$. Note also from the definition of ${\rm LG}_{\ell-1}$ that
\[
\sum_{1\le  i\le k}\left|\ip{R_i}{\sigma}-\ip{\overline{R}_i}{\overline{\sigma}}\right|\ge |{\rm LG}_{\ell-1}|n^{1/2-\beta_{\ell-1}}.
\]
Since $k=\Theta(n)$, we arrive at
\[
\mathbb{P}\Bigl(\bigl|{\rm LG}_{\ell-1}\bigr|=O\left(n^{1+\beta_{\ell-1}-\alpha_{\ell-1}/2}\right)\Bigr)\ge  1-\exp\left(-k/2\right).
\]
Next, introduce the sets
\[
\mathcal{S}(x) \triangleq \Bigl\{1\le i\le  k:\ip{R_i}{\sigma}<x\Bigr\}\qquad\text{and}\qquad \overline{\mathcal{S}}(x) \triangleq \Bigl\{1\le i\le  k:\ip{\overline{R_i}}{\overline{\sigma}}<x\Bigr\}.
\]
In particular by~\eqref{eq:markoff} w.p. at least $1-O\bigl(n^{-1/40+\epsilon}\bigr)$,
\[
k_\ell- O\left(n^{\frac{79}{80}}\right)\le \bigl|\mathcal{S}(x_\ell)\bigr|,\,\,\bigl|\overline{\mathcal{S}}(x_\ell)\bigr| \le k_\ell + O\left(n^{\frac{79}{80}}\right).
\]
We now record some useful set inclusion properties (each holding w.p. $1-O\bigl(n^{-\frac{1}{40}+\epsilon}\bigr)$, which is suppressed for convenience).
\begin{itemize}
    \item We claim 
    \[
    \bigl|\mathcal{I}_\ell\cap \mathcal{S}(x_\ell)\bigr| \ge k_\ell -O\left(n^{\frac{79}{80}}\right)\quad\text{and}\quad \bigl|\overline{\mathcal{I}}_\ell\cap \overline{\mathcal{S}}(x_\ell)\bigr| \ge k_\ell -O\left(n^{\frac{79}{80}}\right).
    \]
    To see this, let $\overline{x_\ell}=\max_{i\in\mathcal{I}_\ell}\ip{R_i}{\sigma}$. Note that if $\overline{x_\ell}\le x_\ell$, then $\mathcal{I}_\ell\subset \mathcal{S}\bigl(x_\ell\bigr)$, yielding the conclusion as $|\mathcal{I}_\ell|=k_\ell$. If $\overline{x_\ell}>x_\ell$, then $\mathcal{S}\bigl(x_\ell\bigr)\subset \mathcal{I}_\ell$, and we have the claim since $|\mathcal{S}\bigl(x_\ell\bigr)|\ge k_\ell - O\bigl(n^{\frac{79}{80}}\bigr)$. The same argument applies also to $\overline{\mathcal{S}}(x_\ell)$ and $\overline{\mathcal{I}}_\ell$.
    \item Next, observe that if  $i\in \mathcal{S}\left(x_\ell-n^{1/2-\beta_{\ell-1}}\right)\setminus {\rm LG}_{\ell-1}$ then $i\in \overline{\mathcal{S}}(x_\ell)$. Likewise, if $i\in\overline{\mathcal{S}}(x_\ell)\setminus {\rm LG}_{\ell-1}$, then $i\in \mathcal{S}\left(x_\ell+n^{1/2-\beta_{\ell-1}}\right)$. That is, except for the indices in ${\rm LG}_{\ell-1}$, we have
    \[
    \mathcal{S}\left(x_\ell-n^{1/2-\beta_{\ell-1}}\right)\subset\overline{\mathcal{S}}(x_\ell)\subset \mathcal{S}\left(x_\ell+n^{1/2-\beta_{\ell-1}}\right).
    \]
    Recall now the relation between $k_\ell$ and $x_\ell$ from~\eqref{eq:x-ell-k-ell}. 
    Using the fact  $\Phi(\cdot)$ is $1-$Lipschitz, we obtain
    \begin{align*}
      \left|k_\ell - k\Phi\left(\frac{1}{\sqrt{S_{\ell-1}}}\left(x_\ell -n^{\frac12-\beta_{\ell-1}}-\sum_{0\le s\le \ell-1}n_s\sqrt{\frac{2k_s}{\pi k^2}}\right)\right)\right|\le k\frac{n^{\frac12-\beta_{\ell-1}}}{\sqrt{S_{\ell-1}}} = O\left(n^{1-\beta_{\ell-1}}\right),
    \end{align*}
    as $k=\Theta(n)$ and $\Theta(n)=n_0\le \sum_{0\le s\le \ell-1}n_s =S_{\ell-1}\le n$ and thus $S_{\ell-1}=\Theta(n)$. Consequently,
   \begin{align*}
   \Bigl| \mathcal{S}\left(x_\ell-n^{1/2-\beta_{\ell-1}}\right)\Bigr|&\ge k_\ell - O(n^{79/80}) -O\left(n^{1-\beta_{\ell-1}}\right) \\  \Bigl| \mathcal{S}\left(x_\ell+n^{1/2-\beta_{\ell-1}}\right)\Bigr|&\le k_\ell + O(n^{79/80}) +O\left(n^{1-\beta_{\ell-1}}\right).
     \end{align*}
\end{itemize}
Finally, observe that any subset of 
\[
\mathcal{S}\left(x_\ell-n^{1/2-\beta_{\ell-1}}\right) \setminus {\rm LG}_{\ell-1}
\]
of cardinality at least
\[
k_\ell - O\left(n^{\frac{79}{80}}\right)-O\left(n^{1-\beta_{\ell-1}}\right)-|{\rm LG}_{\ell-1}|
\]
is necessarily contained in $\mathcal{I}_\ell\cap \overline{\mathcal{I}}_\ell$. Thus, we arrive at
\begin{align*}
    \bigl|\mathcal{I}_\ell\cap\overline{\mathcal{I}}_\ell\bigr| &\ge k_\ell - O\left(n^{\frac{79}{80}}\right)-O\left(n^{1-\beta_{\ell-1}}\right)-|{\rm LG}_{\ell-1}|\\
    &\ge k_\ell -O\left(n^{\frac{79}{80}}\right)-O\left(n^{1-\beta_{\ell-1}}\right)-O\left(n^{1+\beta_{\ell-1} -\alpha_{\ell-1}/2}\right) \\
    &\ge k_\ell -O\left(n^{1-\alpha_{\ell-1}/4}\right)
\end{align*}
using~\eqref{eq:beta-ell} and~\eqref{eq:alpha-ell}. Noticing that this process is valid with probability at least
\[
1-O\left(n^{-1/40+\epsilon}\right)-\exp(-k/3)
\]
with $k=\Theta(n)$, the proof of Lemma~\ref{lemma:almost-full-intesect} is complete. 
\end{proof}
\paragraph{Index Sets Being Nearly  Identical Implies Next Block Being Nearly Identical.}
Denote
\[
\sigma\bigl(k:\ell\bigr) \triangleq \left(\sigma_i:k\le  i\le \ell\right)\in\{ -1,1\}^{\ell-k+1}.
\]
The last auxiliary result we need is the following. \begin{lemma}\label{lemma:kr-n0-p-1-to-n0-p-n1}
Suppose that the algorithm run for $T-1$ rounds generating $\sigma,\overline{\sigma}\in\{-1,1\}^{\sum_{0\le t\le T-1}n_t}$. Consider the inner products  $\ip{R_i}{\sigma}$, $1\le i\le k$ (taken on $\R^{\sum_{0\le t\le T-1}n_t}$) and let $\mathcal{I}_T$ with $|\mathcal{I}_T|=k_T$ be such that $i\in\mathcal{I}_T$ iff $\ip{R_i}{\sigma}$ is among $k_T$ smallest inner products. Similarly, define the set $\overline{\mathcal{I}}_T$ with $|\overline{\mathcal{I}}_T|=k_T$ for the collection $\ip{\overline{R}_i}{\overline{\sigma}}$, $1\le i\le k$; and the random variable
\[
\mathcal{T} \triangleq \bigl|\mathcal{I}_T\cap \overline{\mathcal{I}}_T\bigr|.
\]
Then, conditional on $\mathcal{T}=D$, 
\[
d_H\Bigl(\sigma\bigl(S_{T-1}+1:S_T\bigr),\overline{\sigma}\bigl(S_{T-1}+1:S_T\bigr)\Bigr)\distr {\rm Bin}\left(n_T,\frac12-\frac1\pi \sin^{-1}\Bigl(D\cos(\tau)/k_T\Bigr) \right).
\]
\end{lemma}

\begin{proof}[Proof of Lemma~\ref{lemma:kr-n0-p-1-to-n0-p-n1}]
For convenience, drop $\tau$ appearing in  $\overline{\mathcal{M}}(\tau)$.
Observe  that the randomness in $\mathcal{I}_T$ and $\overline{\mathcal{I}}_T$ are due to $\mathcal{M}_{[k]:[s_{T-1}]}$ and $\overline{\mathcal{M}}_{[k]:[s_{T-1}]}$, respectively (that is, due to first $n_0+n_1+\cdots+n_{T-1}$ columns of corresponding matrices). Having fixed $\mathcal{I}_T$ and $\overline{\mathcal{I}}_T$, note that the next $n_T$ entries of $\sigma$ and $\overline{\sigma}$ are obtained by running the majority algorithm on the submatrices 
\[
\mathcal{M}\Bigl(\mathcal{I}_T:[S_{T-1}+1,S_T]\Bigr)\qquad\text{and}\qquad \overline{\mathcal{M}}\Bigl(\overline{\mathcal{I}}_T:[S_{T-1}+1,S_T]\Bigr),
\]
respectively. Now, condition on $|\mathcal{I}\cap \overline{\mathcal{I}}|=D$. Define variables $A,B,\overline{A}$, and $\overline{B}$ as follows:
\[
A \triangleq \sum_{1\le i\le D}X_i, \qquad \overline{A}\triangleq \sum_{1\le i\le D}\bigl(\cos(\tau)X_i+\sin(\tau)Y_i\bigr);
\]
and
\[
B \triangleq \sum_{D+1\le i\le k_T}X_i, \qquad \overline{B}\triangleq \sum_{D+1\le i\le k_T}X_i',
\]
where $X_i,X_i',Y_i$ are i.i.d. $\mathcal{N}(0,1)$. It is then clear that
\[d_H\Bigl(\sigma\bigl(S_{T-1}+1:S_T\bigr),\overline{\sigma}\bigl(S_{T-1}+1:S_T\bigr)\Bigr)\distr{\rm Bin}(n_T,p),
\]
where
\[
p\triangleq \mathbb{P}\Bigl({\rm sgn}(A+B)\ne {\rm sgn}(\overline{A}+\overline{B})\Bigr)=2\mathbb{P}\Bigl(A+B\ge 0, \overline{A}+\overline{B}\le 0\Bigr)
\]
by symmetry. Observe that
\[
\mathbb{E}[A+B]=\mathbb{E}[\overline{A}+\overline{B}] = 0 \qquad\text{and}\qquad \mathbb{E}[(A+B)(\overline{A}+\overline{B})] = \mathbb{E}[A\overline{A}] = D\cos(\tau). 
\]
From here, applying Lemma~\ref{lemma:gaussian-orthant} to bivariate standard normal variables $k_T^{-\frac12}(A+B)$ and $-k_T^{-\frac12}(\overline{A}+\overline{B})$, we conclude
\begin{align*}
    p &= 2\left(\frac14 + \frac{1}{2\pi}\sin^{-1}(-D\cos(\tau)/k_T)\right)= \frac12 - \frac1\pi \sin^{-1}(D\cos(\tau)/k_T),
\end{align*}
where we used the fact $\sin^{-1}(\cdot)$ is an odd function, completing the proof of Lemma~\ref{lemma:kr-n0-p-1-to-n0-p-n1}.
\end{proof}
Equipped with all necessary auxiliary  tools, we now complete the proof of Proposition~\ref{prop:kr-stable}.
\begin{proof}[Proof of Proposition~\ref{prop:kr-stable}]
Let $T\le c\log_{10}\log_{10} n$ for some $c>0$ small  enough, recall $\alpha_\ell$ from~\eqref{eq:alpha-ell} and $J_\ell$ from~\eqref{eq:J-ell}. Note that applying Lemma~\ref{lemma:gauss-cont-maj} we immediately obtain
\[
|J_0|=d_H\Bigl(\sigma(1:n_0),\overline{\sigma}(1:n_0)\Bigr)\distr {\rm Bin}\left(n_0,\frac{\tau}{\pi}\right).
\]
In particular, applying a Chernoff bound, and recalling $n_0=\Theta(n)$,
\[
|J_0| = O(n\tau) = O\left(n^{1-2\alpha_0}\right)<n^{1-\alpha_0}
\]
with probability at least  $1-\exp\bigl(-\Omega\left(n^{1-2\alpha_0}\right)\bigr)$.

We now proceed by inducting on $\ell$, where the base case, $\ell=0$, has been verified above. Assume now that 
\begin{equation}\label{eq:induction-step}
\mathbb{P}\Bigl(|J_{\ell-1}|\le n^{1-\alpha_{\ell-1}}\Bigr)\ge 1-p_{\ell-1}.
\end{equation}
Using Lemma~\ref{lemma:almost-full-intesect}, we obtain
\begin{equation}\label{eq:almost-fully-intersect}
\mathbb{P}\Bigl(\bigl|\mathcal{I}_\ell\cap\overline{\mathcal{I}}_\ell\bigr|\ge k_\ell -O\left(n^{1-\alpha_{\ell-1}/4}\right)\Bigr) \ge 1-O\left(n^{-1/40+\epsilon}\right).
\end{equation}
Now, conditional on  $|\mathcal{I}_\ell\cap\overline{\mathcal{I}}_\ell|=D$, Lemma~\ref{lemma:kr-n0-p-1-to-n0-p-n1}  implies that
\begin{align}
|J_\ell|-|J_{\ell-1}|&=d_H\Bigl(\sigma(S_{\ell-1}+1:S_\ell),\overline{\sigma}(S_{\ell-1}+1:S_\ell)\Bigr)\nonumber\\
&\distr {\rm Bin}\Bigl(n_\ell,\frac12-\frac1\pi \sin^{-1}\Bigl(D\cos(\tau)/k_\ell\Bigr)\Bigr)\label{eq:this-is-binomial},
\end{align}
where $S_T\triangleq\sum_{0\le j\le T}n_j$ for any $T\in\mathbb{N}$. 

Now, recall that $k_\ell \ge  nf_\ell^3-1$ for $f_\ell=10^{-2^\ell}$ (see~\eqref{eq:f-j} and~\eqref{eq:k-j}). Hence,
\[
\frac{k_\ell-O\left(n^{1-\alpha_{\ell-1}/4}\right)}{k_\ell} = 1-O\left(\frac{n^{1-\alpha_{\ell-1}/4}}{k_\ell}\right)\ge 1-O\left(\frac{n^{-\alpha_{\ell-1}/4}}{10^{-3\cdot 2^\ell}}\right).
\]
We now claim 
\[
1-O\left(\frac{n^{-\alpha_{\ell-1}/4}}{10^{-3\cdot 2^\ell}}\right)\ge 1-O\left(n^{-\alpha_{\ell-1}/4.5}\right)
\]
for all large enough $n$, provided $c>0$ is small enough. Ignoring the absolute constants, it suffices to verify 
\[
n^{-\frac{\alpha_{\ell-1}}{4.5}}\ge 10^{3\cdot 2^\ell} \cdot n^{-\frac{\alpha_{\ell-1}}{4}} \iff n^{\frac{\alpha_{\ell-1}}{36}}\ge 10^{3\cdot 2^\ell}\iff \frac{1}{36}\log_{10}n\cdot \alpha_{\ell-1}\ge 3\cdot 2^\ell.
\]
Recall that $\alpha_{\ell-1}=0.4\cdot 10^{-\ell}$. Thus, it suffices to verify 
\[
\frac{1}{270}\log_{10}n\ge 20^\ell.
\]Recalling $\ell\le c\log_{10}\log_{10}n$ for $c>0$ small enough, we have
\[
20^{\ell}\le 20^{c\log_{10}\log_{10}n} = (\log_{10}n)^{c''},\qquad\text{where}\qquad c'' = c\log_{10}20.
\]
Finally, provided $c''<1$, we indeed have 
\[
\frac{1}{270}\log_{10}n\ge \log_{10}^{c''}n,
\]
thus the claim. 

We next employ this claim and the inequality, $\cos(\tau)\ge 1-\tau^2/2$, which is valid for all $\tau$. Using~\eqref{eq:almost-fully-intersect} it follows that there is an event of probability at least $1-O\left(n^{-1/40+\epsilon}\right)$ such that on this event, $|J_\ell|-|J_{\ell-1}|$ is stochastically dominated by the binomial random variable
\begin{equation}\label{eq:binom-auxil1}
{\rm Bin}\left(n_\ell,\frac12-\frac1\pi \sin^{-1}\Bigl(\bigl(1-\Theta\bigl(n^{-\alpha_{\ell-1}/4.5}\bigr)\bigr)\bigl(1-\tau^2/2\bigr)\Bigr)\right).
\end{equation}
Using Taylor series $\sin^{-1}(1-x) = \frac{\pi}{2} -\sqrt{2x}+o(\sqrt{x})$, we obtain that  the  binomial variable appearing in~\eqref{eq:binom-auxil1} is stochastically dominated further by a binomial random variable ${\rm Bin}(n_\ell,q_\ell)$ where $q_\ell=\Theta(n^{-\alpha_{\ell-1}/9})$. Since $n_\ell \le n$ per~\eqref{eq:n-j}, we also have by Chernoff bound
\begin{equation}\label{eq:binom-tail-auxi}
\mathbb{P}\Bigl({\rm Bin}\bigl(n_\ell,q_\ell\bigr)=O\left(n^{1-\alpha_{\ell-1}/9}\right)\Bigr)\ge 1-\exp\Bigl(-\Omega\Bigl(n^{1-\alpha_{\ell-1}/9}\Bigr)\Bigr).
\end{equation}
Combining~\eqref{eq:induction-step},~\eqref{eq:this-is-binomial} and~\eqref{eq:binom-tail-auxi} via a union bound, we conclude that
\begin{equation}\label{eq:inductive-step2}
    \mathbb{P}\Bigl(|J_\ell|\le n^{1-\alpha_{\ell-1}}+O\left(n^{1-\alpha_{\ell-1}/9}\right)\Bigr)\ge 1-p_{\ell-1}-O\left(n^{-1/40+\epsilon}\right)-\exp\Bigl(-\Omega\Bigl(n^{1-\alpha_{\ell-1}/9}\Bigr)\Bigr).
\end{equation}
We set
\begin{equation}\label{eq:p-ell}
p_\ell \triangleq p_{\ell-1}+O\left(n^{-1/12+\epsilon}\right)+\exp\Bigl(-\Omega\Bigl(n^{1-\alpha_{\ell-1}/9}\Bigr)\Bigr) =p_{\ell-1}+O\left(n^{-1/40+\epsilon}\right).
\end{equation}
We now ensure $|J_\ell|\le n^{1-\alpha_\ell}$ w.h.p., where $\alpha_\ell = \alpha_{\ell-1}/10$. To that end, we first claim
\begin{equation}\label{eq:almost-there}
n^{1-\alpha_{\ell-1}}+O\left(n^{1-\alpha_{\ell-1}/9}\right)\le n^{1-\alpha_{\ell-1}}+n^{1-\alpha_{\ell-1}/9.5}.
\end{equation}
To prove this, it suffices to verify $n^{\frac{\alpha_{\ell-1}}{9}-\frac{\alpha_{\ell-1}}{9.5}} = n^{\frac{\alpha_{\ell-1}}{171}}=\omega(1)$. Using the fact $\alpha_{\ell-1} =0.1\cdot 10^{-\ell}$ per~\eqref{eq:alpha-ell} and $\ell \le L\le c\log_{10}\log_{10}n$ in the setting of Proposition~\ref{prop:kr-stable}, we have
\begin{align*}
    n^{\frac{\alpha_{\ell-1}}{171}}&=\exp_{10}\left(\frac{\alpha_{\ell-1}}{171}\log_{10}n\right) \ge \exp_{10}\left(\frac{10^{-c\log_{10}\log_{10}n}}{1710}\log_{10}n\right) = \exp_{10}\left(\frac{1}{1710}\bigl(\log_{10}n\bigr)^{1-c}\right),
\end{align*}
which is indeed $\omega(1)$ if $c<1$. This yields~\eqref{eq:almost-there}. Hence, it suffices to show $x+x^{1/9.5}\le x^{1/10}$ for $x=n^{-\alpha_{\ell-1}}$. Now, assume $\ell \le L\le c\log_{10}\log_{10} n$ (where $L$ is the number of steps analyzed) for $c>0$ small enough. Then
\begin{align*}
x&=\exp_{10}\Bigl(-\alpha_{\ell-1}\log_{10} n\Bigr)\\
&=\exp_{10}\Bigl(-0.04\cdot 10^{-(\ell-1)}\cdot \log_{10} n\Bigr)\\
&\le \exp_{10}\Bigl(-0.4\cdot 10^{-c\log_{10}\log_{10}n}\cdot \log_{10}n\Bigr) \\ 
&=\exp_{10}\Bigl(-0.4\cdot(\log_{10} n)^{1-c}\Bigr).
\end{align*}
We have $t+t^{1/9.5}<t^{1/10}$ for $t$ sufficiently small (e.g. $0\le t<0.01$ suffices).  Hence, provided $c<1$, it is the case that for $n$ sufficiently large, $x+x^{1/9.5}\le x^{1/10}$ for $x=n^{-\alpha_{\ell-1}}$ and for any $\ell \le c\log_{10}\log_{10} n$. Thus,
\begin{equation}\label{eq:induction-complete}
    \mathbb{P}\Bigl(|J_\ell|\le n^{1-\alpha_\ell}\Bigr)\ge 1-p_\ell \qquad\text{where}\qquad p_\ell = p_{\ell-1} +O\left(n^{-1/40+\epsilon}\right).
\end{equation}
Since the inductive step from  $\ell-1\to \ell$ (more concretely from~\eqref{eq:induction-step} to~\eqref{eq:induction-complete}), increases the probability by $O\left(n^{-1/40+\epsilon}\right)$ and the whole process runs $\log_{10}\log_{10} n$ rounds, we complete the proof of Proposition~\ref{prop:kr-stable}. \qedhere
\end{proof}
\subsubsection{Proof of Theorem~\ref{thm:kr-stable}}\label{sec:sub-sub-kr-stable}
Having established Proposition~\ref{prop:kr-stable}, we now finish the proof of Theorem~\ref{thm:kr-stable}. For simplicity, we omit floor/ceiling operators whenever convenient.
\begin{proof}[Proof of  Theorem~\ref{thm:kr-stable}]
Let $L=c\log_{10}\log_{10} n$ be as in Proposition~\ref{prop:kr-stable} for $c>0$ small enough. We now show that it suffices  to analyze $L$ rounds as opposed to the full implementation of $ C\log_{10} \log_{10} n$ rounds, where $C>c>0$. In particular, we claim
\begin{equation}\label{eq:remaining-coord-o(n)}
\sum_{ c\log_{10} \log_{10} n \le j\le  C\log_{10}\log_{10} n} n_j = o(n), 
\end{equation}
for any constant $c>0$. For $N\triangleq C\log_{10}\log_{10} n$, the number of rounds, and $0\le j\le N$; recall $f_j$ from~\eqref{eq:f-j}, $n_j$ from~\eqref{eq:n-j}, and $k_j$ from~\eqref{eq:k-j}. 
Applying now a telescoping  argument,
\begin{align*}
    \sum_{c\log_{10}\log_{10} n\le j\le C\log_{10}\log_{10} n}n_j & = \sum_{c\log_{10}\log_{10} n\le j\le C\log_{10}\log_{10} n}\left( \left \lfloor \frac{n}{A} \sum_{0\le i\le j}f_i\right\rfloor  - \left \lfloor\frac{n}{A}\sum_{0\le i\le j-1} f_i \right\rfloor\right)\\
    &=\left\lfloor \frac{n}{A}\underbrace{\sum_{0\le i\le C\log_{10}\log_{10} n}f_i}_{=A}\right\rfloor  -\left\lfloor \frac{n}{A}\sum_{0\le i\le c\log_{10}\log_{10} n-1}f_i\right\rfloor \\
    &\le n\left(1-\frac1A\sum_{0\le i\le c\log_{10}\log_{10} n}f_i\right)+1 \\
    &\le \frac{n}{A}\sum_{c\log_{10}\log_{10} n +1\le i\le C\log_{10}\log_{10} n}f_i+1 \\
    &= O\left(n\log_{10}\log_{10} n\cdot 10^{-2^{c\log_{10} \log_{10} n}}\right) \\
    & = O\left(n\cdot \log_{10}\log_{10} n\cdot 10^{-\left(\log_{10} n\right)^{c'}}\right)
\end{align*}
for $c'=c\log_{10} 2<1$. We now verify
\[
n\cdot \log_{10}\log_{10}n\cdot 10^{-(\log_{10}n)^{c'}} = o(n)\iff \log_{10}\log_{10}n\cdot 10^{-(\log_{10}n)^{c'}}=o(1).
\]
Indeed,
\begin{align*}
    \log_{10}\log_{10}n\cdot 10^{-(\log_{10}n)^{c'}} & = \exp_{10}\left(\log_{10}\log_{10}\log_{10}n - \bigl(\log_{10} n\bigr)^{c'}\right) = o(1)
\end{align*}
for any $c>0$. This yields~\eqref{eq:remaining-coord-o(n)}.

Finally, combining~\eqref{eq:remaining-coord-o(n)} with Proposition~\ref{prop:kr-stable}, we  complete the proof of  Theorem~\ref{thm:kr-stable}.
\end{proof}

\subsection{Proof of Theorem~\ref{thm:ogp-universality}}\label{pf:ogp-universality}
\begin{proof}[Proof of Theorem~\ref{thm:ogp-universality}]
The proof is quite similar to that of Theorem~\ref{thm:m-OGP-small-kappa}, hence we only point out the necessary modification. 
Note that the probability term that one considers (cf. Lemma~\ref{lemma:m-ogp-prob}) is
\[
\mathbb{P}\Bigl[\bigl|\ip{\sigma^{(i)}}{X}\bigr|\le \kappa\sqrt{n},1\le i\le m\Bigr]^{\alpha n},
\]
where $X=(X_1,\dots,X_n)$ has i.i.d.\,entries with $X_i\sim \mathcal{D}$. To apply Theorem~\ref{thm:berry-esseen}, set
\begin{equation}\label{eq:y-i-berry}
    Y_j\triangleq \begin{pmatrix}\frac{1}{\sqrt{n}}\sigma_1(j)X_j\\ \frac{1}{\sqrt{n}}\sigma_2(j)X_j\\\vdots \\ \frac{1}{\sqrt{n}}\sigma_m(j)X_j\end{pmatrix}\in\R^m.
\end{equation}
Indeed $Y_j\in\R^m$, $1\le j\le n$, is a collection of independent centered random vectors, and 
\[
\Bigl\{\bigl|\ip{\sigma^{(i)}}{X}\bigr|\le \kappa \sqrt{n},1\le i\le m\Bigr\}= \bigl\{S\in U\bigr\},
\]
for $S=\sum_{j\le n}Y_j$ and $U=[-\kappa,\kappa]^n$. Furthermore, $\Sigma \triangleq {\rm Cov}(S)\in\R^{m\times m}$ is such that (a) $\Sigma_{ii}=1$ for $1\le i\le m$; and (b) $\Sigma_{ij}=\Sigma_{ji}=n^{-1}\ip{\sigma^{(i)}}{\sigma^{(j)}}$ for $1\le i<j\le m$. Next, observe that since $\Sigma\in\R^{m\times m}$ with $m=O_n(1)$, it follows that
\begin{align*}
\mathbb{E}\Bigl[\bigl\|\Sigma^{-\frac12}Y_j\bigr\|_2^3\Bigr]&\le\bigl\|\Sigma^{-\frac12}\bigr\|^{3} \mathbb{E}\Bigl[\bigl\|Y_j\bigr\|_2^3\Bigr]\\
&=\bigl\|\Sigma^{-\frac12}\bigr\|^{3}\mathbb{E}\left[\left(\frac{m}{n}X_j^2\right)^{3/2}\right]\\
&=O\bigl(n^{-\frac32}\bigr).
\end{align*}
Applying now Theorem~\ref{thm:berry-esseen}, we obtain
\[
\mathbb{P}\bigl[S\in U\bigr]\le \mathbb{P}\bigl[Z\in U\bigr]+O\bigl(n^{-\frac12}\bigr).
\]
Consequently, for $Z\sim \mathcal{N}(0,\Sigma)$,
\begin{align*}
\mathbb{P}\bigl[S\in U\bigr]^{\alpha n}&\le \mathbb{P}\bigl[Z\in U\bigr]^{\alpha n}\Bigl(1+O(n^{-\frac12})\Bigr)^{\alpha n}\\
&=\mathbb{P}\bigl[Z\in U\bigr]^{\alpha n}\exp\Bigl(\alpha n\ln\Bigl(1+O\bigl(n^{-1/2}\bigr)\Bigr)\Bigr) \\
&=\mathbb{P}\bigl[Z\in U\bigr]^{\alpha n}\exp\Bigl(\Theta\bigl(\sqrt{n}\bigr)\Bigr),
\end{align*}
where in the last step we used the Taylor expansion, $\ln(1-x)=-x+o(x)$ as $x\to 0$. Modifying Lemma~\ref{lemma:m-ogp-prob} by taking the extra $e^{\Theta(\sqrt{n})}$ factor into account and then applying the \emph{first moment method}, we establish Theorem~\ref{thm:ogp-universality}.
\end{proof}
\subsubsection*{Acknowledgments}
Part of this work was done while the first two authors were visiting the Simons Institute for the Theory of Computing at University of California, Berkeley in Fall 2021 as part of the semester program \emph{Computational Complexity of Statistical Inference}. The first author is supported in part by NSF grant DMS-2015517. The second author would like to thank Paxton Turner for illuminating discussions on discrepancy minimization, and bringing the references in Section~\ref{sec:alg-conj} into his attention. The third author is supported in part by NSF grant DMS-1847451.
\bibliographystyle{amsalpha}
\bibliography{bibliography}

\newcommand{\etalchar}[1]{$^{#1}$}
\providecommand{\bysame}{\leavevmode\hbox to3em{\hrulefill}\thinspace}
\providecommand{\MR}{\relax\ifhmode\unskip\space\fi MR }
\providecommand{\MRhref}[2]{%
  \href{http://www.ams.org/mathscinet-getitem?mr=#1}{#2}
}
\providecommand{\href}[2]{#2}
\begin{thebibliography}{BDVLZ20}

\bibitem[ACO08]{achlioptas2008algorithmic}
Dimitris Achlioptas and Amin Coja-Oghlan, \emph{Algorithmic barriers from phase
  transitions}, 2008 49th Annual IEEE Symposium on Foundations of Computer
  Science, IEEE, 2008, pp.~793--802.

\bibitem[AJ21]{arous2021shattering}
G{\'e}rard~Ben Arous and Aukosh Jagannath, \emph{Shattering versus
  metastability in spin glasses}, arXiv preprint arXiv:2104.08299 (2021).

\bibitem[ALS21a]{abbe2021binary}
Emmanuel Abbe, Shuangping Li, and Allan Sly, \emph{Binary perceptron: efficient
  algorithms can find solutions in a rare well-connected cluster}, arXiv
  preprint arXiv:2111.03084 (2021).

\bibitem[ALS21b]{abbe2021proof}
\bysame, \emph{Proof of the contiguity conjecture and lognormal limit for the
  symmetric perceptron}, arXiv preprint arXiv:2102.13069 (2021).

\bibitem[ALS21c]{alweiss2021discrepancy}
Ryan Alweiss, Yang~P Liu, and Mehtaab Sawhney, \emph{Discrepancy minimization
  via a self-balancing walk}, Proceedings of the 53rd Annual ACM SIGACT
  Symposium on Theory of Computing, 2021, pp.~14--20.

\bibitem[ANW21]{altschuler2021discrepancy}
Dylan~J Altschuler and Jonathan Niles-Weed, \emph{The discrepancy of random
  rectangular matrices}, Random Structures \& Algorithms (2021).

\bibitem[APZ19]{aubin2019storage}
Benjamin Aubin, Will Perkins, and Lenka Zdeborov{\'a}, \emph{Storage capacity
  in symmetric binary perceptrons}, Journal of Physics A: Mathematical and
  Theoretical \textbf{52} (2019), no.~29, 294003.

\bibitem[ART06]{achlioptas2006solution}
Dimitris Achlioptas and Federico Ricci-Tersenghi, \emph{On the solution-space
  geometry of random constraint satisfaction problems}, Proceedings of the
  thirty-eighth annual ACM symposium on Theory of computing, 2006,
  pp.~130--139.

\bibitem[AS20]{alaoui2020algorithmic}
Ahmed~El Alaoui and Mark Sellke, \emph{Algorithmic pure states for the negative
  spherical perceptron}, arXiv preprint arXiv:2010.15811 (2020).

\bibitem[Bal09]{baldassi2009generalization}
Carlo Baldassi, \emph{Generalization learning in a perceptron with binary
  synapses}, Journal of Statistical Physics \textbf{136} (2009), no.~5,
  902--916.

\bibitem[Ban10]{bansal2010constructive}
Nikhil Bansal, \emph{Constructive algorithms for discrepancy minimization},
  2010 IEEE 51st Annual Symposium on Foundations of Computer Science, IEEE,
  2010, pp.~3--10.

\bibitem[BB15]{baldassi2015max}
Carlo Baldassi and Alfredo Braunstein, \emph{A max-sum algorithm for training
  discrete neural networks}, Journal of Statistical Mechanics: Theory and
  Experiment \textbf{2015} (2015), no.~8, P08008.

\bibitem[BB19]{brennan2019optimal}
Matthew Brennan and Guy Bresler, \emph{Optimal average-case reductions to
  sparse pca: From weak assumptions to strong hardness}, arXiv preprint
  arXiv:1902.07380 (2019).

\bibitem[BBBZ07]{baldassi2007efficient}
Carlo Baldassi, Alfredo Braunstein, Nicolas Brunel, and Riccardo Zecchina,
  \emph{Efficient supervised learning in networks with binary synapses},
  Proceedings of the National Academy of Sciences \textbf{104} (2007), no.~26,
  11079--11084.

\bibitem[BBH18]{brennan2018reducibility}
Matthew Brennan, Guy Bresler, and Wasim Huleihel, \emph{Reducibility and
  computational lower bounds for problems with planted sparse structure}, arXiv
  preprint arXiv:1806.07508 (2018).

\bibitem[BDVLZ20]{baldassi2020clustering}
Carlo Baldassi, Riccardo Della~Vecchia, Carlo Lucibello, and Riccardo Zecchina,
  \emph{Clustering of solutions in the symmetric binary perceptron}, Journal of
  Statistical Mechanics: Theory and Experiment \textbf{2020} (2020), no.~7,
  073303.

\bibitem[BGT10]{bayati2010combinatorial}
Mohsen Bayati, David Gamarnik, and Prasad Tetali, \emph{Combinatorial approach
  to the interpolation method and scaling limits in sparse random graphs},
  Proceedings of the forty-second ACM symposium on Theory of computing, 2010,
  pp.~105--114.

\bibitem[BH21]{bresler2021algorithmic}
Guy Bresler and Brice Huang, \emph{The algorithmic phase transition of random $
  k $-sat for low degree polynomials}, arXiv preprint arXiv:2106.02129 (2021).

\bibitem[BHK{\etalchar{+}}19]{barak2019nearly}
Boaz Barak, Samuel Hopkins, Jonathan Kelner, Pravesh~K Kothari, Ankur Moitra,
  and Aaron Potechin, \emph{A nearly tight sum-of-squares lower bound for the
  planted clique problem}, SIAM Journal on Computing \textbf{48} (2019), no.~2,
  687--735.

\bibitem[BIL{\etalchar{+}}15]{baldassi2015subdominant}
Carlo Baldassi, Alessandro Ingrosso, Carlo Lucibello, Luca Saglietti, and
  Riccardo Zecchina, \emph{Subdominant dense clusters allow for simple learning
  and high computational performance in neural networks with discrete
  synapses}, Physical review letters \textbf{115} (2015), no.~12, 128101.

\bibitem[BJM{\etalchar{+}}21]{bansal2021online}
Nikhil Bansal, Haotian Jiang, Raghu Meka, Sahil Singla, and Makrand Sinha,
  \emph{Online discrepancy minimization for stochastic arrivals}, Proceedings
  of the 2021 ACM-SIAM Symposium on Discrete Algorithms (SODA), SIAM, 2021,
  pp.~2842--2861.

\bibitem[BJSS20]{bansal2020online}
Nikhil Bansal, Haotian Jiang, Sahil Singla, and Makrand Sinha, \emph{Online
  vector balancing and geometric discrepancy}, Proceedings of the 52nd Annual
  ACM SIGACT Symposium on Theory of Computing, 2020, pp.~1139--1152.

\bibitem[BNSX21]{bolthausen2021gardner}
Erwin Bolthausen, Shuta Nakajima, Nike Sun, and Changji Xu, \emph{Gardner
  formula for {I}sing perceptron models at small densities}, arXiv preprint
  arXiv:2111.02855 (2021).

\bibitem[BPW18]{bandeira2018notes}
Afonso~S Bandeira, Amelia Perry, and Alexander~S Wein, \emph{Notes on
  computational-to-statistical gaps: predictions using statistical physics},
  arXiv preprint arXiv:1803.11132 (2018).

\bibitem[BR13]{berthet2013computational}
Quentin Berthet and Philippe Rigollet, \emph{Computational lower bounds for
  sparse {PCA}}, arXiv preprint arXiv:1304.0828 (2013).

\bibitem[BS20]{bansalspenceronline}
Nikhil Bansal and {Joel H.} Spencer, \emph{On-line balancing of random inputs},
  Random Structures and Algorithms \textbf{57} (2020), no.~4, 879--891 (English
  (US)).

\bibitem[BZ06]{braunstein2006learning}
Alfredo Braunstein and Riccardo Zecchina, \emph{Learning by message passing in
  networks of discrete synapses}, Physical review letters \textbf{96} (2006),
  no.~3, 030201.

\bibitem[CFS15]{conlon2015recent}
David Conlon, Jacob Fox, and Benny Sudakov, \emph{Recent developments in graph
  ramsey theory.}, Surveys in combinatorics \textbf{424} (2015), no.~2015,
  49--118.

\bibitem[CGPR19]{chen2019suboptimality}
Wei-Kuo Chen, David Gamarnik, Dmitry Panchenko, and Mustazee Rahman,
  \emph{Suboptimality of local algorithms for a class of max-cut problems}, The
  Annals of Probability \textbf{47} (2019), no.~3, 1587--1618.

\bibitem[COE15]{coja2015independent}
Amin Coja-Oghlan and Charilaos Efthymiou, \emph{On independent sets in random
  graphs}, Random Structures \& Algorithms \textbf{47} (2015), no.~3, 436--486.

\bibitem[COHH17]{coja2017walksat}
Amin Coja-Oghlan, Amir Haqshenas, and Samuel Hetterich, \emph{Walksat stalls
  well below satisfiability}, SIAM Journal on Discrete Mathematics \textbf{31}
  (2017), no.~2, 1160--1173.

\bibitem[Cov65]{cover1965geometrical}
Thomas~M Cover, \emph{Geometrical and statistical properties of systems of
  linear inequalities with applications in pattern recognition}, IEEE
  transactions on electronic computers (1965), no.~3, 326--334.

\bibitem[CV14]{chandrasekaran2014integer}
Karthekeyan Chandrasekaran and Santosh~S Vempala, \emph{Integer feasibility of
  random polytopes: random integer programs}, Proceedings of the 5th conference
  on Innovations in theoretical computer science, 2014, pp.~449--458.

\bibitem[DAM17]{deshpande2017asymptotic}
Yash Deshpande, Emmanuel Abbe, and Andrea Montanari, \emph{Asymptotic mutual
  information for the balanced binary stochastic block model}, Information and
  Inference: A Journal of the IMA \textbf{6} (2017), no.~2, 125--170.

\bibitem[DFJ02]{dyer2002counting}
Martin Dyer, Alan Frieze, and Mark Jerrum, \emph{On counting independent sets
  in sparse graphs}, SIAM Journal on Computing \textbf{31} (2002), no.~5,
  1527--1541.

\bibitem[DKS17]{diakonikolas2017statistical}
Ilias Diakonikolas, Daniel~M Kane, and Alistair Stewart, \emph{Statistical
  query lower bounds for robust estimation of high-dimensional {G}aussians and
  {G}aussian mixtures}, 2017 IEEE 58th Annual Symposium on Foundations of
  Computer Science (FOCS), IEEE, 2017, pp.~73--84.

\bibitem[DM14]{deshpande2014information}
Yash Deshpande and Andrea Montanari, \emph{Information-theoretically optimal
  sparse {PCA}}, 2014 IEEE International Symposium on Information Theory, IEEE,
  2014, pp.~2197--2201.

\bibitem[DM15]{deshpande2015improved}
\bysame, \emph{Improved sum-of-squares lower bounds for hidden clique and
  hidden submatrix problems}, Conference on Learning Theory, 2015,
  pp.~523--562.

\bibitem[DS19]{ding2019capacity}
Jian Ding and Nike Sun, \emph{Capacity lower bound for the {I}sing perceptron},
  Proceedings of the 51st Annual ACM SIGACT Symposium on Theory of Computing,
  2019, pp.~816--827.

\bibitem[ES35]{erdos1935combinatorial}
Paul Erd{\"o}s and George Szekeres, \emph{A combinatorial problem in geometry},
  Compositio mathematica \textbf{2} (1935), 463--470.

\bibitem[ES18]{eldan2018efficient}
Ronen Eldan and Mohit Singh, \emph{Efficient algorithms for discrepancy
  minimization in convex sets}, Random Structures \& Algorithms \textbf{53}
  (2018), no.~2, 289--307.

\bibitem[FGR{\etalchar{+}}17]{feldman2017statistical}
Vitaly Feldman, Elena Grigorescu, Lev Reyzin, Santosh~S Vempala, and Ying Xiao,
  \emph{Statistical algorithms and a lower bound for detecting planted
  cliques}, Journal of the ACM (JACM) \textbf{64} (2017), no.~2, 1--37.

\bibitem[F{\L}92]{frieze1992independence}
Alan~M Frieze and T~{\L}uczak, \emph{On the independence and chromatic numbers
  of random regular graphs}, Journal of Combinatorial Theory, Series B
  \textbf{54} (1992), no.~1, 123--132.

\bibitem[FPV18]{feldman2018complexity}
Vitaly Feldman, Will Perkins, and Santosh Vempala, \emph{On the complexity of
  random satisfiability problems with planted solutions}, SIAM Journal on
  Computing \textbf{47} (2018), no.~4, 1294--1338.

\bibitem[Fri90]{frieze1990independence}
Alan~M Frieze, \emph{On the independence number of random graphs}, Discrete
  Mathematics \textbf{81} (1990), no.~2, 171--175.

\bibitem[Fri99]{friedgut1999sharp}
Ehud Friedgut, \emph{Sharp thresholds of graph properties, and the k-{SAT}
  problem}, Journal of the American mathematical Society \textbf{12} (1999),
  no.~4, 1017--1054.

\bibitem[Gam21]{gamarnik2021overlap}
David Gamarnik, \emph{The overlap gap property: A topological barrier to
  optimizing over random structures}, Proceedings of the National Academy of
  Sciences \textbf{118} (2021), no.~41.

\bibitem[Gar87]{gardner1987maximum}
Elizabeth Gardner, \emph{Maximum storage capacity in neural networks}, EPL
  (Europhysics Letters) \textbf{4} (1987), no.~4, 481.

\bibitem[Gar88]{gardner1988space}
\bysame, \emph{The space of interactions in neural network models}, Journal of
  physics A: Mathematical and general \textbf{21} (1988), no.~1, 257.

\bibitem[GD88]{gardner1988optimal}
Elizabeth Gardner and Bernard Derrida, \emph{Optimal storage properties of
  neural network models}, Journal of Physics A: Mathematical and general
  \textbf{21} (1988), no.~1, 271.

\bibitem[GJ21]{gamarnikjagannath2021overlap}
David Gamarnik and Aukosh Jagannath, \emph{The overlap gap property and
  approximate message passing algorithms for $ p $-spin models}, The Annals of
  Probability \textbf{49} (2021), no.~1, 180--205.

\bibitem[GJW20]{gamarnik2020low}
David Gamarnik, Aukosh Jagannath, and Alexander~S Wein, \emph{Low-degree
  hardness of random optimization problems}, 2020 IEEE 61st Annual Symposium on
  Foundations of Computer Science (FOCS), IEEE, 2020, pp.~131--140.

\bibitem[GJW21]{gamarnik2021circuit}
\bysame, \emph{Circuit lower bounds for the p-spin optimization problem}, arXiv
  preprint arXiv:2109.01342 (2021).

\bibitem[GK21a]{gamarnik2021algorithmic}
David Gamarnik and Eren~C K{\i}z{\i}lda{\u{g}}, \emph{Algorithmic obstructions
  in the random number partitioning problem}, arXiv preprint arXiv:2103.01369
  (2021).

\bibitem[GK21b]{GK-SK-AAP}
David Gamarnik and Eren~C. Kızıldağ, \emph{{Computing the partition function
  of the Sherrington–Kirkpatrick model is hard on average}}, The Annals of
  Applied Probability \textbf{31} (2021), no.~3, 1474 -- 1504.

\bibitem[GL18]{gamarnik2018finding}
David Gamarnik and Quan Li, \emph{Finding a large submatrix of a {G}aussian
  random matrix}, The Annals of Statistics \textbf{46} (2018), no.~6A,
  2511--2561.

\bibitem[GS14]{gamarnik2014limits}
David Gamarnik and Madhu Sudan, \emph{Limits of local algorithms over sparse
  random graphs}, Proceedings of the 5th conference on Innovations in
  theoretical computer science, 2014, pp.~369--376.

\bibitem[GS17a]{gamarnik2017}
\bysame, \emph{Limits of local algorithms over sparse random graphs}, Ann.
  Probab. \textbf{45} (2017), no.~4, 2353--2376.

\bibitem[GS17b]{gamarnik2017performance}
\bysame, \emph{Performance of sequential local algorithms for the random
  {NAE}-{K}-{SAT} problem}, SIAM Journal on Computing \textbf{46} (2017),
  no.~2, 590--619.

\bibitem[HKP{\etalchar{+}}17]{hopkins2017power}
Samuel~B Hopkins, Pravesh~K Kothari, Aaron Potechin, Prasad Raghavendra, Tselil
  Schramm, and David Steurer, \emph{The power of sum-of-squares for detecting
  hidden structures}, 2017 IEEE 58th Annual Symposium on Foundations of
  Computer Science (FOCS), IEEE, 2017, pp.~720--731.

\bibitem[Hop18]{hopkins2018statistical}
Samuel Brink~Klevit Hopkins, \emph{Statistical inference and the sum of squares
  method}.

\bibitem[HS21]{huang2021tight}
Brice Huang and Mark Sellke, \emph{Tight lipschitz hardness for optimizing mean
  field spin glasses}, arXiv preprint arXiv:2110.07847 (2021).

\bibitem[HSS15]{hopkins2015tensor}
Samuel~B Hopkins, Jonathan Shi, and David Steurer, \emph{Tensor principal
  component analysis via sum-of-square proofs}, Conference on Learning Theory,
  2015, pp.~956--1006.

\bibitem[HWK13]{huang2013entropy}
Haiping Huang, KY~Michael Wong, and Yoshiyuki Kabashima, \emph{Entropy
  landscape of solutions in the binary perceptron problem}, Journal of Physics
  A: Mathematical and Theoretical \textbf{46} (2013), no.~37, 375002.

\bibitem[Jer92]{jerrum1992large}
Mark Jerrum, \emph{Large cliques elude the metropolis process}, Random
  Structures \& Algorithms \textbf{3} (1992), no.~4, 347--359.

\bibitem[JH60]{joseph1960number}
Roger~David Joseph and Louise Hay, \emph{The number of orthants in n-space
  intersected by an s-dimensional subspace}, Tech. report, CORNELL AERONAUTICAL
  LAB INC BUFFALO NY, 1960.

\bibitem[Kar76]{karp1976probabilistic}
Richard~M. Karp, \emph{The probabilistic analysis of some combinatorial search
  algorithms.}

\bibitem[Kea98]{kearns1998efficient}
Michael Kearns, \emph{Efficient noise-tolerant learning from statistical
  queries}, Journal of the ACM (JACM) \textbf{45} (1998), no.~6, 983--1006.

\bibitem[KM89]{krauth1989storage}
Werner Krauth and Marc M{\'e}zard, \emph{Storage capacity of memory networks
  with binary couplings}, Journal de Physique \textbf{50} (1989), no.~20,
  3057--3066.

\bibitem[KR98]{kim1998covering}
Jeong~Han Kim and James~R Roche, \emph{Covering cubes by random half cubes,
  with applications to binary neural networks}, Journal of Computer and System
  Sciences \textbf{56} (1998), no.~2, 223--252.

\bibitem[KWB19]{kunisky2019notes}
Dmitriy Kunisky, Alexander~S Wein, and Afonso~S Bandeira, \emph{Notes on
  computational hardness of hypothesis testing: Predictions using the
  low-degree likelihood ratio}, arXiv preprint arXiv:1907.11636 (2019).

\bibitem[LKZ15a]{lesieur2015mmse}
Thibault Lesieur, Florent Krzakala, and Lenka Zdeborov{\'a}, \emph{Mmse of
  probabilistic low-rank matrix estimation: Universality with respect to the
  output channel}, 2015 53rd Annual Allerton Conference on Communication,
  Control, and Computing (Allerton), IEEE, 2015, pp.~680--687.

\bibitem[LKZ15b]{lesieur2015phase}
\bysame, \emph{Phase transitions in sparse pca}, 2015 IEEE International
  Symposium on Information Theory (ISIT), IEEE, 2015, pp.~1635--1639.

\bibitem[LM15]{lovett2015constructive}
Shachar Lovett and Raghu Meka, \emph{Constructive discrepancy minimization by
  walking on the edges}, SIAM Journal on Computing \textbf{44} (2015), no.~5,
  1573--1582.

\bibitem[LRR17]{levy2017deterministic}
Avi Levy, Harishchandra Ramadas, and Thomas Rothvoss, \emph{Deterministic
  discrepancy minimization via the multiplicative weight update method},
  International Conference on Integer Programming and Combinatorial
  Optimization, Springer, 2017, pp.~380--391.

\bibitem[LSS21]{liu2021gaussian}
Yang~P Liu, Ashwin Sah, and Mehtaab Sawhney, \emph{A {G}aussian fixed point
  random walk}, arXiv preprint arXiv:2104.07009 (2021).

\bibitem[Mat99]{matousek1999geometric}
Jiri Matousek, \emph{Geometric discrepancy: An illustrated guide}, vol.~18,
  Springer Science \& Business Media, 1999.

\bibitem[MMZ05]{mezard2005clustering}
Marc M{\'e}zard, Thierry Mora, and Riccardo Zecchina, \emph{Clustering of
  solutions in the random satisfiability problem}, Physical Review Letters
  \textbf{94} (2005), no.~19, 197205.

\bibitem[Mul]{Multi}
\emph{\texttt{mvncdf} multivariate normal cumulative distribution function},
  \url{http://web.archive.org/https://www.mathworks.com/help/stats/mvncdf.html},
  Accessed: 2021-07-03.

\bibitem[Pot18]{potukuchi2018discrepancy}
Aditya Potukuchi, \emph{Discrepancy in random hypergraph models}, arXiv
  preprint arXiv:1811.01491 (2018).

\bibitem[Pot20]{potukuchi2020}
Aditya Potukuchi, \emph{{A Spectral Bound on Hypergraph Discrepancy}}, 47th
  International Colloquium on Automata, Languages, and Programming (ICALP 2020)
  (Dagstuhl, Germany) (Artur Czumaj, Anuj Dawar, and Emanuela Merelli, eds.),
  Leibniz International Proceedings in Informatics (LIPIcs), vol. 168, Schloss
  Dagstuhl--Leibniz-Zentrum f{\"u}r Informatik, 2020, pp.~93:1--93:14.

\bibitem[PX21]{perkins2021frozen}
Will Perkins and Changji Xu, \emph{Frozen 1-{RSB} structure of the symmetric
  {I}sing perceptron}, Proceedings of the 53rd Annual ACM SIGACT Symposium on
  Theory of Computing, 2021, pp.~1579--1588.

\bibitem[Rot17]{rothvoss2017constructive}
Thomas Rothvoss, \emph{Constructive discrepancy minimization for convex sets},
  SIAM Journal on Computing \textbf{46} (2017), no.~1, 224--234.

\bibitem[RSS18]{raghavendra2018high}
Prasad Raghavendra, Tselil Schramm, and David Steurer, \emph{High-dimensional
  estimation via sum-of-squares proofs}, arXiv preprint arXiv:1807.11419
  \textbf{6} (2018).

\bibitem[RV17]{rahman2017local}
Mustazee Rahman and Balint Virag, \emph{Local algorithms for independent sets
  are half-optimal}, The Annals of Probability \textbf{45} (2017), no.~3,
  1543--1577.

\bibitem[Sid68]{sidak1968multivariate}
Zbynek Sid{\'a}k, \emph{On multivariate normal probabilities of rectangles:
  their dependence on correlations}, The Annals of Mathematical Statistics
  \textbf{39} (1968), no.~5, 1425--1434.

\bibitem[Spe85]{spencer1985six}
Joel Spencer, \emph{Six standard deviations suffice}, Transactions of the
  American mathematical society \textbf{289} (1985), no.~2, 679--706.

\bibitem[ST03]{shcherbina2003rigorous}
Mariya Shcherbina and Brunello Tirozzi, \emph{Rigorous solution of the
  {G}ardner problem}, Communications in mathematical physics \textbf{234}
  (2003), no.~3, 383--422.

\bibitem[Sto13]{stojnic2013another}
Mihailo Stojnic, \emph{Another look at the {G}ardner problem}, arXiv preprint
  arXiv:1306.3979 (2013).

\bibitem[Tal99]{talagrand1999intersecting}
Michel Talagrand, \emph{Intersecting random half cubes}, Random Structures \&
  Algorithms \textbf{15} (1999), no.~3-4, 436--449.

\bibitem[Tal11]{talagrand2011mean}
\bysame, \emph{Mean field models for spin glasses: Advanced replica-symmetry
  and low temperature}, Springer, 2011.

\bibitem[TMR20]{turner2020balancing}
Paxton Turner, Raghu Meka, and Philippe Rigollet, \emph{Balancing {G}aussian
  vectors in high dimension}, Conference on Learning Theory, PMLR, 2020,
  pp.~3455--3486.

\bibitem[Ver10]{vershynin2010introduction}
Roman Vershynin, \emph{Introduction to the non-asymptotic analysis of random
  matrices}, arXiv preprint arXiv:1011.3027 (2010).

\bibitem[Wei20]{wein2020optimal}
Alexander~S Wein, \emph{Optimal low-degree hardness of maximum independent
  set}, arXiv preprint arXiv:2010.06563 (2020).

\bibitem[Wen62]{wendel1962problem}
James~G Wendel, \emph{A problem in geometric probability}, Mathematica
  Scandinavica \textbf{11} (1962), no.~1, 109--111.

\bibitem[Wil91]{williams1991probability}
David Williams, \emph{Probability with martingales}, Cambridge university
  press, 1991.

\bibitem[Win61]{winder1961single}
Robert~O Winder, \emph{Single stage threshold logic}, 2nd Annual Symposium on
  Switching Circuit Theory and Logical Design (SWCT 1961), IEEE, 1961,
  pp.~321--332.

\bibitem[Xu19]{xu2019sharp}
Changji Xu, \emph{Sharp threshold for the {I}sing perceptron model}, arXiv
  preprint arXiv:1905.05978 (2019).

\bibitem[ZK16]{zdeborova2016statistical}
Lenka Zdeborov{\'a} and Florent Krzakala, \emph{Statistical physics of
  inference: Thresholds and algorithms}, Advances in Physics \textbf{65}
  (2016), no.~5, 453--552.

\end{thebibliography}
\begin{appendices}
\section{MATLAB Code for Verifying Lemma~\ref{asm:negativity}}\label{appendix:matlab}

We verify Lemma~\ref{asm:negativity} numerically using the following MATLAB code. 

Our experiments demonstrate that the functions $f_2(\beta,\alpha),f_3(\beta,\alpha)$ appear to be minimized when $\beta$ is close to one. For this reason, we restrict  our attention to $\beta\in[0.9,0.999]$ and generate  $\beta=0.9:{\rm sp}:0.999$ with ${\rm sp}=10^{-3}$. We take $K=1$ as in the rest of the paper, and set $\alpha=1.667$. In order the compute the probability term, we do not resort to any Monte Carlo simulations. Instead, we employ MATLAB's built-in \texttt{mvncdf} function to compute the associated ``box" probability (for dimensions 2 and 3). (The function, \texttt{mvncdf}, computes rectangular probabilities for multivariate Gaussian distribution using numerical integration, see~\cite{Multi} for a more elaborate description.) In particular, the only potential source of error is the error encountered at the numerical integration step. A feature of the \texttt{mvncdf} function is that the error guarantee in probability calculation is available. In particular, an inspection of our plots reveals that $f_3(\beta,1.677)$ is minimized for $\beta\approx 0.978$; and for this choice of  $\beta$, the  probability term is approximately $0.6205$ whereas the error estimate is of order $10^{-8}$.
\begin{lstlisting}
close all, clear all;
sp  = 1e-3; %spacing
beta = 0:sp:0.999;
L = length(beta);
K = 1;
alpha = 1.67;
%\varphi_count
phi_count_2 = 1+binent((1-beta)./2);
phi_count_3 = 1+(1-beta)./2 +binent((1-beta)./2)+((1+beta)./2).*binent((1-beta)./(2.*(1+beta)));
%probability term
mu_2 = zeros(1,2);
mu_3 = zeros(1,3); %mean
box_low_2 = (-K)*ones(1,2);
box_low_3 = (-K)*ones(1,3); %lower limits -K for probability box
box_high_3 = K*ones(1,3); %upper limits K for probability box
box_high_2 = K*ones(1,2);
f_2 = zeros(1,L);
f_3 = zeros(1,L);
phi_probs_2 = zeros(1,L);
phi_probs_3 = zeros(1,L);
for i=1:L
    phi_probs_2(i) = mvncdf(box_low_2,box_high_2,mu_2,(1-beta(i))*eye(2) + beta(i)*ones(2,2));
    phi_probs_3(i) = mvncdf(box_low_3,box_high_3,mu_3,(1-beta(i))*eye(3) + beta(i)*ones(3,3)); %evaluate probability
    f_2(i)   = phi_count_2(i) + alpha*log2(phi_probs_2(i));
    f_3(i)   = phi_count_3(i) + alpha*log2(phi_probs_3(i)); %construct f_3
end
figure
title('3-OGP')
plot(beta,f_3),
hold on
plot(beta,f_2,'g')
ylabel('$f_3(\beta),f_2(\beta)$','Interpreter','latex');
xlabel('$\beta$','Interpreter','latex');
legend('f_3','f_2','Zero')
RefLine = refline([0,0]);
RefLine.Color = 'r';
disp(['The minima of f_3 is ',num2str(min(f_3)),'.'])

figure
title('2-OGP')
plot(beta,f_2),
ylabel('$f_2(\beta)$','Interpreter','latex');
xlabel('$\beta$','Interpreter','latex');
RefLine = refline([0,0]);
RefLine.Color = 'r';
disp(['The minima of f_2 is ',num2str(min(f_2)),'.'])

function ent = binent(x)
    ent = -x .* log2(x)-(1-x) .* log2(1-x);
end
\end{lstlisting}
\end{appendices}
\end{document}